\documentclass[final,leqno]{siamltex}

\providecommand{\qedhere}{} % for compatibility with SIAM macro

\usepackage{epsfig}
\usepackage{mysetup,xspace}
\usepackage{amsmath}
\usepackage{amstext}
\usepackage{algorithm,algorithmic}

\providecommand{\Appendix}{}
\renewcommand{\Appendix}[2][?]{%
	\refstepcounter{section}%
	\vspace{\parskip}%
	{\flushright\large\bfseries\appendixname\ \thesection: #1}%
	\vspace{\baselineskip}%
}
\renewcommand{\appendix}{%
	\renewcommand{\section}{\secdef\Appendix\Appendix}%
	\renewcommand{\thesection}{\Alph{section}}%
	\setcounter{section}{0}%
}
\usepackage[pdfpagelabels]{hyperref}

\renewcommand{\xhdr}[1]{\subsection{#1}}

\usepackage[suppress]{color-edits}
\addauthor{as}{red}
\addauthor{dk}{blue}
\addauthor{ia}{green}
\addauthor{sc}{magenta}

\newcommand{\mathtext}[1]{\text{#1}}

\newcommand{\indicator}[1]{\mathbf{1}_{\{ #1 \}}}

\newcommand{\E}{\mathbb{E}}             % expectation
\providecommand{\half}{\ensuremath{\frac{1}{2}}\xspace}

\providecommand{\SET}[1]{\ensuremath{\{ #1 \}}\xspace}
\providecommand{\Set}[2]{\ensuremath{\SET{#1 \mid #2}}\xspace}
\providecommand{\PROB}{\ensuremath{{\rm Prob}}\xspace}
\providecommand{\Prob}[2][]{\ensuremath{%
\ifthenelse{\equal{#1}{}}{\PROB\left[\,#2\,\right]}{\PROB_{#1}\left[\,#2\,\right]}}\xspace}
\providecommand{\ProbC}[3][]{\Prob[#1]{#2\;|\;#3}}
\providecommand{\Expect}[2][]{\ensuremath{%
\ifthenelse{\equal{#1}{}}{\E}{\E_{#1}}%
\left[#2\right]}\xspace}
\providecommand{\ExpectC}[3][]{\Expect[#1]{#2\;|\;#3}}
\providecommand{\VAR}{\ensuremath{{\rm Var}}\xspace}
\providecommand{\Var}[1]{\ensuremath{\VAR[#1]}\xspace}

\providecommand{\Event}[2][]{\ensuremath{\ifthenelse{\equal{#1}{}}{%
{\cal #2}}{{\cal #2}_{{#1}}}}\xspace}

\def\QED{{\phantom{x}} \hfill \ensuremath{\rule{1.3ex}{1.3ex}}}

\providecommand{\Kth}[1]{\ensuremath{{#1}^{\rm th}}}

\newcommand{\D}[1][]{\ensuremath{%
\ifthenelse{\equal{#1}{}}{\mathcal{D}}{\mathcal{D}_{#1}}}\xspace}  % distance function
\newcommand{\DP}[1][]{\ensuremath{%
\ifthenelse{\equal{#1}{}}{\D'}{\D'_{#1}}}\xspace}  % distance function
\newcommand{\Dnorm}[1][]{\ensuremath{%
\ifthenelse{\equal{#1}{}}{\mathcal{N}}{\mathcal{N}_{#1}}}\xspace} % normalized distance
\newcommand{\DnormS}[1][]{\ensuremath{\Dnorm[#1]^*}\xspace}
\newcommand{\DnormP}[1][]{\ensuremath{\Dnorm[#1]'}\xspace}
\newcommand{\Dsp}[1][]{\ensuremath{%
\ifthenelse{\equal{#1}{}}{\D^{\mathtt{sp}}}{\D^{\mathtt{sp}}_{#1}}}\xspace} % shortest-paths metric
\newcommand{\Dspnorm}[1][]{\ensuremath{%
\ifthenelse{\equal{#1}{}}{\Dnorm^{\mathtt{sp}}}{\Dnorm^{\mathtt{sp}}_{#1}}}\xspace}
\newcommand{\Ball}[3][]{\ensuremath{%
\ifthenelse{\equal{#1}{}}{B(#2,#3)}{B_{#1}(#2,#3)}}\xspace}
\newcommand{\NBall}[2][]{\ensuremath{%
\ifthenelse{\equal{#1}{}}{\tilde{B}_{#2}^{*}}{\tilde{B}_{#2}^{#1}}}\xspace}
\newcommand{\NDBall}[1]{\ensuremath{\tilde{B}_{#1}}\xspace}
\newcommand{\NRBall}[3][]{\ensuremath{%
\ifthenelse{\equal{#1}{}}{\tilde{B}_{#2}(#3)}{\tilde{B}_{#2}(#3; #1)}}\xspace}

\newcommand{\SW}{{\mathtext{sg}}}
\newcommand{\ESW}[1][]{\ensuremath{%
\ifthenelse{\equal{#1}{}}{E_\SW}{E_\SW^{(#1)}}}\xspace}
\newcommand{\ESWP}[1][]{\ensuremath{%
\ifthenelse{\equal{#1}{}}{\hat{E}_\SW}{\hat{E}_\SW^{(#1)}}}\xspace}
\newcommand{\kSW}[1][]{\ensuremath{%
\ifthenelse{\equal{#1}{}}{k_\SW}{k_\SW^{(#1)}}}\xspace} % expected degree
\newcommand{\CSW}[1][]{\ensuremath{%
\ifthenelse{\equal{#1}{}}{C_\SW}{C_\SW^{(#1)}}}\xspace} %
\newcommand{\fSW}[1][]{\ensuremath{%
\ifthenelse{\equal{#1}{}}{f_\SW}{f_\SW^{(#1)}}}\xspace} % f for the SW

\newcommand{\Gr}[1][]{\ensuremath{\ifthenelse{\equal{#1}{}}{\mathcal{G}}{\mathcal{G}_{#1}}}\xspace} % distance-based random graph
\newcommand{\Eloc}{\ensuremath{E_{\mathtext{loc}}}\xspace}    % local structure
\newcommand{\ElocP}{\ensuremath{E'_{\mathtext{loc}}}\xspace}
\newcommand{\Ecur}{\ensuremath{E_\mathtext{cur}}\xspace}
\newcommand{\Kloc}{\ensuremath{k_{\mathtext{loc}}}\xspace}    % degree of \Eloc

\newcommand{\NeighB}[2]{\ensuremath{\Gamma(#1,#2)}\xspace} % neighborhood

\newcommand{\prunedR}{\ensuremath{r_{\mathtext{pru}}}\xspace}
\newcommand{\prunedE}{\ensuremath{E_{\mathtext{pru}}}\xspace}    % pruned graph
\newcommand{\localR}{\ensuremath{r_{\mathtext{loc}}}\xspace}    % local radius
\newcommand{\constDR}[1][]{\ensuremath{%
\ifthenelse{\equal{#1}{}}{r_{\mathtt{EDP}}}{r_{\mathtt{EDP}}(#1)}}\xspace} % radius for constant degree

\newcommand{\amoebaSUB}{\mathtext{amb}}
\newcommand{\amoebaE}[1][]{\ensuremath{%
\ifthenelse{\equal{#1}{}}{E_\amoebaSUB}{E_\amoebaSUB^{(#1)}}}\xspace}  % amoeba graph
\newcommand{\amoebaM}{\ensuremath{M_\amoebaSUB}\xspace}  % threshold #edges
\newcommand{\amoebaN}{\ensuremath{N_\amoebaSUB}\xspace}  % threshold #nodes
\newcommand{\amoebaR}{\ensuremath{r_\amoebaSUB}\xspace}  % threshold radius
\newcommand{\amoebaC}{\ensuremath{\gamma_\amoebaSUB}\xspace}  % constant for Amoeba radius
\newcommand{\amoebaTest}{Amoeba Test\xspace}

\newcommand{\EDPtest}[2]{\ensuremath{T}\xspace}
\newcommand{\EDPfix}[2]{\ensuremath{T^*}\xspace}
\newcommand{\ItAlg}[2]{$(#1,#2)$-EDP Pruning Algorithm\xspace}

\newcommand{\Mtwo}{\ensuremath{M_{\Lambda}}\xspace}         % #two-hop paths

\newcommand{\ShrinkR}[1]{\ensuremath{\hat{r}_{#1}}\xspace}
\newcommand{\ShrinkT}[1]{\ensuremath{r_{#1}}\xspace}
\newcommand{\AddDist}[1]{\ensuremath{a(#1)}\xspace}
\newcommand{\NumEdges}{\ensuremath{\mathtt{\#edges}}}
\newcommand{\NEdges}[3][]{\ensuremath{%
\ifthenelse{\equal{#1}{}}{\tilde{M}_{#2,#3}}{\tilde{M}^{(#1)}_{#2,#3}}
}\xspace}

\newcommand{\PRUNEDG}[1][]{\ensuremath{%
\ifthenelse{\equal{#1}{}}{T}{T_{#1}}
}\xspace}
\newcommand{\PrunedG}[2][]{\ensuremath{\PRUNEDG[#1](#2)}\xspace}

\newcommand{\EE}{\mathcal{E}}
\newcommand{\NET}{\ensuremath{\mathcal{N}}\xspace}
\newcommand{\PL}{\ensuremath{h}\xspace}
\newcommand{\PLSET}{\ensuremath{H}\xspace}
\newcommand{\PN}{\ensuremath{b}\xspace}

\newcommand{\Binomial}[2]{\mathtt{Bin}_{#1,#2}}
\newcommand{\BDF}[3]{\ensuremath{\Prob{\Binomial{#1}{#2} \geq #3}}\xspace}
\newcommand{\distortion}{\ensuremath{\Delta}\xspace}

\DeclareMathOperator{\polylog}{polylog}

\newcommand{\SimpleTest}{Two-Hop Test\xspace}
\newcommand{\Amoeba}{Amoeba\xspace}
\newcommand{\TwoBallTest}{Two-Ball Algorithm\xspace}
\newcommand{\extTwoBall}{Extended Two-Ball Algorithm\xspace}
\newcommand{\recTwoBall}{Recursive Two-Ball Algorithm\xspace}

\newcommand{\NUMCAT}{\ensuremath{K}\xspace} % #categories
\newcommand{\DIM}[1][]{\ensuremath{%
\ifthenelse{\equal{#1}{}}{d}{d_{#1}}}\xspace}   % # dimensions
\newcommand{\DIMCONST}{\ensuremath{c_{\DIM}}\xspace}

\newcommand{\densityK}{\ensuremath{C_{\mathtext{UD}}}\xspace}
\newcommand{\CPD}{\ensuremath{C_{\mathtext{PD}}}\xspace}
\newcommand{\CONTR}{\ensuremath{\sigma}\xspace}
\newcommand{\EXPAN}{\ensuremath{\delta}\xspace}

\newcommand{\LCD}{Local Category-Disjointness\xspace}

\newcommand{\GCDscale}{\ensuremath{R}\xspace}
\newcommand{\globalCD}[1][]{%
\ifthenelse{\equal{#1}{}}{Scale-$\infty$ Category-Disjointness}{%
Scale-$#1$ Category-Disjointness}\xspace}
\newcommand{\pairs}{\mathtt{\#pairs}}
\newcommand{\diam}{\text{diam}}

\newcommand{\singleSG}{single-category social graph\xspace}
\newcommand{\multiSG}{multi-category social graph\xspace}

\begin{document}

%\pagenumbering{Alph}
%\begin{titlepage}

\title{Low-distortion Inference of Latent Similarities\\
from a Multiplex Social Network%
\footnote{An extended abstract of this work has appeared in \emph{ACM-SIAM Conference on Discrete Algorithms (SODA)}, 2013. }}

%\date{February 2012\\ This revision: April 2012}

\author{Ittai Abraham
\thanks{Microsoft Research, Mountain View CA, USA.
Email: {\tt ittaia@microsoft.com}.}
\and
Shiri Chechik
\thanks{Microsoft Research, Mountain View CA, USA.
Email: {\tt schechik@microsoft.com}. \newline
Research done in part while S. Chechik was a graduate student at
Dept.~of Computer Science, Weizmann Institute of Science, Rehovot, Israel.}
\and David Kempe
\thanks{Dept.~of Computer Science,
University of Southern California, Los Angeles CA, USA.
Email: {\tt dkempe@usc.edu}.
Research done in part while visiting Microsoft Research.
Work supported in part by grant NSF CAREER 0545855 and an Okawa
Foundation Grant.}
\and Aleksandrs Slivkins
\thanks{Microsoft Research, New York NY, USA.
Email: {\tt slivkins@microsoft.com}.}
}

\maketitle

\begin{abstract}
Much of social network analysis is --- implicitly or explicitly ---
predicated on the assumption that individuals tend to be more
similar to their friends than to strangers. Thus, an observed
social network provides a noisy signal about the latent underlying
``social space:'' the way in which individuals are similar or
dissimilar. Many research questions frequently addressed via
social network analysis are in reality questions about this social
space, raising the question of inverting the process: Given a
social network, how accurately can we reconstruct the social
structure of similarities and dissimilarities?

We begin to address this problem formally. Observed social
networks are usually multiplex, in the sense that they reflect
(dis)similarities in several different ``categories,'' such as
geographical proximity, kinship, or similarity of
professions/hobbies. We assume that each such category is
characterized by a latent metric capturing (dis)similarities in
this category. Each category gives rise to a separate social
network: a random graph parameterized by this metric. For a
concrete model, we consider Kleinberg's small world model and some
variations thereof. The observed social network is the unlabeled
union of these graphs, i.e., the presence or absence of edges can
be observed, but not their origins. Our main result is an
efficient algorithm which reconstructs each metric with provably low
distortion.

\end{abstract}



\pagestyle{myheadings}
\thispagestyle{plain}
\markboth{Abraham, Chechik, Kempe and Slivkins}{Inferring Latent Similarities from a Social Network}

\section{Introduction} \label{Intro}
Much of social network analysis is, implicitly or explicitly,
predicated on the assumption that people tend
to be more similar to their friends than to strangers.
While many tasks --- such as analyzing power and centrality,
trading and exchange, or understanding and influencing the diffusion
of viruses or information --- rely crucially on the precise network
structure, many others --- such as link prediction, identification of
communities, or marketing to friends of past buyers --- use network
structure as a noisy signal about an underlying social similarity
space.
To illustrate this insight differently, consider altering
a social network data set by removing links between ``dissimilar''
pairs of individuals, and inserting instead links between ``similar''
(but previously unconnected) pairs. If this change makes the analysis
task easier, rather than impossible, then the analysis task is really
about the ``social structure'' --- the latent similarities and
dissimilarities between individuals ---  rather than about the actual
network structure.

Given the abundance of important problems naturally phrased in terms
of social structure (discussed in more detail below), it is a
natural goal to explicitly reconstruct social structures from a given
social network.
Knowing the social structure may also be of independent interest, as
it sheds light on the forces governing social link formation.
%\emph{In this paper, we present and analyze natural algorithms for the task
%of social structure reconstruction with provably low distortion.}

The task of inferring social structure in this sense is made
non-trivial by the following two obstacles.
First, despite a general tendency for friends to be more similar
than strangers, many friends are still sufficiently different from
each other to look essentially random.
Second, and perhaps more fundamentally, social networks are \emph{multiplex}
\cite{fienberg:meyer:wasserman,minor:multiplexity,szell:lambiotte:thurner}:
they tend to be the union of multiple often independent relations
among the same actors.
For instance, friendships could result from physical proximity, similarity of
occupation, kinship, similarities of hobbies, etc.
If individuals are very similar in even one such attribute, they are more
likely to be connected.

The main contribution of this paper is a near-linear time algorithm
for reconstructing the latent social structure with provably low
distortion.
The model explicitly produces a union of graphs, one for each
category, and an important feature of the algorithm is that it
separates the different graphs from each other.
%, thus contributing to the goal of automatically inferring different
%friendship communities of an individual.
We also provide two extensions which, respectively, further improve
the distortion, and partially address the issue of data scarcity
(i.e., very small node degrees).
%social circles.
%Our results make significant progress on a strong generalization of
%an open question raised in \cite{kleinberg:decentralized-search}.
The algorithms in this paper are based on, and significant extensions of,
a natural idea that is widely used in practice:
nodes are likely to be close if they share many common neighbors.

%\subsection{An overview of the model}
\xhdr{An overview of the model}
We posit a latent space model (described in detail in
Section \ref{Prelims}) for the generation of social networks
akin to models widely used in the mathematical sociology, statistics,
and computer science communities
\cite{clauset:moore:newman:link-prediction,fraigniaud:lebhar:lotker,handcock:raftery:tantrum,hoff:raftery:handcock,kermarrec:leroy:tredan,small-world,krivitsky:handcock:raftery:hoff,raftery:niu:hoff:yeung,sarkar:chakrabarti:moore,sarkar:moore:dynamic,schweinberger:snijders:settings}
(see also the survey \cite[pages 15--21]{snijders:statistical-models-survey}).

The model is based on two widely accepted tenets about social networks
(e.g., \cite{blau:inequality,mcfarland:brown}).
First, people are more likely to have ties with those who are
similar to them, but also have many ties to others who are
dissimilar.\footnote{The model is agnostic about whether this
similarity is caused more by \emph{homophily}
\cite{lazarsfeld:merton,mcpherson:smith-lovin:cook} (the tendency to
form ties with those who are similar)
or by \emph{social influence}
\cite{marsden:friedkin:network-studies,rogers} (the tendency to
\emph{become} similar to one's associates).}
Second, multiple social dimensions (such as geography, occupation,
kinship, hobbies, etc.) can independently lead to
interactions and the formation of ties.

We call the social dimensions along which people can be (dis)similar
\emph{(social) categories}, to avoid confusion with the geometric dimensions of
individual metric spaces.
Each category is
given by a metric space $\D[i], i=1, \ldots, \NUMCAT$;
together, the $\D[i]$ define the \emph{social distances} between the
individuals.
Each of the $n$ individuals occupies a point in each of the categories.
For concreteness, and in accordance with much of the preceding literature,
we assume that each category is a Euclidean space of known dimensionality
\cite{handcock:raftery:tantrum,hoff:raftery:handcock,small-world,krivitsky:handcock:raftery:hoff,raftery:niu:hoff:yeung,sarkar:chakrabarti:moore},
and that the density of the points corresponding to
individuals is nearly uniform
\cite{hoff:raftery:handcock,small-world,sarkar:chakrabarti:moore}.
Furthermore, we assume that the categories have small local correlation.
The ``local correlation'' of two categories is the maximal
overlap between any two small balls in those categories
(see Equation (\ref{eq:cat-disj}) in Section~\ref{Prelims}).

Each category independently gives rise to a social network \Gr[i],
modeled as a random graph whose edge distribution is parameterized by
the corresponding metric space $\D[i]$.
Specifically, we use a slight variation of Kleinberg's small-world
model \cite{small-world}, in which edge probabilities decrease
  polynomially in $\D[i](u,v)$.
For our purposes, the key feature of the model is that the
probability of shorter links is much higher, but long-range links also
appear with a significant probability; this captures the first tenet.
The algorithm observes the \emph{union} $\Gr = \bigcup_i \Gr[i]$
of the individual networks \Gr[i] (on the same node set), but does not
learn which \emph{particular} network(s) \Gr[i] an edge belonged to.
This captures the second tenet; only the existence, but not the
social ``origins,'' of ties can be observed.\footnote{%
Our model does not include any information such as demographics,
location, wall posts, or communications which would frequently be
available to social networking sites %such as Facebook
\cite{backstrom:sun:marlow}.
Our goal here is to understand at a fundamental level how much
information on social structures can be inferred algorithmically from
the observed social network alone.}
\emph{The algorithm's goal is to use \Gr to reconstruct the individual metrics
\D[i] with small distortion, with high probability (over the
random network generation process).}

Importantly, social similarity spaces in general tend not to be
metrics (see, e.g., \cite{butts:predictability}), in the sense that
the triangle inequality fails to hold.
The main reason is the presence of multiple social
categories. For example, one's co-worker and one's relative could be
very dissimilar to one another, even though the individual is similar
to both.
The inclusion of a union or minimum in the model is crucial to capture
this.

%\subsection{Algorithms and Results}
\xhdr{Algorithms and results}
Our main contribution is a near-linear time algorithm,
called the \emph{\Amoeba algorithm},
which infers all individual categories with provably low distortion,
with high probability.
The following theorem captures the result slightly informally.

\begin{theorem}[informal] \label{thm:intro-main}
If the \NUMCAT metric spaces \D[i] are locally sufficiently different,
and the average node degrees are at least $\Omega(\NUMCAT^3 \log^2 n)$,
then with high probability, the \Amoeba algorithm, in near-linear time,
reconstructs metrics \DP[i] such that \DP[i] approximates \D[i] with
constant multiplicative distortion (and at most polylogarithmic additive error).
\end{theorem}

That this approximate reconstruction should be possible at all
--- regardless of the running time --- is somewhat surprising.
One might think a priori that after combining two social networks,
there would simply be no way to tease them apart.

In other words, a priori, the challenge appears to be
information-theoretical (does the network contain enough information
for distance reconstruction with any provable guarantees?)
as much as computational.
We also remark that even the single-category version was raised
by Kleinberg \cite{kleinberg:decentralized-search} as an open
question; we answer the reconstruction question in the positive even
for multiple categories.

%This result makes significant progress on a strong generalization
%of an open question from \cite{kleinberg:decentralized-search},
%namely, whether the process of small-world generation can be inverted
%in the sense that the metric space can be reconstructed from an
%observed small-world graph. We show that --- under some restrictions
%--- this is even possible for the \emph{union} of multiple small-world
%graphs.

The \Amoeba algorithm, we well as  all other
algorithms in this paper, is broadly based on a heuristic
widely used in practice (e.g., in Facebook, or see \cite{adamic:adar:friends,liben-nowell:kleinberg:link-prediction,sarkar:chakrabarti:moore,schwartz:wood}):
edges $(u,v)$ are more likely to be between friends in a category if
they are ``supported'' by many common neighbors of $u$ and $v$ in that
category.
However, to deal with multiple categories, low node degrees, or to
sharpen the distance estimates, the basic idea of counting common
neighbors needs to be extended significantly.

The \Amoeba algorithm, presented and analyzed in detail in Section
\ref{sec:amoeba}, consists of two stages.
In a first stage, individual edges are pruned if they do not
have enough common neighbors, a direct implementation of the common neighbors
heuristic.\footnote{Sarkar et al.~\cite{sarkar:chakrabarti:moore} showed
that under a model similar to ours
(but using edge probabilities that decrease exponentially with distance),
counting common neighbors leads to an accurate distance estimate
for a single-category social network.}
In the second stage, which we call \emph{the \Amoeba stage},
basic estimates of the individual categories are constructed one by one.
Each iteration starts with a
%poly-logarithmically large
polylog-sized clique in the graph computed by the first stage,
which is then expanded one edge at a time:
an edge $(u,v)$ is added to a category only when enough of $u$'s
neighbors lie in a small ball around $v$ according to the current
estimate of the category.
The basic idea is that any sufficiently large clique must be
sufficiently close in one category. The clique then bootstraps further
iterations, in that a node $u$ with many edges to a small ball around
$v$ must itself be close to $v$.
While this intuition is straightforward, each iteration loses
accuracy, so it takes a delicate proof to show that
this refined version of the common neighbors heuristic guarantees
low distortion.

We improve the main result in the following two directions.
The first direction (Sections \ref{sec:additive-singleCat} and
\ref{sec:additive-mult}) focuses on improving the distortion using
long-range links, which are now treated as an additional data source
rather than an obstacle to be pruned. We improve the distortion from a
multiplicative constant to a factor $1+o(1)$,
using a post-processing phase (run after the Amoeba
algorithm) which we call \emph{\TwoBallTest.}
This is a variation of the common neighbors heuristic where instead of common
neighbors of two nodes $(u,v)$, the algorithm counts long-range links between two node sets.
The node sets are low-radius balls around $u$ and $v$ according to the initial
distance estimates.
%3-hop paths whose first and last hops are short according to the initial estimates.
This result requires a stronger notion of low correlation between categories.
Under a stronger uniform density conditions, the \TwoBallTest can be applied
recursively, yielding \emph{unit} distortion (with at
most polylogarithmic additive error).

Second (in Section \ref{sec:constDeg}), we deal with the issue of data
scarcity, which in our setting translates to low node degrees.
In the low (constant) node degree regime,
the common neighbors heuristic is uninformative,
and it instead becomes necessary to count disjoint constant-length
paths for a suitably chosen constant.
Combining the new initial pruning
phase with a subsequent \TwoBallTest requires a much more careful
analysis, which shows that all sufficiently long edges can be treated
as mutually independent given the pruned graph. We recover
(essentially) all our results for the single-category case; extending
the results to multiple categories remains a direction for
future work.

For both extensions, more detailed descriptions of challenges,
results, and high-level approaches are deferred to the introductory
portions of the corresponding sections.

Our algorithms are modular: a pre-processing step
(counting common neighbors, or the low-degree algorithm of Section
\ref{sec:constDeg}) prunes away very long edges.
The \Amoeba step separates different metrics and constructs initial
distance estimates (though we have not adapted the algorithm and
analysis to low node degrees).
Finally, the \TwoBallTest and its recursive version can be used to
further improve the distortion in individual categories.

% DK: I didn't like this way of phrasing it
%\ASedit{While our main contributions are for multiple social categories, the above results contain a storyline for the single-category case. This storyline consists of the pruning step (including the more difficult case of constant node degrees), followed by the \TwoBallTest to obtain $1+o(1)$ distortion or (under stronger assumptions) recursive \SimpleTest to obtain unit distortion with polylogarithmic additive error.}

%\subsection{Discussion of the Model}
\xhdr{Discussion of the model}
Our modeling goal is not to define a model of social networks
capturing all of their features; this would be a formidable/impossible
task for which there is much research but not much consensus.
Instead, we aim for generally accepted modeling choices which capture
in a clean way the main algorithmic challenges inherent in rigorous
distance reconstruction.
In particular, our main goal was to capture the two conceptual
obstacles to distance reconstruction: links between dissimilar
individuals, and multiple social categories.

\asedit{From the algorithmic point of view, we are looking for
  modeling assumptions that allow non-trivial provable guarantees for
  social distance reconstruction. A natural progression is to start
  with a basic model with the strongest assumptions, and then to relax
  them.
  For our work, a natural basic model is that each category
  is Kleinberg's small world on a rectangular grid, and the mapping of
  individuals to the locations on the grid is chosen as a random
  permutation, independently for each category.
  The reconstruction problem in this basic model is already
  difficult, and our main result (Theorem~\ref{thm:intro-main})
  does not get much easier to derive. Compared to the basic model, our
  actual model is relaxed in several directions: we consider point
  sets of near-uniform density rather than rectangular grids; we
  replace the global condition of ``random permutations'' with a much
  weaker condition of small ``local correlation'' between categories;
  and we allow each category to have different parameters such as
  Euclidean dimension and node degree. The investigation of what can
  and cannot be done with further relaxations of the model is a
  natural direction for future work.}

%Nevertheless,
Let us discuss some particular modeling choices in more detail.

\begin{enumerate}
\newcommand{\fakeItem}[1]{\item}

\fakeItem{1}
In Kleinberg's small-world model
  \cite{small-world,kleinberg:navigation,small-world-nips,kleinberg:decentralized-search,fraigniaud:small-worlds-survey},
  a version of which we adapt as a generative model for individual
  categories, the probability for an edge between two nodes to exist decreases polynomially in the nodes' distance.
%  The fact that the
%  probability of an edge decreases in the social distance of
%  individuals is widely accepted and documented (e.g.,
%  \cite{blau:inequality,mcpherson:smith-lovin:cook}).
  Naturally, many other distributions lead to
  distance-based random graphs \cite{barthelemy:spatial-networks}.

  Much of the past work in the statistics community
  \cite{handcock:raftery:tantrum,hoff:raftery:handcock,krivitsky:handcock:raftery:hoff,raftery:niu:hoff:yeung,sarkar:chakrabarti:moore}
  assumed that the edge probabilities were logit-linear in the distance,
  i.e., that $\log(\tfrac{p}{1-p})$ is linear in $\D(u,v)$.
%  \ASdelete{so that the edge probabilities decrease as ... ; in the logit-linear model,
%    for larger distances, the probability of an edge existing decreases exponentially fast.}
  Since long-range links are thus exponentially unlikely
    ($p=\frac{e^{-\alpha \D(u,v)}}{1+e^{-\alpha \D(u,v)}}$),
    the reconstruction  task becomes much easier.
  More importantly, to the extent that precise distributions have been
  empirically tested, remarkable fits have been found
  \cite{adamic:adar:search,backstrom:sun:marlow,liben-nowell:novak:kumar:raghavan:tomkins}
  with Kleinberg's inverse polynomial
  distribution~\cite{small-world,small-world-nips}.\footnote{%
  However, links that appear long could plausibly be
  short in another metric; whether inverse
    polynomial distributions remain prevalent when multiple metrics
    are considered is an interesting --- although difficult ---
    direction for future empirical work.}
  Furthermore, our main constant-distortion result holds for a much
  more general class of distributions, including logit-linear
  distributions.

\fakeItem{2}
  The choice of Euclidean spaces with near-uniform density.
  Both choices (Euclidean and near-uniform) are ubiquitous in past
  work\footnote{In many respects, our kind of latent space models deteriorate
if node densities can be highly non-uniform \cite{fraigniaud:lebhar:lotker:threshold}.}
    \cite{fraigniaud:lebhar:lotker,handcock:raftery:tantrum,hoff:raftery:handcock,kermarrec:leroy:tredan,small-world,krivitsky:handcock:raftery:hoff,raftery:niu:hoff:yeung,sarkar:chakrabarti:moore},
  and are made mostly for technical convenience; they allow us to
  separate the conceptual difficulty of teasing apart different metrics
  and inferring distances with low distortion from the technical
  difficulty of dealing with arbitrary metric spaces.
  We believe that future work will achieve similar results for more
  general metric spaces or related structures, in particular, ultrametrics
  \cite{clauset:moore:newman:link-prediction,small-world-nips,schweinberger:snijders:settings},
  which are another popular choice of latent metric spaces.

\fakeItem{3}
The choice of a union or minimum to combine individual metrics.
  This choice is clearly a simplification of reality: individuals are
  more likely to form ties if they share similarities in multiple
  dimensions, e.g., they work in the same field \emph{and} live in the same
  town. Our model is supposed to capture in the cleanest way the
  difficulty of separating edges originating from different
  categories, and is certainly a better approximation to reality than
  widely used models treating the social structure as one metric space.

  Our model is closely related to (and a slight generalization of)
  a notion of social distance proposed by Watts, Dodds, and Newman
  \cite{watts:dodds:newman}, which treats the social distance as the
  minimum of distances in multiple metrics.
  To the extent that past work explicitly discussed models of
   multiple categories, it was also based on the minimum
\cite[pp.~337, 348]{handcock:raftery:tantrum}, \cite[p.~335]{schweinberger:snijders:settings}.
A generalization to more realistic models is a natural
direction for future work.

\fakeItem{4}
We capture a notion of ``independence'' between categories by
requiring that small balls in different categories have small overlap.
Even without restrictions on computational resources, some assumption
about ``independence'' is clearly necessary: if categories could be
extremely similar, then no low-distortion reconstruction seems possible.
It is an interesting direction for future work whether a few isolated
violations of the condition permit low-distortion reconstruction in
all but the affected areas of the metric spaces.

Our condition is significantly weaker than requiring probabilistic independence.
Several past papers (using a single metric space) assumed that nodes were
placed independently and uniformly at random over some space
\cite{hoff:raftery:handcock,sarkar:chakrabarti:moore}; such a model of
individual categories would imply our ``small intersection'' condition
with high probability.
In fact, we show in Section \ref{sec:permutations} that with high
probability, the ``small intersection'' condition holds even when
nodes are placed adversarially, and their names are permuted randomly.
We also remark that while in reality, we will frequently observe high
correlation between ``categories'' (such as work and hobbies), this
could be construed as a sign that the categories should be chosen
differently, in order to represent the latent traits that manifest
themselves in choices of both occupations and hobbies.
\end{enumerate}

%\subsection{Applications}
\xhdr{Applications}
Our work provides two natural reconstruction abilities: separating
edges by categories, and reconstructing individual categories with low
distortion. Both of them have multiple useful applications.

Important industrial applications for social network
information include improving ad placement (\emph{social advertising}),
web search results (\emph{social search}), and product recommendations.
These applications are of vital importance for some of the major
players on the Web.
A key commonality of all three applications is that
they use the behavior of friends (clicking, searching, purchasing) to
predict the behavior of an individual.
Yet, two recent
studies~\cite{goel:goldstein:birds,liu:tang:behavioral-targeting}
undertaking a quantitative evaluation of the predictive power of
social links for purchases and click behavior have found at best mixed
evidence.

This apparent conundrum is resolved by noticing that many
links are long-range, and short-range links may be short in an
irrelevant category for the prediction task.
Indeed, a recent data-driven study by Tang and Liu
\cite{tang:liu:latent-social-dimensions} has shown that social link-based
classifiers perform much better when edges are labeled with
categories in which they are short. We conjecture that such
classifiers would improve even further if instead of edges, the
actual \emph{social distance} between nodes were used.

The ability to separate social categories also enables the automatic
detection of circles of friends from different contexts in
  social networking sites. This automatic detection
has been cited as one of the main selling points of Google+, and is at
the heart of the startup Katango. In this sense, our work provides
some theoretical underpinnings for this fast-growing facet of the social networking
market. Separating edges by categories has the additional benefit that
one can identify when edges are short in more than one category,
which could enable the automatic detection of close friends
\cite{wellman:types,wellman:wortley:strokes}.

Another natural application is the discovery of ``social communities''
% of densely linked nodes in social networks
\cite{brandes:erlebach:network-analysis,fortunato:community-survey,fortunato:castellano:community-survey,danon:duch:diaz-guilera:arenas,schaeffer:graphs-clustering-survey}.
One might argue that the plethora of different network
  community detection
objectives and heuristics is largely a result of stating
the objectives and algorithms in terms of the graph structure, when
the goal is really to identify clusters in the metric spaces.
Since the social space is rarely explicitly modeled or related to the
network, the connection between the objective function and the actual
desired object is absent.
Explicitly reconstructing the social space would constitute the first
step toward a more sound community identification algorithm.
The presence of multiple categories in the model will naturally give
rise to overlapping communities as well.
Indeed, some of the work on reconstructing Euclidean spaces in the
statistics community \cite{handcock:raftery:tantrum,krivitsky:handcock:raftery:hoff}
is explicitly motivated by the desire to identify communities, and
builds community structure into a Bayesian prior.
%Furthermore, two recent manuscripts by
%Arora et al.~\cite{arora:ge:sachdeva:schoenebeck}
%and Balcan et al.~\cite{balcan:borgs:braverman:chayes:teng}
%aim to reconstruct overlapping community structure, by inferring a
%set-based latent structure from a social network.

Social distances can also be used to predict unobserved or potential
social links.
Link prediction has been studied in
\cite{adamic:adar:friends,clauset:moore:newman:link-prediction,liben-nowell:kleinberg:link-prediction,sarkar:chakrabarti:moore,schwartz:wood}.
Unobserved or potential links are most likely present between node
pairs at small distances; hence, once distances are known, missing
links can be predicted easily
\cite{clauset:moore:newman:link-prediction,sarkar:chakrabarti:moore}.

\OMIT{ %%%
%\subsection{Related Work} \label{Related-Work}
\xhdr{Related work}
Due to space constraints, an in-depth discussion of related work is
deferred to Section \ref{sec:related-work-long}.

A lot of recent work \cite{backstrom:sun:marlow,clauset:moore:newman:link-prediction,fraigniaud:lebhar:lotker,handcock:raftery:tantrum,hoff:raftery:handcock,kermarrec:leroy:tredan,krivitsky:handcock:raftery:hoff,raftery:niu:hoff:yeung,sarkar:moore:dynamic,schweinberger:snijders:settings}
uses Bayesian Models or Maximum Likelihood Estimation to reconstruct
metric spaces (mostly, but not exclusively, Euclidean).
These papers do not model multiple categories, and they do not come
with any guarantees on the quality of approximation
of the inferred metric; in addition, their inference
problems are often not tractable, and heuristics without guarantees even
on likelihood or probability are used.
The most notable exception is the work of Sarkar, Chakrabarti, and Moore
\cite{sarkar:chakrabarti:moore}, who are motivated by the goal of
explaining why simple heuristics for link prediction, such as counting
common neighbors, are successful
\cite{adamic:adar:friends,liben-nowell:kleinberg:link-prediction,schwartz:wood}.
As part of their analysis, they show that for a single category with
logit-linear edge probabilities, counting common neighbors gives
accurate distance estimates.

There are conceptual similarities between the present paper
and simultaneous independent work by Arora et al.~\cite{arora:ge:sachdeva:schoenebeck}
and Balcan et al.~\cite{balcan:borgs:braverman:chayes:teng}.
Their goal is to reconstruct overlapping community structure with
provable guarantees.
They posit latent set-based structures which can be interpreted as
0-1 metrics.
Interestingly, they also require a ``limited overlap'' condition, and
some of the algorithmic ideas used are similar.
However, the reconstructed objects are different, and there is no
analogue in their work to our post-processing steps and the algorithms
we design for dealing with low degrees.
} %%%%% 
\section{Related work} \label{sec:related-work-long}
Our work is related to work in a large number of communities:
latent space reconstruction in statistics and mathematical sociology,
community discovery, small-world networks, network localization, and
metric space embeddings. We discuss the different areas in their
separate sections.

%\subsection{Latent Space Reconstruction}
%\label{sec:related:latent-space}
\xhdr{Latent Space Reconstruction}
Several recent papers \cite{backstrom:sun:marlow,clauset:moore:newman:link-prediction,fraigniaud:lebhar:lotker,handcock:raftery:tantrum,hoff:raftery:handcock,kermarrec:leroy:tredan,krivitsky:handcock:raftery:hoff,raftery:niu:hoff:yeung,sarkar:chakrabarti:moore,sarkar:moore:dynamic,schweinberger:snijders:settings}
aim to reconstruct latent metrics from an observed social network.
The precise models differ across these papers: most assume Euclidean
spaces
\cite{backstrom:sun:marlow,fraigniaud:lebhar:lotker,handcock:raftery:tantrum,hoff:raftery:handcock,kermarrec:leroy:tredan,krivitsky:handcock:raftery:hoff,raftery:niu:hoff:yeung,sarkar:chakrabarti:moore,sarkar:moore:dynamic},
while a few consider ultrametrics to model hierarchical communities
\cite{clauset:moore:newman:link-prediction,schweinberger:snijders:settings}.
Among the papers considering Euclidean spaces, there are different
assumptions about link distributions: most assume a logit-linear model
\cite{handcock:raftery:tantrum,hoff:raftery:handcock,krivitsky:handcock:raftery:hoff,raftery:niu:hoff:yeung,sarkar:chakrabarti:moore},
while a few consider inverse polynomial ``small-world'' distributions
\cite{backstrom:sun:marlow,fraigniaud:lebhar:lotker,kermarrec:leroy:tredan}.\footnote{%
We remark that several recent studies
\cite{adamic:adar:search,backstrom:sun:marlow,liben-nowell:novak:kumar:raghavan:tomkins}
show that the frequency of friendships as a function of
(2-dimensional) geographic distance, when corrected for non-uniform
densities, appears to decrease as $\Theta(r^{-2})$. This gives some
tentative empirical evidence in favor of ``small-world'' distributions.}
There are many other modeling dimensions along which these papers (and
ours) differ, including: variance in node degrees, additional
information about nodes (such as locations of some nodes
\cite{backstrom:sun:marlow}), uniform or clustered priors for node
locations, whether algorithms are supposed to be centralized or
distributed \cite{kermarrec:leroy:tredan}, etc.\footnote{%
Much of the recent work in the mathematical sociology
community has focused on exponential random graph models, which in a
sense ``hard-wire'' desired distributions of certain features.
These models are generally of a very different nature from
latent-space models.
A recent paper by Butts \cite{butts:location-systems} combines
features of both location-based and exponential random graph models;
like the other papers listed above, it is not clear whether inference
of model parameters would be tractable, and whether it would lead to
any guarantees on distortion.}

Two main differences stand out between our work and the majority of
these papers (in addition to the more minor modeling differences).
First, we model multiple categories, which is extremely realistic, but
makes the model, algorithms, and analysis significantly more complex.
Second, the majority of the work cited above
\cite{backstrom:sun:marlow,clauset:moore:newman:link-prediction,handcock:raftery:tantrum,hoff:raftery:handcock,kermarrec:leroy:tredan,krivitsky:handcock:raftery:hoff,raftery:niu:hoff:yeung,schweinberger:snijders:settings}
estimates the underlying space either using Maximum Likelihood
Estimates (MLE), or by imposing a Bayesian Prior and maximizing the
probability of the chosen locations. Both appear to be very complex
problems, and indeed, all of the papers employ heuristics (based on
Gibbs Sampling, Metropolis-Hastings, Simulated Annealing, etc.)
without guarantees on the likelihood or probability of the solution
returned. More fundamentally, even if it were possible to obtain
the MLE or highest-probability solution, it is not clear that it would
come with any guarantees on the worst-case (or even average)
distortion; the objective function does not explicitly model
distortion, and in particular may be sacrificing the distortion of
some edges in order to optimize the more global objective.

Two notable exceptions to the MLE/Bayesian approach are the works of
Fraigniaud, Lebhar, and Lotker \cite{fraigniaud:lebhar:lotker} and
Sarkar, Chakrabarti, and Moore \cite{sarkar:chakrabarti:moore}.
Fraigniaud et al.~\cite{fraigniaud:lebhar:lotker} aim to reconstruct a
single-category small-world model in order to use the distance
estimates for greedy routing.
They propose a heuristic based on an MLE intuition;
interestingly, this heuristic leads to essentially counting common neighbors.
Their algorithm may retain a small number of long-range edges, and
hence does not come with provable guarantees on the distortion of the
reconstructed metric space.
They prove that this does not stand in the way of greedy routing:
despite the lack of distortion guarantees, the distances they
construct provably enable greedy routing along poly-logarithmic length paths.

Sarkar et al.~\cite{sarkar:chakrabarti:moore} begin from the goal of
explaining why simple heuristics for link prediction, such as counting
common neighbors, are successful.
They show that such heuristics can be understood as identifying close
pairs of nodes in a latent Euclidean space, and use this insight to
give provable guarantees on the performance of several heuristics for
link prediction.
(They also suggest additional heuristics).
In the process, they show how a metric space is implicitly
reconstructed by counting common neighbors.
There are a few key differences between their work and ours. First,
their distributions are logit-linear, implying that long-range edges
are extremely unlikely. The reconstruction task is still non-trivial,
but they do not have to deal with any very long-range edges, of which
our model will have many.
Second, they only consider a single category; for us, the
single-category pruning step is a departure point for the more complex
stages of separating the different categories, and using long-range
links to improve the distortion.

%\subsection{Overlapping Communities}
%\label{sec:related:communities}
\xhdr{Overlapping Communities}
There are conceptual similarities between our work and concurrent and
independent work by Arora et al.~\cite{arora:ge:sachdeva:schoenebeck}
and Balcan et al.~\cite{balcan:borgs:braverman:chayes:teng}.
Their goal is more specifically to reconstruct overlapping community
structure in graphs; similar to our approach, they also posit that the
social network is a noisy signal about some true underlying social
structure, and communities are defined with respect to those structures.
Recall that the goal of properly identifying communities
is also one of the motivations for our work, although we do not
explicitly pursue the question of reconstructing communities with
provable guarantees.

The major difference between our work and that of
\cite{arora:ge:sachdeva:schoenebeck,balcan:borgs:braverman:chayes:teng}
is that both Arora et al.~and Balcan et al.~assume a set-based latent
structure (each community is modeled as a set),
whereas we assume a latent structure based on a near-uniform-density
metric (each social category is modeled as a separate metric space).
This difference, in turn, leads to different random graph models
and algorithmic ideas.
In principle, the set-based structures could be modeled as 0-1 metrics
(and thus fit into our framework); however, such metrics would
dramatically violate our uniform density assumption, so that our
algorithms are not applicable.

Nonetheless, some conceptual similarities between our work and
\cite{arora:ge:sachdeva:schoenebeck,balcan:borgs:braverman:chayes:teng}
are worth noting.
First, a crucial aspect of all three papers is the ability to deal
with overlapping latent structures:
multiple social categories in the present paper,
and multiple communities for
\cite{arora:ge:sachdeva:schoenebeck,balcan:borgs:braverman:chayes:teng}.
All three papers need some notion of ``gap assumption'' that limit overlaps in order to handle such structures.
Second, a high-level idea present in all three papers is to start with
a ``seed'' and then ``grow'' it to find the respective latent
structures.
While the high-level algorithmic ideas are similar, the details differ
significantly between our Amoeba algorithm and the algorithms in
\cite{arora:ge:sachdeva:schoenebeck,balcan:borgs:braverman:chayes:teng}.
The Amoeba algorithm grows the ``Amoeba'' gradually, and using short
disjoint paths, whereas
\cite{arora:ge:sachdeva:schoenebeck,balcan:borgs:braverman:chayes:teng}
use ideas related to finding hidden cliques.
In addition, the goal of reconstructing metrics motivates substantial
algorithmic extensions (discussed in Sections
\ref{sec:additive-singleCat}--\ref{sec:constDeg}) related to improving
the distortion and dealing with small node degrees.
These algorithmic questions have no direct analogue in the setting of
reconstructing communities.

%\subsection{Network Localization, Embeddings, and Distance Oracles}
\xhdr{Network Localization, Embeddings, and Distance Oracles}
Reconstructing (low-dimensional, Euclidean) node distances from
distance measurements has been studied both theoretically
and practically from a wide variety of angles.
In \emph{network localization} for mobile and sensor networks
(e.g.,~\cite{AEGMWYAB,so:ye:tensegrity,zhu:so:ye:universal-rigidity}),
and \emph{network embedding} for peer-to-peer networks and the Internet
(e.g.,~\cite{ng:zhang:coordinates,dabek:cox:kasshoek:morris:vivaldi,wong:slivkins:sirer:meridian,kleinberg:slivkins:wexler:triangulation}),
distances are known fairly accurately, but typically only to a few ``beacon''
nodes.
The challenge is to choose beacons, and combine measurements, to
estimate pairwise distance.
In our setting, the presence or absence of edges provides much less
reliable estimates of distances.
However, once we succeed in obtaining basic distance measurements, the
techniques from network embedding/localization can lead to further
improvements in the estimates without a blowup in the running time,
as shown in Section~\ref{sec:constDeg}.

We measure the quality of our inferred metrics in terms of the
distortion of the estimates. Distortion is commonly used as a measure
of quality in metric embeddings and distance oracles
(see, respectively, \cite{indyk:matousek:embeddings}
and \cite{zwick:exact-approximate} for surveys).
In those domains, distances are known precisely, and the challenge is
typically to find a compact and faithful representation, for instance
in terms of low dimensionality of the target metric or small space of
the oracle.
In our setting, the true distances (in each category) are not
explicitly known, and the estimates are very noisy.
Similar to metric space embeddings, our goal is to extract a faithful
representation of each category.
However, a second fundamental difference is that the space we
``embed'' in consists of multiple metrics, and thus
severely violates the triangle inequality.

Our focus on near-uniform density metrics is
motivated by similar notions of low dimensionality in metric
embedding, nearest neighbor search, and a number of other problems,
e.g.~\cite{karger:ruhl:growth-restricted,gupta:krauthgamer:lee:fractals,krauthgamer:lee:navigating,talwar:bypassing,kleinberg:slivkins:wexler:triangulation}.
\asedit{A particularly related (albeit somewhat less restrictive) notion is \emph{grid dimension} \cite{karger:ruhl:growth-restricted,slivkins:locality-aware,Abr05-SPAA}: the smallest $d\geq 0$ such that doubling the radius of a
ball increases the number of points by at most $2^d$. In comparison, the near-uniform density assumption implies both upper and lower bounds on the number of points.
The near-uniform density assumption has been used, along with various other
modeling assumptions, in
\cite{ResourceLocation,GossipImpossibility,small-world,sarkar:chakrabarti:moore}.
In \cite{hoff:raftery:handcock}, the authors make a qualitatively
similar assumption that the point locations are i.i.d.~samples from a multi-variate Gaussian.}

% slivkins-podc07, ittai-spaa05

%\subsection{Small-World Networks}
%\label{sec:related:small-world}
\xhdr{Small-World Networks}
A long line of empirical studies confirms that many social ties and
interactions correlate strongly with social distance, and particularly
geographical distance
(see, e.g., \cite{mcfarland:brown,mok:wellman:before-internet} for a discussion).
For example, Butts \cite{butts:predictability} gives calculations
showing that geographical information alone could reduce the entropy
in network prediction by roughly 90\% under moderate assumptions.
More specifically, several recent studies
\cite{adamic:adar:search,backstrom:sun:marlow,liben-nowell:novak:kumar:raghavan:tomkins}
show that the frequency of friendships as a function of
(2-dimensional) geographic distance, when corrected for non-uniform
densities, appears to decrease as $\Theta(r^{-2})$.

Small-world models aim to capture the natural tradeoff between a
preference for shorter links and the randomness observed in the
presence of long-range links.
Initial models were due to
Watts and Strogatz \cite{watts:strogatz} and
Kleinberg \cite{small-world,small-world-nips}.
One of the main goals in these papers was to explain why greedy
routing --- based only on the position of one's neighbors in the
metric space --- can discover paths of polylogarithmic length.
Since the publication of \cite{small-world,small-world-nips}, a large
number of papers in the theoretical computer science community have
expanded the models and results in various ways
\cite{barbella:kachergis:liben-nowell:sallstrom:sowell,duchon:hanusse:lebhar:schabanel,fraigniaud:gavoille:compact,fraigniaud:gavoille:kosowski:lebhar:lotker,fraigniaud:gavoille:paul,fraigniaud:giakkoupis:arbitrary,fraigniaud:giakkoupis:power-law,fraigniaud:lebhar:lotker:threshold,giakkoupis:schabanel:dimension,kumar:liben-nowell:tomkins,lebhar:schabanel:augmentation,lebhar:schabanel:routing,manku:naor:wieder,nguyen:martel}.
The main focus in the community has continued to be the ability of
small-world networks to route greedily and efficiently.
In particular, the goal has been to find ways to augment graphs with
suitable long-range links or (semi-)metrics, provide nodes with
additional knowledge or let them perform some local graph exploration,
or exploit non-uniformity in node degrees, all in an effort to achieve
routing along paths of essentially optimal length.
Several good recent surveys summarize the work along these lines
\cite{fraigniaud:small-worlds-survey,kleinberg:decentralized-search,liben-nowell:wayfinding}.

%%%%%%%%%%%%
\section{Definitions and Preliminaries}
\label{Prelims}

We define a formal model for the latent social space
that gives rise to observed social networks.
In general, it will not be a metric space: it
naturally possesses multiple social dimensions, and proximity in just
one of those dimensions (e.g., geography or occupation) usually means
that individuals are ``close.''
First, we define a basic model with a single metric space that models
a single social dimension.
We then discuss how to extend the concept to multiple metrics; in
particular, we formalize a notion of metric spaces being sufficiently
``independent.''

We begin with some formalities.
Throughout, $V$ is a \emph{ground set} of $n$ nodes. For a metric \D,
we use the standard notion of balls, i.e.,
    $\Ball{u}{r} = \Set{v}{\D(u,v) \leq r}$.
We liberally use $O(\cdot)$ notation to simplify the presentation.
In theorem statements, the constants in $O(\cdot)$ can depend on the
constants in our setting. Elsewhere, the constants in $O(\cdot)$ are
absolute, unless noted otherwise.

\asedit{As we introduce a considerable amount of notation, we provide
  an easy reference summarizing all of the notation in
  Appendix~\ref{app:notation}.}

\OMIT{
The high-probability bounds in our analysis stem
from standard Chernoff Bounds, or extensions thereof to sums of
randomly permuted numbers. (See Section~\ref{sec:permutations} for
precise formulations of both.)}

\subsection{A Model for One Social Category}
\label{Distance-Based}
%\xhdr{A model for one social category}
A single category of the latent space is modeled
essentially as a \DIM-dimensional Euclidean space.
More precisely,
$V$ is a subset of the \DIM-dimensional \emph{torus}\footnote{%
Prior work deals with a \DIM-dimensional grid, which is somewhat
undesirable, as there is an asymmetry between the nodes on the border
and on the inside, which gets more pronounced in higher dimensions.},
that is, the nodes lie in $[0,R]^{\DIM}$ for some $R$,
and the distance between points $x, y \in [0,R]^{\DIM}$ is
$\D(x,y)  = ( \sum_i (\min(|x_i-y_i|, R-|x_i-y_i|))^p )^{1/p}$.
%\dkedit{Similar to much previous work \cite{fraigniaud:lebhar:lotker,hoff:raftery:handcock,kermarrec:leroy:tredan,small-world,sarkar:chakrabarti:moore},}
We require that the node density be \emph{nearly uniform}, in the following sense:
any unit cube in the torus contains at least one and at most
\densityK nodes, for some known constant $\densityK \geq 1$.
(Since \densityK will always be a constant, we will sometimes
  hide \densityK factors in $O(\cdot)$ notation.)
%\dkedit{Indeed, in many respects, our kind of latent space models deteriorate
%if node densities can be highly non-uniform} \cite{fraigniaud:lebhar:lotker:threshold}.
For some of our results, we also want to use the actual lattice
structure as a reference: We refer to the graph of integer points from
$[0,R]^{\DIM}$ with edges between all pairs at distance $\D(x,y) \leq 1$
as the \emph{toroidal grid}.

%\dkedit{%
%Each edge $e=(u,v)$ is present \emph{independently} with a probability
%that decreases in $\D(u,v)$.
%While much of the work in the statistics literature
%\cite{handcock:raftery:tantrum,hoff:raftery:handcock,krivitsky:handcock:raftery:hoff,raftery:niu:hoff:yeung,sarkar:chakrabarti:moore}
%assumes a logit-linear model (i.e., that the probability is
%$\Theta(\frac{e^{-\alpha \D(u,v)}}{1+e^{-\alpha \D(u,v)}})$),
%we prefer using an inverse-polynomial distribution pioneered by the
%work on small-world graphs \cite{fraigniaud:small-worlds-survey,small-world,small-world-nips,kleinberg:decentralized-search,liben-nowell:wayfinding,nguyen:martel,watts:strogatz},
%for two reasons.
%First, small-world models naturally include a large number of
%long-range links, which are an important feature of social networks,
%and an important algorithmic challenge.
%Second, to the extent that empirical studies have attempted to measure
%the distributions of link lengths, remarkable fits have been found
%with small-world inverse-polynomial distributions
%\cite{adamic:adar:search,backstrom:sun:marlow,liben-nowell:novak:kumar:raghavan:tomkins}.
%Nevertheless, we believe that an extension to larger classes of
%distributions is a natural and interesting direction for future work.
%}

%Each edge $e=(u,v)$ is present \emph{independently} with a probability
%that decreases in $\D(u,v)$. More specifically,
% we assume that
If nodes $u,v$ are at distance $r = \D(u,v)$, then the
edge $(u,v)$ is present \emph{independently of other edges}, with probability
  $f(r) = \min(1, \CSW \kSW\, r^{-\DIM})$.
Here, $\CSW = \Theta(\frac{1}{\log n})$
is a normalization constant chosen to ensure that the expected
average node degree is $1$ whenever $\kSW = 1$.
Then, \kSW is a parameter controlling the expected average node degree.
When $\CSW \kSW \leq 1$, the expected average degree is exactly \kSW;
otherwise, the dependence of the node degree on \kSW is sublinear and
strictly monotone.
We call \kSW the \emph{target degree}, even though strictly speaking,
it does not equal the average degree.
Following the literature (e.g., \cite{small-world,small-world-nips}),
we focus on the cases $\kSW = O(1)$ and $\kSW = \polylog(n)$.
We use \ESW to denote the edge set obtained from this distribution,
and $\Gr(V,\D[i])$ for the random graph model,
which we call the \emph{\singleSG.}

When $\kSW \geq 1/\CSW$, all edges of length at most 1 are present in \ESW
with probability $1$. Otherwise, even to ensure connectivity of the
social graph, one must insert a suitable ``local edge set'' separately.
(For instance, much of the literature on small-world networks
assumes that the \DIM-dimensional grid is always part of the graph.)
This issue is discussed in more detail in Section \ref{sec:constDeg},
in the context of low node degrees.

Our main result easily extends to a more general model in which, for a
suitably large $R=\polylog(n)$, an edge $(u,v)$ of length $r=\D(u,v)$
is present with probability at least $f(r)$ for all $r<R$, and with
probability smaller than $f(r)$ for all $r\geq R$.
We omit this generalization for ease of presentation.

\subsection{Multiple Social categories}
\label{multi-category}
%\xhdr{Multiple social categories}
%When multiple social categories give rise to edges independently (such
%as work-related, geography-related, and hobby-related friends),
%the social distance tends to match the \emph{minimum} of the distances
%in the different categories, a model
%proposed by Watts et al.~\cite{watts:dodds:newman}, \dkedit{and also
%discussed in \cite[pp.~337, 348]{handcock:raftery:tantrum} and
%\cite[p.~335]{schweinberger:snijders:settings}}.
%The resulting space will not be a metric.
%We define a slight generalization, allowing for different degrees in
%the different categories.

When multiple social categories give rise to edges independently (such
as work-related, geography-related, and hobby-related friends),
we model the observed social network as the \emph{union} of
the graphs generated by the individual categories.
Formally, each social category is a \singleSG $\Gr[i] = \Gr(V,\D[i])$
with near-uniform density for $i = 1, \ldots, \NUMCAT$,
and the edge sets of the \Gr[i] are mutually independent.
\NUMCAT is a (small) constant.
Balls with respect to the category-$i$ metric are denoted by
\Ball[i]{u}{r}.
A \emph{\multiSG} is obtained by taking the \emph{union} of
all edges, i.e. $\ESW = \bigcup_{i=1}^\NUMCAT \ESW[i]$.
Taking the union is analogous to defining the social distance
  as the minimum over the categories; in particular, the social space
  thus defined is not a metric.

The different categories may have different parameters, such as the
target degree or number of dimensions. If the target degrees are
vastly different, then one category could be completely
``drowned out'' by other, denser, categories, which would make it
impossible to observe its structure.
Therefore, we assume that the target degrees \kSW[i] of the
categories are within a known constant factor of one another.
We define \emph{the} target degree of the \multiSG as the average
$\kSW = \frac{1}{\NUMCAT} \cdot \sum_i \kSW[i]$.

%lets highlight this aspect - is it crucial for our results.

\subsection{Local disjointness of categories}
%\xhdr{Local Disjointness of Categories}
In order to be able to distinguish the edges arising from different
categories, it is necessary that the underlying metrics of different
categories be sufficiently different. We capture this intuition by
requiring that any pair of small balls in two different
  categories be sufficiently different:
formally, the \emph{\LCD condition} states that for any two balls
\Ball[i]{u}{r}, \Ball[i']{u'}{r'} in distinct categories $i \neq i'$,
with $r, r' = O(\polylog(n))$,
\begin{align}\label{eq:cat-disj}
|\Ball[i]{u}{r} \cap \Ball[i']{u'}{r'}| \leq O(\log n).
\end{align}

This condition suffices for our main result;
%\footnote{%
%\ASedit{In fact, the right-hand side of~\eqref{eq:cat-disj} can be increased, with the  corresponding degradation of additive error, as long as it is approximately known to the algorithm, and significantly smaller than the maximal allowed radius of the balls in~\eqref{eq:cat-disj}.}}
some of the extensions
require a similar but stronger local condition called
\globalCD[R], which will be introduced in
Section~\ref{sec:additive-mult}.
The \LCD condition is not overly strong; for instance,
we prove (in Section~\ref{sec:permutations}) that both \LCD and
\globalCD[R] hold with high probability when node identifiers within
each category are randomly permuted.
%\dkedit{Notice that this still allows for adversarial placement of
%nodes within categories; our assumption is thus significantly weaker
%than the uniformly random point locations assumed, e.g., in
%\cite{hoff:raftery:handcock,sarkar:chakrabarti:moore}.}

%\subsection{Input and Output}
%\label{input-output}

\subsection{Input and output}
%\xhdr{Input and output}
Since our model has several parameters, we need to be precise about
what is known to the algorithm.
Most importantly, in terms of the social network, only the union \ESW
of all social network edges is revealed to the algorithm;
the division into individual categories \ESW[i] is not given.

We assume that the algorithm knows how many embeddings it needs to
construct, and into what spaces. More formally, this means that
\NUMCAT (the number of categories), \DIM[i] (the number of
dimensions), and $R_i$ (the sizes of the tori) are known to the
algorithm. The average target degree \kSW can be estimated from the
expected degree, and by Chernoff Bounds, such an estimate will be
within $1\pm O(n^{-1/2})$ of the correct value with high probability.
According to the model, the individual target degrees \kSW[i]
lie within a constant factor of \kSW, and we assume that this constant
factor is also known to the algorithm.
To simplify presentation (i.e., this is not a part
  of the model) we assume that the target degrees $\kSW[i]$ and the
  dimensions $\DIM[i]$ are the same for all categories $i$, and that
  $\kSW$ is known.

%Since most of our analysis is for the near-uniform density model,
We also assume that the upper bound \densityK on the number of
points in any unit cube is known to the algorithm.
Knowing \densityK and the other model parameters,
the normalization constant $\CSW = \Theta(\tfrac{1}{\log n})$ can also
be computed to within a constant factor.

The goal of the algorithm is to output metrics \DP[i] that approximate
the original \D[i]. If the output satisfies

\[
\CONTR \D[i](u,v) \; \leq \; \DP[i](u,v) \; \leq \; \EXPAN \D[i](u,v) + \Delta
\]
for all node pairs $u,v$, then we say that \DP[i] estimates \D[i] with
\emph{contraction} \CONTR, \emph{expansion} \EXPAN and
\emph{additive error} $\Delta$.
The \emph{multiplicative distortion} of \DP[i] is then $\EXPAN/\CONTR$.
If we mention no multiplicative distortion (or contraction), then we
implicitly refer to the case of distortion (contraction) 1.
We do not require that \DP[i] itself be a \DIM[i]-dimensional
  Euclidean metric, only that it approximate \D[i] with low distortion.

\subsection{Probability}
%\xhdr{Probability}

Most of our results are with high probability, with respect to the
randomness in the graph generation process.
By this, we mean that the success probabilities are $1-n^{-c}$, where
the constant $c\geq 1$ is large enough to allow all needed
applications of the Union Bound (over polynomially many events).
By a slight abuse of notation, we will write \emph{with high probability}
for probability $1-n^{-c}$, without explicitly specifying the
constant $c \geq 1$.

In many places, we bound tail deviations using
standard \emph{Chernoff Bounds}. Specifically, we use the following
version, which can be found, e.g.,
in~\cite[pages 6--8]{dubhashi:panconesi:concentration-book}.

\begin{theorem}[Chernoff Bounds]\label{thm:Chernoff}
Let $X$ be the sum of independent random variables distributed in
$[0,1]$, and let $\mu' \geq \mu = E[X]$.
Then the following hold:
\begin{align}
\Prob{|X-\mu| > \delta\mu} & \leq \exp(-\mu\, \delta^2/3),
    \qquad \text{ for any } \delta>0 \label{eqn:chernoff-two-sided}\\
\Prob{X > (1+\delta)\mu'} & \leq \exp(-\mu'\, \delta^2/3),
    \qquad \text{ for any } \delta\in(0,1). \label{eqn:chernoff-upper}
\end{align}
\end{theorem}

The bounds in Theorem~\ref{thm:Chernoff} sometimes apply (and are
useful) even when the summands are not independent.
In particular, our analysis of \LCD and \globalCD[R] in
Section~\ref{sec:permutations} uses one such result in which the
randomness arises from a random permutation.
We state and prove the corresponding version of Chernoff
Bounds in that section.

%%%%%%%%% OMITTED STUFF BELOW %%%%%%%%%%%%%

\OMIT{ %%%
If no additive error is mentioned, then $\Delta=0$.
We say that an algorithm estimates \D with additive error $\Delta(r)$ if
  $|\D(u,v) - \DP(u,v)| \leq \Delta(\D(u,v))$
for all node pairs $(u,v)$.
In this paper, we are especially interested in the case when
$\Delta(r)$ is a sublinear function.
} %%%

\OMIT{
Let \Eloc be a set of some (not necessarily all) edges of length at
most 1; we call \Eloc the \emph{local edges}.
\Eloc is \emph{faithful} with respect to $(V,\D)$ iff
the unweighted shortest-paths distances in $(V,\Eloc)$ are within a
(known) constant factor of $\D$, for each node pair $u,v$.
Thus, faithful local edge sets ensure that the local structure
``roughly'' captures the structure of the latent space;
we assume throughout that local edge sets are faithful.
}

\OMIT{
For a non-decreasing function $f: \R^+ \to [0,1]$,
a \emph{distance-based random graph} $\Gr = \Gr(V,\D,f)$
is a distribution over graphs $G=(V,\ESW)$ on $V$ such that
$\ESW \supseteq \Eloc$ deterministically,
and each edge $e=(u,v)$ is in \ESW independently, with probability
$f(\D(u,v))$.
In this sense, our model differs slightly from the standard model
\cite{small-world,fraigniaud:small-worlds-survey}, where the number of
(outgoing) random links per node is fixed.
For our purposes, fixed degrees introduce dependencies between the
presence of edges, without any compensating qualitative insights;
hence, we use a model with edge independence.
}

\OMIT{
The canonical example of distance-based random graphs are
\emph{small-world graphs} (SWG)
\cite{fraigniaud:small-worlds-survey,small-world,small-world-nips,kleinberg:decentralized-search,liben-nowell:wayfinding,nguyen:martel,watts:strogatz}.
In the notation above, SWG are typically defined with a ground set
$V \subseteq \R^{\DIM}$, social distance \D being the $\ell_1$ or
$\ell_2$ norm in $\R^{\DIM}$, the local edge set being the
\DIM-dimensional integer grid,
and an edge probability function $f(r) = \Theta(r^{-\DIM})$.
There is overwhelming experimental evidence that the distribution of
friendships follows a decreasing function $f(r)$; see, e.g.,
\cite{mok:wellman:before-internet} for a discussion.
Indeed, several empirical studies suggest that --- so long as the
analysis corrects for non-uniform node densities by considering
rank-based friendship distributions --- an edge probability of
$\Theta(r^{-\DIM})$ matches observed friendships and interactions surprisingly well
}

\OMIT{
\begin{enumerate}
\item The simplest such model, a natural modification of
the model in \cite{small-world}, assumes that the nodes are located at
all integer points in $[0,n^{1/\DIM}]^{\DIM}$, and the local edges are
axis-parallel between all unit-distance pairs of nodes.
The downside of this model is that the rigid structure would give
algorithms additional (occupancy-based) information for placing nodes,
which would significantly limit the practical relevance of the
algorithms.
\item A somewhat more general model, which we call
\emph{perfectly uniform density}, posits that each hypercube with side
length $N \in \N$ contains $N^{\DIM} \pm O(N^{\DIM-1})$ points.
This model rules out algorithms based on occupancy arguments.
\end{enumerate}
}

\section{The main result}
\label{sec:amoeba}

In this section, we present our main result, an algorithm for
distance reconstruction for multiple categories with constant
distortion.

%\iaedit{MOVE TO SECTION 2:For simplicity of notation and exposition,
%we assume that all of the categories have the same dimension $\DIM = \DIM[i]$.}

\begin{theorem}\label{thm:main}
Consider a \multiSG with $\CSW \kSW=\Omega(\log n)$,
near-uniform density and \LCD.
There is an algorithm that with high probability reconstructs
distances in each category with constant expansion, no contraction,
and $\polylog(n)$ additive error.
Moreover, such distance estimates
(as spanner graphs or as distance labels)
can be computed in time $n \polylog(n)$.
\end{theorem}

%\ASedit{The main thrust is to obtain a slower (quasi-polynomial) algorithm with the desired reconstruction guarantees. Then in Section~\ref{sec:amoeba-runtime} we outline how to implement this algorithm in near-linear time.}

\subsection{Overview and intuition}

We begin with a high-level overview of the algorithm and
the intuition for the proof, before discussing the different stages in
detail in individual subsections. Recall that the algorithm's input is
the set $\ESW = \bigcup_i \ESW[i]$ of edges from all categories.
For the entire section, we assume that the average node degree is high
enough: $\CSW \kSW = \Omega(16^{\DIM} \NUMCAT^3 \log n)$
for a sufficiently large constant in $\Omega(\cdot)$.
Let $\localR = \Theta((\CSW \kSW)^{1/\DIM})$
be the \emph{local radius}:
by definition of the generative model, all edges between node pairs
$(u,v)$ at distance $\D(u,v) \leq \localR$ are in \ESW with
probability 1.
We define the \emph{pruning radius} to be
$\prunedR = \Theta(\localR \NUMCAT^{2/\DIM})$.

The algorithm proceeds in multiple stages. Each of these stages makes
use of the (random) long-range edges. To avoid stochastic dependencies
between the stages, we can randomly partition the edges of \ESW into a
constant number of sets. Each stage then makes use of its own
set. Since the nodes' degrees are high enough, this does not affect
the high-probability guarantees. For ease of notation, we will not
explicitly talk about the partitions for the remainder of this section.
%Unless we explicitly state that a guarantee holds deterministically,
All results in this section hold with high probability.

In the first stage, called the \emph{\SimpleTest,} the algorithm
produces a \emph{pruned set} \prunedE (which need \emph{not} be a subset of
\ESW), with the following guarantee for all node pairs  $(u,u')$:

\begin{itemize}
\item If $u,u'$ are at distance at most \localR in (at least) one
  category $i$, then $(u,u') \in \prunedE$.

\item If $u,u'$ are at distance at least \prunedR in all categories
  $i$, then $(u,u') \notin \prunedE$.
\end{itemize}

Thus, the guarantee is that all short edges are present,
and all sufficiently long edges are absent.
The algorithm makes no guarantees for node pairs in the intermediate
distance range.

To achieve this pruning, the \SimpleTest counts the number of 2-hop paths
(common neighbors) between $(u,u')$, and compares it to a carefully
chosen threshold.
Similar to what Sarkar et al.~\cite{sarkar:chakrabarti:moore} showed
for the single-category case and the logit-linear edge probabilities,
our analysis shows that this simple
heuristic can provide provable distortion guarantees under the
small-world model, even in the more difficult case of multiple categories.

In the second stage, called \emph{\Amoeba stage}, the algorithm covers \prunedE with
individual edge sets \amoebaE[i] (which need not be disjoint);
the set \amoebaE[i] corresponds to category $i$.
The key property we prove is that whenever $u,v$ are at distance at
most \localR in category $i$, then $(u,v) \in \amoebaE[i]$, whereas
$(u,v) \notin \amoebaE[i]$ whenever $u$ and $v$ are at distance at least
$\amoebaR = \Theta(\prunedR \NUMCAT^{3/\DIM}) = \Theta(\localR \NUMCAT^{5/\DIM})$.
Again, for the intermediate range, the algorithm makes no guarantees
about the presence or absence of edges.
This guarantee
implies that the shortest-path metric%
\footnote{\asedit{Here and throughout, a shortest-path metric of a
    given edge set is with respect to the hop count, unless specified
    otherwise.}}
 of \amoebaE[i] gives an
embedding of \D[i] with constant  multiplicative distortion
$O(\NUMCAT^{5/\DIM})$ for all node pairs at distance at least \localR,
and poly-logarithmic additive distortion for all node pairs at
distance at most \localR.

The algorithm constructs the edge sets \amoebaE[i] one by one. For each $i$,
it begins by finding a poly-logarithmically large clique in \prunedE
that is sufficiently spread out in all previously constructed \amoebaE[j].
(We show using the \LCD condition that the node set of this clique will have
diameter at most $4 \prunedR$ in some category $i$).
Starting from this clique, as long as possible, it adds edges $(u,v)$
that are ``supported'' by enough edges (in \ESW) between $v$'s
neighborhood in \amoebaE[i] and $u$.
The key part of our analysis is to show
that this process will indeed add all sufficiently short edges
(and in particular end up having added all nodes),
while excluding all edges that are long in category $i$.

Throughout this section, we frequently count the number of edges in \ESW
between two node sets (one of which may be a single node).
We usually calculate the expectation, and then invoke Chernoff Bounds
to guarantee that the number of edges is within the desired range. The
expectation or desired number of edges will be (at least) logarithmic,
allowing the application of Chernoff Bounds.

%\subsection{Pruning stage: the \SimpleTest}
%\label{sec:two-hop-test}

% We first provide more details and the analysis for the \SimpleTest.

%%%%%%%%%%
\subsection{Pruning stage: the \SimpleTest}

For a node pair $u,v$, let $\Mtwo(u,v)$ be the number of
two-hop $u$-$v$ paths in \ESW, i.e., the number of common neighbors of
$u$ and $v$ in \ESW. The \SimpleTest{} is as follows:
\begin{align}\label{eq:SimpleTest}
\text{for each node pair $(u,u')$, accept if $\Mtwo(u,u') \geq \Mtwo$, reject otherwise.}
\end{align}
We define the threshold as $\Mtwo = \Theta(\kSW \CSW)$, where the constant in $\Theta(\cdot)$
can be calculated explicitly from the known parameters.
%\dkmargincomment{I thought all $\Theta(\cdot)$ were supposed to be
%  absolute. Does this one depend on parameters after all?}
Henceforth, let $\prunedE$ be the set of all accepted node pairs.

\begin{lemma} \label{lm:two-hop}
With high probability, the \SimpleTest accepts all node pairs of
distance at most \localR in some category, and rejects all node pairs
whose distance is at least $\prunedR$ in all categories.
\end{lemma}

\begin{proof}
The proof is based on a careful decomposition of the metric space into
intersections of rings around $u$ and $u'$, allowing a sufficiently
accurate estimate of the number of their common neighbors.

We begin by proving the positive (acceptance) part.
If $u,u'$ are at distance $\D[i](u,u') \leq \localR$,
then they are close enough such that the balls
\Ball[i]{u}{\localR} and \Ball[i]{u'}{\localR}
overlap in a (dimension-dependent) constant fraction of their nodes.
Counting the size of this overlap, and using that
$\localR = \Theta((\kSW \CSW)^{1/\DIM})$, we get that
\begin{align*}
|\Ball[i]{u}{\localR} \cap \Ball[i]{u'}{\localR}|
& \geq \; \Omega(2^{-\DIM} |\Ball[i]{u}{\localR}|) \\
& \geq \; \Omega(2^{-\DIM} \Theta((\kSW \CSW)^{1/\DIM})^{\DIM}) \\
& \geq \; \Omega(\kSW \CSW),
\end{align*}
for a sufficiently large constant in the definition of \localR.
In the original model, each edge between $u$ or $u'$ and a node in
$\Ball[i]{u}{\localR} \cap \Ball[i]{u'}{\localR}$ is present with
probability 1.
Even if the edge set is randomly partitioned into a constant
number of edge sets for the different stages of the algorithm,
both $u$ and $u'$ will have edges to each node in
    $\Ball[i]{u}{\localR} \cap \Ball[i]{u'}{\localR}$
independently with constant probability.
An application of the Chernoff Bound therefore guarantees that
$\Mtwo(u,u') > \Omega(\kSW \CSW)$
with high probability, and $\Mtwo = \Omega(\kSW \CSW)$
for a suitably chosen constant.

For the second part of the lemma (rejection), fix two nodes $u,u'$
such that $\D[i](u,u') > \prunedR$ for all categories $i$.
Consider two categories $i,i'$ ($i=i'$ is possible), and define
$S_{i,i'}$ to be the set of all nodes $v$ such that
$(u,v) \in \ESW[i]$ and $(u',v) \in \ESW[i']$.
We prove a high-probability bound of
$O(\CSW \kSW/\NUMCAT^2)$ on $|S_{i,i'}|$ for a suitably small
(absolute) constant in the $O(\cdot)$.
A union bound over all $\NUMCAT^2$ pairs $i,i'$ then
implies the claim.

We define a sequence of concentric rings of exponentially increasing
radius around $u$, as follows:
\begin{align*}
R_0 &= \Ball[i]{u}{\prunedR/2} \\
R_j &= \Ball[i]{u}{2^{j/\DIM} \cdot \prunedR/2}
            \setminus \Ball[i]{u}{2^{(j-1)/\DIM} \cdot \prunedR/2} \\
    &= \Set{v}{\D[i](u,v) \in
            (2^{(j-1)/\DIM}\cdot \prunedR/2,\; 2^{j/\DIM}\cdot \prunedR/2)},
            \; \text{for each } j\geq 1.
\end{align*}
So $R_j$ is the set of nodes at distance roughly $2^{j/\DIM} \cdot \prunedR/2$ from $u$ in category $i$.
Likewise, we define the concentric rings around $u'$, with respect to
category $i'$:
\begin{align*}
R_0 &= \Ball[i']{u'}{\prunedR/2} \\
R_j &= \Ball[i']{u'}{2^{j/\DIM} \cdot \prunedR/2}
            \setminus \Ball[i']{u'}{2^{(j-1)/\DIM} \cdot \prunedR/2}
            \; \text{for each } j\geq 1.
\end{align*}

The rings $\{R_j\}_{j \geq 0}$ form a disjoint cover of $V$, as do the rings
$\{R'_j\}_{j \geq 0}$. To bound the size of $S_{i,i'}$, we bound
$S_{i,i'} \cap R_j \cap R'_{j'}$ for all $j,j' \geq 0$.

First consider the case $j=j'=0$.
For $i = i'$, $R_0$ and $R'_0$ are disjoint by definition, and for
$i \neq i'$, the \LCD condition ensures that
$|R_0 \cap R'_0| = O(\log n)$.

Next, we consider the case $j \geq j',\; j \geq 1$.
(The case $j'\geq j,\; j' \geq 1$ is symmetric.)
We write $r = 2^{j/\DIM} \cdot \prunedR/2$ and $r' = 2^{j'/\DIM} \cdot \prunedR/2$.
By definition of the edge generation model, the probability that
$v \in R_j$ has an edge to $u$ in $\ESW[i]$ is at most
$\CSW \kSW (r/2^{1/\DIM})^{-\DIM} = 2\, \CSW \kSW r^{-\DIM}$,
while the probability that $v\in R'_{j'}$ has an edge to $u'$ in \ESW[i'] is at most
$2\, \CSW \kSW (r')^{-\DIM}$, or at most $1$ if $j'=0$.
The presence of these edges is independent of one another.
Because $R_j \cap R'_{j'}$ is contained in $\Ball[i']{u'}{r'}$,
it can contain at most
$\densityK (r')^{\DIM} = O((r')^{\DIM})$ nodes.%
\footnote{Recall that we include \densityK terms in $O(\cdot)$.}
Thus, both for the case $j'=0$ and $j' > 0$, we obtain that
\begin{align*}
\Expect{\,|S_{i,i'} \cap R_{j} \cap R'_{j'}|\,}
& \leq O\left((\CSW \kSW)^2 \; r^{-\DIM} (r')^{-\DIM} (r')^{\DIM} \right)\\
& \leq O\left((\CSW \kSW)^2 \; (2^{j/\DIM} \cdot \prunedR/2)^{-\DIM}\right)\\
& \leq O\left((\CSW \kSW)^2\; 2^{\DIM}\, \prunedR^{-\DIM}\, \cdot 2^{-j} \right).
\end{align*}
We now first sum over all $j \geq j'$ (using that
$\sum_{j \geq j'} 2^{-j} = O(2^{-j'})$), and then over all
$j'$, to obtain that
\begin{align*}
\sum_{j,j':\; j+j' > 0} \Expect{\,|S_{i,i'} \cap R_{j} \cap R'_{j'}|\,}
\leq O((\CSW \kSW)^2\, 2^{\DIM}\, \prunedR^{-\DIM}).
\end{align*}
By choosing $\prunedR = \Theta(\localR \NUMCAT^{2/\DIM})$ with a
suitably large (absolute) constant, we can cancel out the $2^{\DIM}$
term and obtain an arbitrarily small absolute constant $\gamma$ in the
$O(\cdot)$ term.
Recalling that $\localR = \Theta((\CSW \kSW)^{1/\DIM})$
and adding the at most $O(\log n)$ nodes (with some absolute constant)
in $S_{i,i'} \cap R_0 \cap R'_0$, we see that
\[ \Expect{|S_{i,i'}|} \leq O(\gamma \CSW \kSW/\NUMCAT^2) + O(\log n).\]

Applying Chernoff Bounds, we obtain that with high probability,
$$|S_{i,i'}| = O(\gamma\, \CSW \kSW/\NUMCAT^2 + \log n),$$
and a union bound over all $i,i'$ now shows that with high probability we have
\[ \Mtwo(u,v) = O(\gamma\, \CSW \kSW + \NUMCAT^2 \log n) < \Mtwo\]
(when $\CSW \kSW$ is large enough and $\gamma$ small enough),
which means that $(u,v)$ will be rejected.
\end{proof}

For the remainder of this section, we condition on the high
probability event of Lemma \ref{lm:two-hop}, i.e., we assume that
\prunedE contains all edges of length at most \localR (in at least one
category) and no edges whose length would exceed \prunedR in all
categories.

Notice that in the single-category case ($\NUMCAT = 1$), the result of
Lemma \ref{lm:two-hop} by itself already gives an expansion of
$\prunedR/\localR = \Theta(1)$, no contraction,
and additive error $\polylog(n)$.
We simply estimate $\D(u,v)$ by the length of the shortest $u$-$v$
path in the pruned graph, multiplied by $\prunedR$.
Lemma~\ref{lem:shortest-paths} analyzes the distortion for a single category,
and will also be used for the multi-category case.
The lemma requires the unit-disk graph to be a good
approximation of the metric space, a property that is obvious for
near-uniform density sets in $\R^d$.

\begin{lemma} \label{lem:shortest-paths}
Let $(V,\D)$ be a metric space. Let $G$ be a graph on $V$ that
includes all node pairs at distance at most $r$ and no node pairs at
distance more than $r'$, for some $r' > r \geq 1$.
Let \D[G] be the shortest-paths metric of $G$.
Let \Dsp be the shortest-paths metric of the unit disk graph on
$(V,D)$, and assume that $\Dsp (u,v) \leq c\,\D(u,v)$ for all node
pairs $(u,v)$, for some constant $c$. Then
\[
    \D(u,v) \; \leq \; r'\cdot \D_G(u,v) \;
    \leq \; \tfrac{cr'}{r} \cdot \D(u,v) + r'.
\]
In words, $r'\cdot \D_G$ reconstructs \D with expansion
$\tfrac{cr'}{r}$, no contraction, and additive error $r'$.
\end{lemma}

\begin{proof}
Fix a node pair $(u,v)$, and let $\rho$ be a shortest $u$-$v$ path in
$G$.
By the triangle inequality, $\D(u,v)$ is a lower bound on the
total metric length of $\rho$, which in turn is at most
$r'\, \D_G(u,v)$, because each hop in $G$ has length at most $r'$.
So $\D(u,v) \leq r'\,\D_G(u,v)$.
Now, let $P$ be a shortest $u$-$v$ path in \Dsp.
Any two nodes on $P$ that are within $r$ hops from
one another are connected by an edge in $G$.
Therefore, $G$ contains a $u$-$v$ path of at most
$\cel{\frac{|P|}{r}}$ hops, which implies that
    $\D_G(u,v) \leq \cel{\frac{\Dsp(u,v)}{r}} \leq 1+\tfrac{c\,\D(u,v)}{r}$.
\end{proof}

%%%%%%%%%%%%%%%%
\subsection{\Amoeba{} stage: mapping edges to categories}
\label{sec:amoeba-alg}
%\xhdr{\Amoeba{} stage: mapping edges to categories}

We now define the \Amoeba stage of the algorithm.
The \Amoeba stage consists of \NUMCAT iterations $i=1, \ldots, \NUMCAT$:
in each successive iteration $i$, a new category is identified
(and re-numbered as category $i$),
and some edges in $\prunedE$ are mapped to this category.
These edges constitute the edge set $\amoebaE[i]$.
Eventually, each edge $e\in\prunedE$ is mapped to at least one category.

The \Amoeba stage is summarized in Algorithm~\ref{alg:amoeba}.
Each iteration $i$ consists of an \emph{initialization phase},
in which we find a suitable clique in $\prunedE$,
and a \emph{growth phase}, in which we grow $\amoebaE[i]$ one edge at a time.
We think of this process as \emph{growing the amoeba}.

In Algorithm~\ref{alg:amoeba} and the subsequent analysis thereof,
we use the following notation.
For a subset $S\subseteq V$, let $\diam_j(S)$ be its diameter in \amoebaE[j].
Let \NeighB{v}{E} denote the (1-hop) neighborhood of node $v$ in the
edge set $E$. We call the clique $C$ from iteration $i$ the \emph{seed
  clique} for category $i$. The condition~\eqref{eq:amoeba-test} is
called the \emph{\amoebaTest:} more precisely, edge $(u,v)$ passes the
test if and only if~\eqref{eq:amoeba-test} is satisfied.

\begin{algorithm}[htb]

{\bf Output.} Estimated social distance $\D'_i$, for each category
$i=1, \ldots, \NUMCAT$.

{\bf Parameters}. Numbers $(\Mtwo,\amoebaM,\amoebaN,\amoebaR)$.

\vspace{2mm}

{\bf Pruning Stage}.
Let $\Mtwo(u,u')$ be \# common neighbors of $u$ and $u'$ in \ESW.
    \[ \prunedE \leftarrow \{ (u,u')\in V\times V:\; \Mtwo(u,u') \geq \Mtwo \}.\]

\vspace{2mm}

{\bf Amoeba Stage}. For each iteration $i=1, \ldots, \NUMCAT$,
\begin{enumerate}
\item \emph{Initialization phase.}
Find any clique $C \subseteq V$ in \prunedE such that
    $|C|\geq \amoebaN$,
and
%   $\diam_j(C) \geq \amoebaR\,\log^2(n)$
    $\diam_j(C) \geq \log^2(n)$
for each category $j=1, \ldots, i-1$.\\
\asedit{If such $C$ does not exist, $\mathtt{halt}$.}
Initialize $\amoebaE = C\times C$.

%\AScomment{Changed $\amoebaR\,\log^2(n)$ to $\log^2(n)$.
%Note that $\diam_j(C)$ is the %diameter in $\amoebaE[j]$, not in $\D_j$.}

\item \emph{Growth phase.}
While there exists an edge $(u,v) \in \prunedE \setminus \amoebaE$ such that
\begin{align}\label{eq:amoeba-test}
\text{\ESW contains at least \amoebaM edges between $u$ and \NeighB{v}{\amoebaE},}
\end{align}
pick any such edge and insert it into \amoebaE.

\item Set $\amoebaE[i] = \amoebaE$.\\
Let $\D'_i$ be the shortest-paths metric of $\amoebaE[i]$, multiplied by $\amoebaR$.
%Exit the cycle if $\prunedE = \cup_{j\leq i}\; \amoebaE[j]$.
%and move to the next iteration if there is an uncovered edge.
\end{enumerate}

{\bf Notation.} Recall that $\diam_j(S)$ is the diameter of a subset
$S\subseteq V$ in \amoebaE[j], and \NeighB{v}{E} denotes the (1-hop)
neighborhood of node $v$ in the edge set
$E$. Condition~\eqref{eq:amoeba-test} is called the
\emph{\amoebaTest}.
\caption{The \Amoeba algorithm.}
\label{alg:amoeba}
\end{algorithm}

The \Amoeba stage is parameterized by numbers $(\amoebaM,\amoebaN,\amoebaR)$. We set
    $\amoebaN = \Theta((\localR/2)^{\DIM})$
and
    $\amoebaM = \Theta(\amoebaN/(8^{\DIM} \NUMCAT^2))$
for suitable constants in $\Theta(\cdot)$. We define
    $\amoebaR = \amoebaC \cdot \NUMCAT^{3/d} \cdot \prunedR$
for a sufficiently large absolute constant \amoebaC,
and call it the \emph{amoeba radius}.\footnote{Recall that
    $\kSW \CSW = \Omega(16^{\DIM} \NUMCAT^3 \log n)$
with a sufficiently large constant. In particular, if
    $\kSW \CSW = \Theta(16^{\DIM} \NUMCAT^3 \log n)$,
then the parameters are
    $\amoebaN = \Theta(8^{\DIM} \NUMCAT^3 \log n)$,
    $\amoebaM = \Theta(\NUMCAT \log n)$
and
    $\amoebaR = \Theta(\NUMCAT^8 \log n)^{1/d}$.}

%%%%%%%%%%%%%%%%%%%%%%%%%%%%%%%%%%%%%%%%
\subsection{Analysis of the \Amoeba stage}
\label{sec:amoeba-analysis}

%\xhdr{Analysis of the \Amoeba stage}
An edge $(u,v)\in \prunedE$ is called \emph{$i$-long} if
$\D[i](u,v) > \amoebaR$, and \emph{$i$-short} if
$\D[i](u,v) \leq \localR$.
An edge set $\amoebaE \subseteq \prunedE$ is an
\emph{$i$-amoeba} iff $(V,\amoebaE)$ contains no $i$-long edges, and
it contains a clique of at least \amoebaN nodes whose category-$i$
diameter is at most $4 \prunedR$.

The high-level outline of the correctness proof for the \Amoeba stage
is as follows. We will prove by induction on $i$ that each
edge set \amoebaE[i] captures (at least) all $i$-short edges
(renumbering the categories appropriately), and does not include
any $i$-long edges.

The induction step requires that the algorithm be able to reconstruct
another category $i$ while there is an uncovered edge.
Thereto, we show that \amoebaE remains an $i$-amoeba throughout
the algorithm. We break the induction step into multiple
lemmas capturing the following \asedit{three} key points:

\begin{itemize}
\item \asedit{All edges in any sufficiently large clique have
    sufficiently small length in at least one category $i$}
\asedit{(Lemma \ref{cl:clique-diam}).}

%\item The required seed clique $C$ of size \amoebaN exists in \prunedE.

\item No $i$-long edge passes the \amoebaTest
\asedit{(Lemma \ref{lm:amoeba-step}).}

\item While there is an $i$-short edge not yet added to \amoebaE, at
least one such edge passes the \amoebaTest.
\asedit{(Lemma \ref{lm:edge-progress}).}

\end{itemize}

\begin{lemma}\label{cl:clique-diam}
Let $C$ be a clique in \prunedE of size
   $|C| > \Omega(\NUMCAT^3 \log n)$,
for a sufficiently large constant in $\Omega(\cdot)$.
Then, with high probability, there exists a category $i$ such that
$\D[i](u,v) \leq 4 \prunedR$ for all $u, v \in C$.
\end{lemma}

\begin{proof}
Fix an arbitrary $w \in C$. Because each edge $(u,v) \in \prunedE$
satisfies $\D[i](u,v) \leq \prunedR$ for some category $i$,
there is a category $i$ such that for at least $|C|/\NUMCAT$ nodes
$v \in C$, we have $\D[i](w,v) \leq \prunedR$.
Fix such a category $i$, and let $S$ be the set of all
$v \in C$ with $\D[i](w,v) \leq \prunedR$.
If $S = C$, then we are done.

Otherwise, consider a node $u \in C \setminus S$.
For each node $v \in S$, there is a category
$i'$ with $\D[i'](u,v) \leq \prunedR$. In particular, there must be a
category $i'$ such that $\D[i'](u,v) \leq \prunedR$ for at least
$|C|/\NUMCAT^2 > \Omega(\log n)$ nodes $v \in S$, with a large enough
constant in $\Omega(\cdot)$.
Fix such a category $i'$, and let $S'$ be the set of nodes
$v \in S$ with $\D[i'](u,v) \leq \prunedR$.
Because
$S' \subseteq \Ball[i]{w}{\prunedR} \cap \Ball[i']{u}{\prunedR}$,
the assumption $i' \neq i$ would contradict the \LCD condition.
Hence $i'=i$, and $u$ is at distance at most $2\prunedR$ from $w$ in
category $i$.
Since this argument holds for every $u \in C \setminus S$,
we have proved that $C$ has diameter at most $4\prunedR$ in
category $i$.
\OMIT{
The preceding argument gives a diameter bound of at most
$4 \prunedR$. We can improve it slightly as follows.
Suppose that there were a node pair $u,u'$ of distance more than
$2\prunedR$ in category $i$. In particular, for each node $v \in C$,
at least one of $u,u'$ must have distance more than \prunedR in
category $i$. But then, at least one of $u,u'$ (w.l.o.g., $u$) must
have $\D[i](u,v) > \prunedR$ for at least $|C|/2$ nodes $v \in C$;
hence, for some $i'$, there must be at least $|C|/(2\NUMCAT)$ of these
nodes for which $\D[i'](u,v) \leq \prunedR$. This implies that
$\Ball[i]{w}{2\prunedR} \cap \Ball[i']{u}{\prunedR}
\supseteq C \cap \Ball[i']{u}{\prunedR}$ has size at least
$|C|/(2\NUMCAT)$, contradicting the \LCD condition.}
\end{proof}

\begin{note}{Remark.}
Lemma \ref{cl:clique-diam} can be restated as saying that for any
edge-coloring of a sufficiently large clique that is consistent with
the \LCD condition\footnote{Reformulated in terms of edge colorings,
the \LCD states that two balls with respect to edges of colors
$i \neq i'$, each of radius $\polylog(n)$,
overlap in at most $O(\log n)$ nodes.},
there is a color $i$ such that the set of edges of color $i$ has
diameter at most 4. Without the \LCD condition, this statement
is false in general for $\NUMCAT \geq 3$.
For a simple counter-example, consider a clique $C$ whose nodes are
partitioned into three sets $C_1,C_2,C_3$ so that color
$i\in \{1,2,3\}$ is assigned to all edges with both endpoints in $C_i$
and to all edges with neither endpoint in $C_i$.
Then, the edge set corresponding to any one color $i$ is not even
connected.
For $\NUMCAT=2$, there is a simple combinatorial proof that does not
involve \LCD.
\end{note}

\asedit{
\begin{claim}\label{cl:clique-large-diameter}
Let $C$ be a clique in \prunedE of size
   $|C| > \Omega(\NUMCAT^3 \log n)$,
for a sufficiently large constant in $\Omega(\cdot)$.
Let $i$ be a category such that
$\D[i](u,v) \leq 4 \prunedR$ for all $u, v \in C$;
the existence of such a category is guaranteed by Lemma \ref{cl:clique-diam}.
Then,  w.h.p., the diameter of $C$ in $\D[j]$ is at least $\amoebaR \log^2(n)$
for all categories $j \neq i$.
\end{claim}

\begin{proof}
For any $j \neq i$, the \LCD condition condition implies that
$|\Ball[j]{u}{\amoebaR \cdot \log^2(n)} \cap B| \leq O(\log n)$.

Thus, there is at least one
$v \in C \setminus \Ball[j]{u}{\amoebaR \cdot \log^2(n)}$.
The lemma follows.
\end{proof}
}

%Because each edge in \amoebaE[j] has length at most \amoebaR in
%category $j$, this means that $\D[j](u,v) > \log^2(n)$;
%in particular, $B$ cannot have diameter less than $\log^2(n)$ in
%\amoebaE[j]. Since this holds for all $j$, $B$ is a candidate for seed
%clique $i$, and the algorithm thus guarantees progress.

%\scedit{
%By Lemma \ref{lm:clique-exists} we have the following corollary.
%\begin{corollary}\label{cr:clique-exists}
%Consider iteration $i$ of the \Amoeba algorithm. Suppose $i \leq \NUMCAT$.
%%there is an edge $e\in \prunedE$ not included in any \amoebaE[j], $j<i$.}
%Then, w.h.p., \prunedE contains a clique of at least \amoebaN nodes whose diameter
%in \amoebaE[j] is at least $\amoebaR  \log^2(n)$ for all $j<i$.
%\end{corollary}
%}

\OMIT{ %%%%
However, for the special case $\NUMCAT = 2$, a simple combinatorial
proof gives a constant diameter in at least one category, for any edge
coloring. For any node $u$, let $R(u)$ and $Y(u)$ be the neighbors of
$u$ along blue and yellow edges respectively (with $u \in R(u), u \in Y(u)$).
If for all $u,v$, we have $R(u) \cap R(v) \neq \emptyset$, then the
set of red edges has diameter at most 2.
Otherwise, the edge $(u,v)$ is yellow and $Y(u) \cup Y(v) = C$,
implying that the set of yellow edges has diameter at most 3.
} %%%%%

\begin{lemma}\label{lm:amoeba-step}
Assume that $\amoebaE \subseteq \prunedE$ contains no $i$-long edge,
and let $u,v$ be nodes with $(u,v) \in \prunedE$ and $\D[i](u,v) > \amoebaR$.
Then, with high probability, $(u,v)$ does not pass the \amoebaTest.
\end{lemma}

\begin{proof}
We bound the number of edges between $u$ and \NeighB{v}{\amoebaE} in two parts:
by the number of edges between $u$ and \Ball[i]{v}{\prunedR},
and the number of edges between $u$ and
$\NeighB{v}{\amoebaE} \setminus \Ball[i]{v}{\prunedR}$.

First, we claim that
$|\NeighB{v}{\amoebaE} \setminus \Ball[i]{v}{\prunedR}| \leq O(\NUMCAT \log n)$.
The reason is that any node
$w \in \NeighB{v}{\amoebaE} \setminus \Ball[i]{v}{\prunedR}$ must be at distance
at most \prunedR from $v$ in some category $j \neq i$ (because $(v,w) \in \prunedE$),
so $w \in \Ball[j]{v}{\prunedR} \cap \Ball[i]{v}{\amoebaR}$.
Now, the \LCD condition implies that there can be at most $O(\log n)$
such nodes $w$ for any fixed $j$, and thus at most $O(\NUMCAT \log n)$
total.

Next, we consider nodes $w \in \Ball[i]{v}{\prunedR}$.
By the \LCD condition for
$\Ball[i]{v}{\prunedR} \cap \Ball[j]{u}{\amoebaR}$,
there can be at most $O(\log n)$ such nodes $w$ at distance at most
\amoebaR from $u$ in category $j$, for a total of $O(K \log n)$ nodes.

All other nodes $w \in \Ball[i]{v}{\prunedR}$ are at distance at least
\amoebaR from $u$ in all categories $j \neq i$, and at distance
at least $\amoebaR-\prunedR \geq \amoebaR/2$ from $u$ in category $i$.
Thus, the probability for the edge $(u,w)$ to exist in any one
category $j$ is at most
$q = O(\CSW \kSW \amoebaR^{-\DIM})
= O(\CSW \kSW /(\amoebaC^{\DIM} \NUMCAT^3) \cdot \prunedR^{-\DIM})$.
Summing over all $w \in \Ball[i]{v}{\prunedR}$ and all categories
gives us at most
$q \NUMCAT | \Ball[i]{v}{\prunedR} |
= O(\CSW \kSW /(\amoebaC^{\DIM} \NUMCAT^2))$
edges in expectation, and Chernoff Bounds prove concentration.
Adding the at most $O(\NUMCAT \log n)$ edges of the first two types,
and recalling that \amoebaC is a suitably large constant
and $\CSW \kSW = \Omega(\NUMCAT^3 \log n)$ with a large constant,
we see that with high probability, the total number of edges between
$u$ and \NeighB{v}{\amoebaE} is less than \amoebaM, so the edge $(u,v)$ does not
pass the Amoeba Test.
\end{proof}

\begin{lemma} \label{lm:edge-progress}
Let \amoebaE be an $i$-amoeba that does not include all $i$-short
edges. Then, w.h.p., there exists an edge $(u,v) \in \prunedE$ that is
accepted by the Amoeba Test.
\end{lemma}

\begin{proof}
First notice that because the Amoeba Test only counts edges from $u$
to a neighborhood of $v$, it is monotone in the following sense:
if the edge $e$ passes for some current edge set \amoebaE, then it
also passes for any $\amoebaE' \supseteq \amoebaE$.
We will define an ordering $e_1, e_2, \ldots$ of all edges in category
$i$ such that with high probability, $e_\ell$ will pass the Amoeba
Test whenever $C \cup \SET{e_1, \ldots, e_{\ell-1}} \subseteq \amoebaE$.
Thus, \Amoeba, starting from $C$, can always make
progress when considering the lowest-numbered edge $e_\ell$ not yet included.
(Notice that this does not require the algorithm to actually know the ordering.)

Let $C$ be the clique in $(V,\amoebaE)$ of size at least \amoebaN
whose existence is guaranteed by the definition of an $i$-amoeba.
$C \subseteq \Ball[i]{w}{2 \prunedR}$ for some $w$, and
\Ball[i]{w}{2 \prunedR} can be covered by
$O((\prunedR/\localR)^{\DIM}) = O(\NUMCAT^2)$ balls of radius $\localR/2$,
at least one of which must therefore contain a sub-clique
$C' \subseteq C$ of at least $\amoebaN/\NUMCAT^2$ nodes.
Let $v_0$ be the center of such a ball \Ball[i]{v'}{\localR/2}.

First, all edges between $u \in \Ball[i]{v_0}{\localR/2}$
and $v \in C'$ will pass the Amoeba Test, because
$(u,w)$ is $i$-short for all $w \in C' \subseteq \NeighB{v}{\amoebaE}$
(implying that the edge $(u,w)$ is in \prunedE),
and $|C'| \geq \amoebaN/\NUMCAT^2 \geq \amoebaM$.

Second, because each $v \in \Ball[i]{v_0}{\localR/2}$
is now connected to all of $C'$ in \amoebaE, the exact same argument
applies to all node pairs $u, v \in \Ball[i]{v_0}{\localR/2}$.

Third, we use induction on $r$, showing that once
all edges in \Ball[i]{v_0}{r} have been included, all edges in
\Ball[i]{v_0}{r+1} will be included next in some order.
For the base case, we use $r=\localR/2$.
Let $u$ be any node in $\Ball[i]{v_0}{r+1}\setminus \Ball[i]{v_0}{r}$,
and $w$ a node ``close to $u$ on the line from $v_0$ to $u$.''
More formally, $w$ is a node with
    $\D_i(v_0,w)\leq r-\localR/4$
and
    $\D_i(u,w)\leq \localR/4 + O(1)$.
The existence of $w$ follows by the near-uniform density assumption.

By near-uniform density, the ball $B' = \Ball[i]{w}{\localR/4}$ contains at
least $\Omega(2^{-\DIM} \amoebaN)$ nodes,
and by induction hypothesis, all nodes of $B'$ are neighbors of
$v$. Furthermore, \ESW contains edges between $u$ and all $w$ with
constant probability, so using Chernoff Bounds, with high probability,
the pair $(u,v)$ will pass the Amoeba Test for all
$v \in B'$, inserting all these edges. Once all $i$-short edges
between $u \in \Ball[i]{v_0}{r+1}$ and $v \in \Ball[i]{v_0}{r}$ have
been inserted, the $i$-short edges between the remaining pairs
$u, v \in \Ball[i]{v_0}{r+1}$ will be inserted by the following
argument. Node $u$ has $i$-short edges to all nodes in $B'$ (which are
already in \amoebaE), so
$\D[i](v,w) \leq 2 \localR$ for all $w \in B'$.
Thus, each edge from $v$ to $w \in B'$ is included with probability at
least $p= \Omega(\CSW \kSW 2^{-\DIM} \localR^{-\DIM})$, and there are
at least $|B'| \geq \Omega(4^{-\DIM} \localR^{\DIM})$ such nodes,
implying that the expected number of edges between $v$ and the
neighborhood of $u$ is at least $\Omega(8^{-\DIM} \CSW \kSW)$.
By Chernoff Bounds, we obtain concentration
results, and because $\amoebaM \leq \Theta(8^{-\DIM} \CSW \kSW)$,
the edge $(u,v)$ will be included with high probability.
\end{proof}

The following theorem combines the previous lemmas and proves the correctness of the \Amoeba algorithm.

\asedit{
\begin{theorem}
The \Amoeba algorithm has exactly \NUMCAT iterations,
and for every category $1 \leq j \leq \NUMCAT$, there exists an
iteration $i$ such that $\amoebaE[i]$ contains all $j$-short
edges and no $j$-long edges.
\end{theorem}

\begin{proof}
We prove by induction that in each iteration $i$, the algorithm finds
a distinct category $j(i)$ such that
$\amoebaE[i]$ contains all $j(i)$-short edges and no
$j(i)$-long edges.
For the base case where $i=0$, the claim is trivial.

Consider an iteration $i$, and let $j$ be a category with
$j \neq j(i')$ for all $i' < i$.
For an arbitrary node $u$, consider the ball
$B = \Ball[j]{u}{\localR/2}$.
Because $\D[j](v,v') \leq \localR$ for all $v, v' \in B$, the set $B$
forms a clique in \prunedE.
Furthermore, because of the near-uniform density of category $j$, $B$ has
$\Theta((\localR/2)^{\DIM}) = \Theta(\CSW\kSW) = \Omega(\NUMCAT^3 \log n)$
nodes, for a sufficiently large constant in the $\Omega(\cdot)$.
Further, fix iteration $i'<i$ and let $j'=j(i')$ be the corresponding category.
By Claim~\ref{cl:clique-large-diameter}, the diameter of $B$ in $\D[j']$ is at least $\amoebaR \log^2(n)$.
By the induction hypothesis, the set $\amoebaE[i']$ contains no $j'$-long edges; therefore
    $\diam_{i'}(B)\geq \log^2(n)$.

Thus, the \Amoeba algorithm is guaranteed to find a seed clique in the
initialization phase. Let $C_i$ be the clique that is actually found.
By Lemma~\ref{cl:clique-diam}, there exists a category $j(i)$
(possibly distinct from $j$) such that
$\D[j(i)](u,v) \leq 4 \prunedR$ for all $u, v \in C_i$.
Moreover, by construction of $C_i$, for every $i' < i$,
the diameter of $C_i$ in \amoebaE[j(i')] is at least $\log^2(n)$;
hence, the diameter of $C_i$ in $\D[j(i')]$ is also at least
$\log^2(n) > 4 \prunedR$, and in particular, $j(i) \neq j(i')$ for all
$i' < i$.
By Lemmas \ref{lm:amoeba-step} and \ref{lm:edge-progress},
\amoebaE[i] contains all $j(i)$-short edges and no $j(i)$-long edges.
\end{proof}
}

%\dkedit{Finally, because all short edges for a distinct category are included
%in any one iteration, the number of iterations must be exactly \NUMCAT.}

The algorithm will thus terminate with $i$-amoebas
$\amoebaE[i], i = 1, \ldots, \NUMCAT$.
The distance $\D_i(u,v)$ is now estimated as the shortest-path
distance between $u$ and $v$ in \amoebaE[i], multiplied by \amoebaR.
By Lemma \ref{lem:shortest-paths}, this gives constant expansion
$\amoebaR/\localR = \Theta(\NUMCAT^{5/\DIM})$, no contraction, and
additive error $\amoebaR$.

\OMIT{ %%%%%%%
If the true distance between $u$ and $v$ is at most
\localR, the estimate of
$\amoebaR = \Theta(\localR \NUMCAT^{5/\DIM}) = \Theta((\NUMCAT^5 \CSW \kSW)^{1/\DIM})$
is off by at most an additive poly-logarithmic term.
If the distance between $u$ and $v$ is at least $\localR$,
then the estimate is off by at most a multiplicative factor
$\amoebaR/\localR = \Theta(\NUMCAT^{5/\DIM})$, i.e., a constant.
\dkcomment{How do we get the multiplicative bound again? Is there a
  lemma written up somewhere?}
}

%Finally, let us discuss the running time. Since most of
%the steps are just counting edges and comparing the count to suitable
%thresholds, the running time  is dominated by the time to find suitably
%large cliques. Since we are searching for cliques of poly-logarithmic
%size, the algorithm takes quasi-polynomial time.
%%%%%%%%%%%%%%%%%%%%%%%%%%%%%%%%%%%%%%%%

%%%%%%%%%%%%%%%%%%%%%%%%%%%%%%%%%%%%%%%%%%
\subsection{Efficient Implementation of the \Amoeba algorithm}
\label{sec:amoeba-runtime}

We outline how to implement the \Amoeba algorithm in near-linear time.
The first (and perhaps most surprising) step is quickly
finding the seed clique.
Then, we need to execute each \Amoeba step in (amortized)
polylogarithmic time.
The resulting algorithm computes the graph \amoebaE[i] for each
category $i$ in near-linear time.
Recall that $\amoebaE[i]$ is a constant-distortion \emph{spanner} for
$\D_i$, in the sense that its shortest-path metric approximates $\D_i$.
Once we have a spanner, we can compute succinct distance labels by
adapting a hierarchical beaconing technique from prior work on
distance labeling and routing schemes
(e.g.~\cite{gupta:krauthgamer:lee:fractals,chan:gupta:maggs:zhou,slivkins:rings,slivkins:locality-aware}).
We next describe each of these steps in more detail.

\vspace{2ex}
% there was a strange spacing issue for me, with negative space?

\subsubsection{Finding the seed clique}

By suitably adjusting the threshold $\Mtwo$,
the \SimpleTest can be modified to accept all node pairs that are
within distance $\localR'=3\,\prunedR$ in some category,
and to reject all node pairs that are at distance at
least $\prunedR' = \Theta(K^{2/d}\, \localR')$ in all categories.
We run the \Amoeba algorithm on the pruned graph $\prunedE'$ obtained by
this modified \SimpleTest.
\asedit{Note that the entire analysis still applies if we replace
$\prunedR$ with $\prunedR'$, because the guarantees on the lengths of edges
differs only by constant factors between $E'_{pru}$ and $E_{pru}$.}
Let $\amoebaR'$ be the corresponding Amoeba radius.
To produce the seed cliques for $\prunedE'$, we use
the original \SimpleTest \asedit{with the original $\prunedE$},
in the way described below.

Consider the original \SimpleTest, and let $\prunedE$ be the
corresponding pruned graph.
Let $N(u)$ denote the 1-hop neighborhood of node $u$ in \prunedE,
including $u$ itself.
For a node set $S$, define $N(S)$ to be the
\emph{intersection} $N(S)\triangleq \bigcap_{u\in S} N(u)$.
We focus on such intersections for
node sets $S \subseteq N(u)$ of size
$|S| = \NUMCAT$.

\begin{lemma}\label{lem:find-clique-fast}
For any node $u$ and category $i$, there exists a set
$S \subseteq N(u)$ of size \NUMCAT such that
the intersection $N(S)$ contains at least $\amoebaN$ nodes,
has diameter at most $3\,\prunedR$ in category $i$,
and diameter at least $R = \amoebaR'\,\log^2(n)$ in all other categories.
\end{lemma}

\begin{proof}
Let $B = \Ball[i]{u}{\localR/2}$.
We show that there exists a candidate set $S\subseteq B$.
Recall that $B$ induces a clique in the pruned graph \prunedE, so
for any subset $S \subseteq B$, we have $B \subseteq N(S)$.
Since $B$ contains at least $\amoebaN$ nodes and has diameter at least $R$ in
each category $j \neq i$, $N(S)$ inherits these properties.
Thus, it remains to ensure that $N(S)$ has low diameter in category $i$.

We claim that \LCD implies the existence of a subset
$S\subseteq B$ of size $\NUMCAT$,
such that any two nodes in $S$ are at distance at least $2\,\prunedR$
in each category $j\neq i$.
Consider (for the proof only) the following simple algorithm.
The algorithm works with two set-valued variables, $S$ and $U$,
initialized to $S = \emptyset$ and $U = B$.
It runs the following loop $\NUMCAT$ times: pick any node
$v\in U$, add this node to $S$, and remove from $U$ all balls
$\Ball[j]{v}{2\,\prunedR}, j\neq i$.
Clearly, the following invariant is maintained after each iteration:
any two nodes $v \in S, w \in S\cup U$ are at distance at least
$2\,\prunedR$ in any category $j \neq i$.
Therefore, the algorithm finds the desired set $S$ unless $U$ were to
become empty prematurely.
This cannot happen because by \LCD,
$B$ and any $\Ball[j]{v}{2\,\prunedR}, j\neq i$ overlap in at most
$O(\log n)$ nodes, so the cardinality of $U$ decreases by at most
$O(\NUMCAT \log n)$ in each iteration.

Now fix the subset $S$ guaranteed by the previous paragraph.
Consider some node $w\in N(S)$.
For any category $j \neq i$, there can be at most one
node in $S$ within category-$j$ distance $\prunedR$ from $w$.
(If there were two such nodes $v,v'\in S$ then $\D_j(v,v')\leq \prunedR$,
a contradiction.)
It follows that at least one node $v \in S$ is at
distance more than $\prunedR$ from $w$ in \emph{each} category $j\neq i$.
Since the pruned graph $\prunedE$ contains the edge $(v,w)$,
$v$ and $w$ must be close in some category, and we have
proved that they can only be close in category $i$.
Therefore $\D_i(v,w)\leq \prunedR$.
Since $S\subseteq B$, it follows that $\D_i(u,w)\leq \prunedR + \localR/2$.
Therefore, any two nodes in $N(S)$ are at category-$i$ distance at most
$2\,\prunedR + \localR$ from one another.
\end{proof}

For each iteration $i$ of the \Amoeba Stage, we need to find a
seed clique $C$ for $\prunedE'$ such that $|C| \geq \amoebaN$ and
$\diam_j(C)\geq \log^2(n)$, for each category $j < i$.
By Lemma~\ref{lem:find-clique-fast}, one such clique is given by
$N(S)$, for any given node $u$ and some subset $S\subseteq N(u)$
of size \NUMCAT.
Therefore, we can run the original \SimpleTest to obtain the
pruned graph $\prunedE$, pick any node $u$, and iterate through all
$\NUMCAT$-node subsets $S \subseteq N(u)$ until we find a set $S$ such
that $N(S)$ is a clique in $\prunedE'$.
It is easy to see that this approach results in running time
$n \polylog(n)$. In fact, one only needs the initial pruning step to be local to node $u$, so the list of all candidate subsets $N(S)$ can be obtained in $\polylog(n)$ time.

\subsubsection{Efficient implementation of the \Amoeba step}

To implement the \Amoeba step efficiently, we use a queue which
initially contains all edges.
In each \Amoeba step, edges are popped from the queue until one is
found that satisfies Condition~\eqref{eq:amoeba-test} holds.
Once an edge $(u,v)$ satisfies this condition, it is added to the amoeba,
while all its adjacent edges are (re-)enqueued.
Any one edge is adjacent to at most polylogarithmically many other
edges, and can therefore be enqueued
at most polylogarithmically many times.
Thus the entire growth phase of the Amoeba algorithm is
implemented in $n\polylog(n)$ running time.
The following argument shows the correctness of this queue policy:
If an edge $(u,v)$ is checked and does not satisfy
Condition~\eqref{eq:amoeba-test}, then it can satisfy this condition
at some later point
in the execution of the Amoeba algorithm
only if another edge incident to $u$ or $v$ has been
added to the Amoeba, i.e., only if $(u,v)$ is re-enqueued.

\subsubsection{From a spanner to succinct distance labels}

Fix a category $i$. For the remainder of this section, all
  ``balls'' and ``distances'' refer to category $i$.
We use the spanner $\amoebaE = \amoebaE[i]$ produced by the \Amoeba
algorithm to produce \emph{distance labels} for $\D_i$ of polylogarithmic size,
so that for any two nodes $u$, $v$ the distance $\D_i(u,v)$ can be estimated with constant distortion from their labels alone (in polylogarithmic time).

Consider exponentially increasing distance scales $r$.
For each distance scale $r$, pick $k_r$
\emph{scale-$r$ beacon nodes} independently and
uniformly at random;
$k_r$ is chosen so that with high probability, each ball of radius $r$
contains $\Theta(\log n)$ scale-$r$ beacon nodes;
For each scale-$r$ beacon $b$, run a breadth-first search in
$\amoebaE$ for $\Theta(r)$ steps, to compute distance estimates
between $b$ and all nodes within distance $\Theta(r)$ from $b$.
Simple accounting shows that computing the
estimates for all scales and all beacons takes $n\polylog(n)$ time.

Thus, for every given node $u$, we have computed distance estimates
between $u$ and some subset $S_u$ of beacons.
$S_u$ includes all scale-$r$ beacons within distance $\Theta(r)$ from
$u$, for each scale $r$.
Together, these distance estimates constitute $u$'s distance label.
Given the distance labels of two nodes $u$ and $v$, one can
reconstruct the distance estimate for the pair $(u,v)$ by picking
the beacon $b \in S_u \cap S_v$ closest to node $u$, and
using the distance estimate for the pair $(b,v)$ as an estimate for
$(u,v)$.

\section{Improving the distortion for a single category}
\label{sec:additive-singleCat}

Our first improvement is to reduce the distortion from a
multiplicative constant to a factor $1+o(1)$.
In fact, under stronger assumptions on the uniformity of the metric
space, we will be able to reduce the distortion to additively
polylogarithmic. We first show the improvement for a single category,
and discuss the necessary extensions for multiple categories in
Section \ref{sec:additive-mult}.

In trying to improve the distortion beyond a multiplicative constant,
we face an immediate obstacle:
as discussed in Section~\ref{Prelims}, an algorithm can estimate the
normalization constant \CSW and the target degree \kSW only up to
a constant factor.
However, for further improvements of the distortion, more accurate
estimates of \CSW and \kSW appear to be necessary.
In order to side-step this technical obstacle, we define
\emph{normalized distances}
\begin{align}
\Dnorm(u,v) = \D(u,v)/ (\CSW\,\kSW)^{1/\DIM}, \label{eqn:normalized}
\end{align}
and we focus on \Dnorm instead of actual
distances as the quantities to be inferred.

Note that Theorem~\ref{thm:main} can also be interpreted to
yield an estimate \DnormS for \Dnorm which with high probability has
no contraction, constant expansion and $\polylog(n)$ additive error.
In this section, we improve this bound to unit
distortion with sub-linear additive error.

\begin{theorem}\label{thm:simple2ball-singleCat}
Consider a \singleSG of dimension $\DIM$, with
$\CSW \kSW = \Omega(\log n)$ and near-uniform density.
There is a polynomial-time algorithm which, with high probability,
reconstructs each normalized distance $\Dnorm(u,v)$ with additive
error $\pm \Dnorm^{\gamma} \, \log^{O(1)} n$,
    where $\gamma = \tfrac{\DIM+2}{2\DIM+2}$.
The algorithm runs in polynomial time.
\end{theorem}

The high-level idea is to augment the \SimpleTest from Section
\ref{sec:amoeba} with a post-processing step we call
\emph{\TwoBallTest.}
This is a variation of the common neighbors
heuristic where instead of common neighbors, the algorithm
counts 3-hop paths whose first and last hops are
sufficiently short according to the initial estimates.
More precisely, to estimate $\Dnorm(s,t)$, the algorithm counts
edges between two node sets \NBall{s} and \NBall{t} that are
small balls (centered at $s$ and $t$, respectively) with respect to
the initial estimates \DnormS.

The \TwoBallTest proceeds as follows. The input consists of \DnormS
and the original edge set \ESW. For every two nodes $s$ and $t$, the
normalized distance $\Dnorm(s,t)$ is estimated as follows. Let
\NRBall[\DnormS]{u}{\kappa} be the set of the $\kappa$ closest
nodes to node $u$ according to \DnormS, breaking ties arbitrarily;
note that this set is --- up to tie-breaking ---
a ball with respect to \DnormS.  Consider balls
    $\NBall{s} = \NRBall[\DnormS]{s}{\kappa}$
and
    $\NBall{t} = \NRBall[\DnormS]{t}{\kappa}$,
for some cardinality $\kappa$ to be specified later. Count the number
of edges in $\ESW$ between \NBall{s} and \NBall{t}, and let
\NEdges{s}{t} be that number. The new estimate is
\[ \DnormP(s,t) = \left(\kappa^2/\NEdges{s}{t} \right)^{1/\DIM}.\]
We take $\kappa = \ShrinkT{x}^{\DIM}$, where
    $\ShrinkT{x} \triangleq x^{(\DIM+2)/(2\DIM+2)}$ and $x = \DnormS(s,t)$.
See Algorithm~\ref{alg:2B} for the pseudocode.
%%%%%%%%%%%%%%%%%%%%%%%%%%%%%%%%%%%%%%%
\begin{algorithm}[htb]
{\bf Inputs.} Original edge set \ESW and initial estimates \DnormS from Theorem~\ref{thm:main}.

{\bf Output.} Improved distance estimates \DnormP.

{\bf For} each node pair $(s,t)$:
\begin{OneLiners}
\item[1.] $\NBall{s} = \NRBall[\DnormS]{s}{\kappa}$
and
    $\NBall{t} = \NRBall[\DnormS]{t}{\kappa}$,\\
where $\kappa = x^{\DIM(\DIM+2)/(2\DIM+2)}$ and $x = \DnormS(s,t)$.

\item[2.] \NEdges{s}{t} is the number of edges in \ESW
  between \NBall{s} and \NBall{t}.
\item[3.] $\DnormP(s,t) = (\kappa^2 /\NEdges{s}{t})^{1/\DIM}$.
\end{OneLiners}

\vspace{2mm}
{\bf Notation.} \NRBall[\DnormS]{u}{\kappa} is the set of
the $\kappa$ closest nodes to $u$ according to \DnormS,
breaking ties arbitrarily.\\

\caption{The \TwoBallTest.}
\label{alg:2B}
\end{algorithm}
%%%%%%%%%%%%%%%%%%%%%%%%%%%%%%%%%%%%%%%

The idea is that
$\Expect{\NEdges{s}{t}} \approx \kappa^2\, \Dnorm^{-\DIM}(s,t)$,
and our estimate inverts this relation. We pick $\kappa$ to optimize
the trade-off between the ``spatial uncertainty'' (the
pairwise distances between nodes in \NBall{s} and \NBall{t} are not
exactly $\Dnorm(s,t)$)
and ``sampling uncertainty'' (deviations of the number of
edges from the expectation).
The former increases with $\kappa$, while the latter decreases with $\kappa$.

\begin{proofof}[Proof of Theorem~\ref{thm:simple2ball-singleCat}]
Assume that \DnormS satisfies the high-probability property that it
is an estimate of $\Dnorm$ with constant distortion and $\polylog(n)$
additive error. Consider a node pair $(s,t)$,
at normalized distance $y = \Dnorm(s,t)$.

Assume that $\Dnorm(s,t)$ is large enough to ensure that \ShrinkT{y} is
larger than the $\polylog(n)$ additive error.
(Otherwise, the additive error guarantee is trivially satisfied.)
Then, by near-uniform density, all nodes in
\NRBall[\DnormS]{s}{\kappa} are at normalized distance at
most $c\,\ShrinkT{y}$ from $s$, for some constant $c$.
Likewise, all nodes in \NRBall[\DnormS]{t}{\kappa} are at
normalized distance at most $c\,\ShrinkT{y}$ from $t$.
Therefore
\begin{align}
\frac{\kappa^2}{(y + 2c\, \ShrinkT{y})^{\DIM}}
 \leq \Expect{\NEdges{s}{t}}
 \leq \frac{\kappa^2}{(y - 2c\, \ShrinkT{y})^{\DIM}}.
\label{eqn:expected-two-ball-edges}
\end{align}
We next apply Chernoff bounds to \NEdges{s}{t},
and use the bounds that
$\tfrac{1}{1-2\beta}\cdot (1+6\beta) \leq \tfrac{1}{1-8\beta}$ and
$\tfrac{1}{1+2\beta}\cdot (1-6\beta) \geq \tfrac{1}{1+8\beta}$
(with $\beta = c\,\tfrac{\ShrinkT{y}}{y}$)
to derive that
\begin{align*}
\Prob{\frac{\kappa^2}{(y + 8c\, \ShrinkT{y})^{\DIM}} \leq \NEdges{s}{t}
    \leq \frac{\kappa^2}{(y - 8c\, \ShrinkT{y})^{\DIM}}}
    \geq 1 - 1/{n^{O(\log n)}}.
\end{align*}
Taking the union bound over all node pairs $(s,t)$, it follows that
w.h.p.~$|(\kappa^2/\NEdges{s}{t})^{1/\DIM} - y| \leq O(\ShrinkT{y})$.
\end{proofof}

\subsection{The \recTwoBall}
\label{sec:rectwoball-single}
Given that the \TwoBallTest produces improved estimates of (normalized)
distances, it seems natural to run the algorithm again, using the
improved estimates as a starting point for defining the balls
\NBall{s} and \NBall{t} more accurately. \asedit{More precisely,} to estimate $\D(s,t)$, the algorithm can use the previously
computed estimates for smaller distance scales to define \NBall{s}
and \NBall{t}. \asedit{(The distance scales can be defined according to the coarse estimates provided by Theorem~\ref{thm:main}.)}
The technical goal is to improve the additive error in
Theorem~\ref{thm:simple2ball-singleCat}.

\asedit{Conceptually, this approach is recursive: the algorithm to estimate the distance for a given node pair $(s,t)$ recursively calls itself for smaller distance scales.
However, we present and analyse an iterative implementation, since it is more efficient computationally and somewhat easier to analyze.
We call the resulting algorithm (with carefully
optimized selection of cardinality $\kappa$ for the two balls)} the \emph{\recTwoBall.}

The analysis of this algorithm is significantly more delicate and
involved. In particular, in order to take advantage of the improved
estimates, a stronger uniformity condition is needed on the metric: we
say that the metric space has \emph{perfectly uniform density} iff
each ball of radius $r$ contains
  $\CPD\,r^{\DIM} \pm O(r^{\DIM-1})$
points, where \CPD is a known constant.
Then we can improve the additive error to $\polylog(n)$.

\begin{theorem}\label{thm:rec2ball-singleCat}
Consider a \singleSG with
$\CSW \kSW = \Omega(\log n)$ and perfectly uniform density.
Assume that the social distance is defined by the $\ell_2^{\DIM}$
norm, with $\DIM>2$.
Then, the \recTwoBall w.h.p.~reconstructs
all normalized distances with unit distortion and additive error $\polylog(n)$.
\end{theorem}

\begin{note}{Remark.}
The algorithm uses a constant \DIMCONST that captures, up to the
first-order term, how the expected number of edges between two
radius-$r$ balls depends on $r$ and the distance between centers.
Specifically, in the setting of Theorem~\ref{thm:rec2ball-singleCat},
consider two radius-$r$ balls whose centers are at distance
$x > 4r$. The expected number of edges between these two balls is
    $(\DIMCONST\, r^2/x)^{\DIM}$,
up to a multiplicative factor $1+ O(r^{-2})$. Here, \DIMCONST is a constant
that depends only on the dimension \DIM and the constant \CPD in the
definition of perfectly uniform density. We assume that \DIMCONST is known
to the algorithm.

The restriction to the $\ell_2$ norm is essential to define \DIMCONST: under
$\ell_p$, $p\neq 2$, the expected number of edges between the two balls
significantly depends on the alignment of the $s$-$t$ line relative to
the coordinate axes.

For $\DIM=2$, a similar (but slightly more complicated) algorithm and
analysis yield additive error $2^{O(\sqrt{\log x})}$ for node
pairs at normalized distance $x$; we omit the details.
\end{note}

\OMIT{ $\ell_1$ does not allow \emph{rotations} (bijective maps which fix the origin and preserve distances) other than those that map the union of coordinate axes to itself, which implies that the more precise estimates of $\Dnorm(s,t)$ would need to depend on the alignment of the $s$-$t$ line relative to the coordinate axes. (See Appendix \ref{sec:L1-rotations} for more discussion).}

We next define the algorithm. Let us first set up the notation.
Let \DnormS be the normalized distance
estimates guaranteed by Theorem~\ref{thm:main}.
We will compute refined estimates \DnormP, which are initialized to
\DnormS.
Let \NRBall[\DnormP]{u}{\kappa} be the set of the $\kappa$ closest
nodes to $u$ according to \DnormP, breaking ties
arbitrarily.
%Let $N(r)$ be the number of nodes in a ball of radius $r$ if the node
%set is the toroidal grid.

\OMIT{Note that $E(r,x) = (\DIMCONST\,r^2/x)^{\DIM}$ for some constant $\DIMCONST$; we assume that this constant is known to the algorithm.}

\OMIT{ %%%%%%%%%%%%
The algorithm is parameterized by distance scales $x_1<x_2<\ldots\,$  and an increasing function $r(x)$ such that $r(x_{i+1}) = x_i$ for each $i\geq 1$. The algorithm proceeds in stages $i=1,2,\,\ldots\,$such that in each stage $i$ it computes estimates for each node pairs $(s,t)$ with
    $x = \DnormS(s,t) \in (x_i,x_{i+1}]$.
} %%%%%%%%%%%%%%%%%

The \recTwoBall proceeds as follows.
The input consists of \DnormS and the original edge set \ESW.
The algorithm considers node pairs $(s,t)$ such that
$\DnormS(s,t)>\polylog(n)$, in order of increasing \DnormS.
For each such node pair, we define balls around $s$ and $t$
whose radius is roughly \ShrinkR{x}, where $x = \DnormS(s,t)$ and
$\ShrinkR{x} = x^{1/2+1/\DIM}$.
%is a function defined later in Equation~\eqref{eq:Rec2B-scales}.
Formally, we define balls
    $\NBall[\prime]{s} = \NRBall[\DnormP]{s}{\kappa}$
and
    $\NBall[\prime]{t} = \NRBall[\DnormP]{t}{\kappa}$,
where $\kappa = \CPD\,\ShrinkR{x}^{\DIM}$.
Note that these balls are defined with respect to the improved estimates \DnormP.
Let \NEdges{s}{t} be the number of edges between \NBall[\prime]{s} and \NBall[\prime]{t}.
The new estimate is
$\DnormP(s,t) = \DIMCONST\;\ShrinkR{x}^2\;\NEdges{s}{t}^{-1/\DIM}$.
The pseudocode is shown in Algorithm~\ref{alg:Rec2B}.

%%%%%%%%%%%%%%%%%%%%%%%%%%%%%%%%%%%%%%%%%%%
\begin{algorithm}[htb]
{\bf Inputs.} Original edge set \ESW and initial estimates \DnormS from Theorem~\ref{thm:main}.

{\bf Output.} Improved distance estimates \DnormP.

$\DnormP\leftarrow \DnormS$.

{\bf For} each node pair $(s,t)$ such that $\DnormS(s,t)>\polylog(n)$,
in order of increasing \DnormS:
\begin{OneLiners}
\item[1.] $\kappa = \CPD\,\ShrinkR{x}^{\DIM}$, where $x = \DnormS(s,t)$ and $\ShrinkR{x} = x^{1/2+1/\DIM}$.
\item[2.] $\NBall[\prime]{s} = \NRBall[\DnormP]{s}{\kappa}$
and
    $\NBall[\prime]{t} = \NRBall[\DnormP]{t}{\kappa}$.
\item[3.] \NEdges{s}{t} is the number of edges in \ESW between \NBall[\prime]{s} and \NBall[\prime]{t}.
\item[4.] $\DnormP(s,t) = \DIMCONST\;\ShrinkR{x}^2\;\NEdges{s}{t}^{-1/\DIM}$.
\end{OneLiners}

\vspace{2mm}
{\bf Notation.}
%$N(r)$ is the number of nodes in a ball of radius $r$ for the toroidal grid.\\
\NRBall[\DnormP]{u}{\kappa} is the set of the $\kappa$
closest nodes to node $u$ according to \DnormP, breaking ties
arbitrarily. \\
\DIMCONST is the constant from the remark after
Theorem~\ref{thm:rec2ball-singleCat}.

\caption{The \recTwoBall.}
\label{alg:Rec2B}
\end{algorithm}

\asedit{The algorithm is quite simple: as in the \TwoBallTest, the balls \NBall[\prime]{s} and \NBall[\prime]{t} are defined via cardinality $\kappa$, and the improved distance estimate $\DnormP(s,t)$ is computed as a function of the number of edges between the two balls. The only complication is how
to pick $\kappa$ as a function of the initial distance estimate $x=\DnormS(s,t)$.}

\subsection{Recursive nature of the algorithm}
\label{app:rec2ball-prelim}

\ascomment{Note to review team: this subsection is new}

To illustrate the recursive nature of the algorithm, we argue that the improved distance estimates for distance scale $x$ depend only on those for distance scale $O(\ShrinkR{x})$, where the distance scales are defined according to the initial estimates \DnormS.

\begin{lemma}\label{cl:rec2ball-prelim}
In Algorithm~\ref{alg:Rec2B}, fix the original edge set \ESW and initial estimates \DnormS from Theorem~\ref{thm:main}. Let $C$ be the (constant) expansion in Theorem~\ref{thm:main}. Consider a node pair $(s,t)$ with $x=\DnormS(s,t)>\polylog(n)$. Then the improved distance estimate $\DnormP(s,t)$ depends only on
    $\DnormP(u,v)$ for node pairs $(u,v)$ with $\DnormS(u,v) \leq 8C\,\ShrinkR{x}$.
\end{lemma}

In other words, the improved distance estimates implicitly rely
on recursion from distance scale $x$ to distance scale \ShrinkR{x}.
Let $\rho(x)$ be the depth of this recursion: the number of steps until the
distance scale goes below $\polylog(n)$.
It is easy to see that $\rho(x) = O(\log \log n)$.

As an auxiliary step to prove Lemma~\ref{cl:rec2ball-prelim}, we claim
that Algorithm~\ref{alg:Rec2B} w.h.p.~reconstructs the normalized
distances with constant multiplicative distortion and $\polylog(n)$
additive error.

\begin{claim}\label{cl:rec2ball-prelim-multiplicative}
For each node pair $(u,v)$ with $\DnormS(u,v) \geq \polylog(n)$, with
high probability,
\begin{align*}%\label{eq:cl:rec2ball-prelim-multiplicative}
\tfrac12 \Dnorm(u,v) \leq \DnormP(u,v) \leq 2\,\Dnorm(u,v).
\end{align*}
\end{claim}

We prove this claim by induction on $\DnormS(u,v)$; the inductive step
is proved using an argument similar to the proof of
Theorem~\ref{thm:simple2ball-singleCat}; the easy details are omitted.

\begin{proofof}[Proof of Lemma~\ref{cl:rec2ball-prelim}]
Recall that $\DnormP(s,t)$ depends only on which nodes comprise the two balls
    $\NBall[\prime]{s}$ and $\NBall[\prime]{t}$.
Let us focus on the membership in $\NBall[\prime]{s}$.
(For $\NBall[\prime]{t}$, we argue similarly.)
\asedit{
Note that
\[
\NBall[\prime]{s}
   \; = \;\NRBall[\DnormP]{s}{\kappa}
   \; \subseteq \; \Ball{u}{2\,\ShrinkR{x}}
   \; \subseteq \; \{ u\in V:\; \DnormP(s,u)\leq 4\, \ShrinkR{x} \},
\]
where the first inclusion followed from perfectly uniform density and
the second from Claim~\ref{cl:rec2ball-prelim-multiplicative}.
Now, consider a node $u$ such that $\DnormS(s,u) > 8C\,\ShrinkR{x}$.
By Theorem~\ref{thm:main},  $\DnormS(s,u) > 8\,\ShrinkR{x}$, and by
Claim~\ref{cl:rec2ball-prelim-multiplicative},
this term is in turn bounded below by
$4\,\ShrinkR{x}$.
Therefore, we conclude that $u\not\in \NBall[\prime]{s}$.}
\end{proofof}

\subsection{Proof of Theorem~\ref{thm:rec2ball-singleCat}}
\label{app:rec2ball-singleCat-proof}

The high-level idea of the analysis is as follows.
Let \AddDist{x} be the maximum additive error for node pairs at
normalized distance at most $x$.
As in the \SimpleTest, the error comes from two sources: spatial
uncertainty and sampling uncertainty.
We show that the spatial uncertainty can contribute at most
$O(\AddDist{\ShrinkR{x}})$ to the overall additive error;
interestingly, this holds for any choice of \ShrinkR{x}.
We use Chernoff Bounds to bound the contribution of sampling
uncertainty by $O(\AddDist{\ShrinkR{x}})$ as well;
this is where the particular exponent in \ShrinkR{x} is used.
It follows that
    $\AddDist{x} = O(\AddDist{\ShrinkR{x}})$.
Finally, \asedit{by Lemma~\ref{cl:rec2ball-prelim}}, the distance
estimates for a given node pair implicitly rely on recursion from
distance scale $x$ to distance scale \ShrinkR{x}.
Let
    $\rho(x) \asedit{ = O(\log\log n)}$
be the depth of this recursion: the number of steps until the
distance scale goes below $\polylog(n)$.
It is easy to see that
    $\AddDist{x} = 2^{O(\rho(x))} \asedit{ = \polylog(n)}$.

%\dkdelete{Let $\AddDist{x}$ be the maximum additive distortion
%of the algorithm's distance estimates for node pairs whose
%distances are at most $x$.}

Consider two nodes $s$ and $t$ whose normalized distance is
    $x = \Dnorm(s,t)$.
%Recalling that $\ShrinkR{x} = x^{1/2 + 1/\DIM}$,

Let $\NDBall{s} = \NRBall[\Dnorm]{u}{\kappa}$ and
$\NDBall{t} = \NRBall[\Dnorm]{u}{\kappa}$ be the
sets of the $\kappa$ closest nodes to $s$ and $t$, respectively, under
the (correct) normalized distances $(V,\Dnorm)$.

%$\Ball[\Dnorm]{s}{\ShrinkR{x}}$ and $\NDBall{t} = \Ball[\Dnorm]{t}{\ShrinkR{x}}$
%be the balls of radius \ShrinkR{x} around $s$ and $t$, respectively, under
%the (correct) normalized distances $(V,\Dnorm)$.
%By contrast, recall that \NBall[\prime]{s} and \NBall[\prime]{t} are balls under the
%\emph{estimated} distances $(V,\DnormP)$, whose number of nodes was
%chosen to equal the number of nodes in \NDBall{s} and \NDBall{t}, respectively.
We start with a simple lemma showing that this choice implies that
the actual \emph{sets} of nodes are very close between
\NBall[\prime]{s} and \NDBall{s} (and \NBall[\prime]{t} and \NDBall{t}, respectively).

\begin{lemma} \label{lem:ball-containment}
For a sufficiently large constant $\beta$, we have that
\begin{align*}
\Ball[\Dnorm]{s}{\ShrinkR{x} - 2 \AddDist{\ShrinkR{x}} - \beta}
& \subseteq \NBall[\prime]{s}
\subseteq \Ball[\Dnorm]{s}{\ShrinkR{x} + 2 \AddDist{\ShrinkR{x}} + \beta},\\
\Ball[\Dnorm]{t}{\ShrinkR{x} - 2 \AddDist{\ShrinkR{x}} - \beta}
& \subseteq \NBall[\prime]{t}
\subseteq \Ball[\Dnorm]{t}{\ShrinkR{x} + 2 \AddDist{\ShrinkR{x}} + \beta}.
\end{align*}
%All nodes at distance at most $\ShrinkR{x} - \AddDist{\ShrinkR{x}}$
%from $s$ ($t$ resp.) in \NDBall{s} (\NDBall{t})
%appear in \NBall[\prime]{s} (\NBall[\prime]{t}) as well.
\end{lemma}

\begin{proof}
We first prove the first inclusion.
Let $v \in \Ball[\Dnorm]{s}{\ShrinkR{x} - 2\AddDist{\ShrinkR{x}} - \beta}$
be arbitrary.
Because $\Dnorm(s,v) \leq \ShrinkR{x} - 2\AddDist{\ShrinkR{x}} - \beta$,
the definition of \AddDist{\cdot} implies that
$\DnormP(s,v) \leq \Dnorm(s,v) + \AddDist{\ShrinkR{x}}
\leq \ShrinkR{x} - \AddDist{\ShrinkR{x}} - \beta$.
On the other hand,
$\DnormP(s,u) \geq \ShrinkR{x} - \AddDist{\ShrinkR{x}} - \beta$
for all nodes $u$ such that $\Dnorm(s,u) \geq \ShrinkR{x} - \beta$.
Therefore, there can be at most
    $\CPD\,(\ShrinkR{x}-\beta)^{\DIM} \pm O((\ShrinkR{x}-\beta)^{\DIM-1})$
nodes $u$ with
    $\DnormP(s,u) \leq \DnormP(s,v)$.
This number is less than
    $\CPD\,\ShrinkR{x}^{\DIM} = \kappa$
whenever $\beta$ is large enough.

Because \NBall[\prime]{s} contains the $\kappa$ nodes closest to $s$ under
\DnormP (by its definition), this means that $v \in \NBall[\prime]{s}$.
Since this argument holds for arbitrary $v$, we have proved the first
claim.
The second inclusion is proved by an analogous calculation.
\end{proof}

We next show that the number of edges between \NDBall{s} and \NDBall{t} is close
to the number of edges between \NBall[\prime]{s} and \NBall[\prime]{t}.
To state this claim concisely, let
$\NumEdges(S,S')$ be the number of edges in \ESW between node
sets $S$ and $S'$.

\begin{lemma}
\label{lem:edge-count}
With high probability,
\begin{align*}
\left| \Expect{\NumEdges(\NBall[\prime]{s},\NBall[\prime]{t})} -
\Expect{\NumEdges(\NDBall{s},\NDBall{t})} \right|
= O(x \cdot \AddDist{\ShrinkR{x}}).
\end{align*}
\end{lemma}

\begin{proof}
We construct a bijection
    $\phi: (\NBall[\prime]{s} \cup \NBall[\prime]{t}) \to (\NDBall{s} \cup \NDBall{t})$
as follows. Partition the domain and the co-domain into four disjoint regions each
(using $\oplus$ to denote the disjoint union of sets):
\begin{align*}
(\NBall[\prime]{s} \cup \NBall[\prime]{t}) &=
(\NBall[\prime]{s} \cap \NDBall{s})        \oplus
(\NBall[\prime]{s} \setminus \NDBall{s})   \oplus
(\NBall[\prime]{t} \cap \NDBall{t})        \oplus
(\NBall[\prime]{t} \setminus \NDBall{t}),\\
(\NDBall{s} \cup \NDBall{t}) &=
(\NBall[\prime]{s} \cap \NDBall{s})        \oplus
(\NDBall{s} \setminus \NBall[\prime]{s})   \oplus
(\NBall[\prime]{t} \cap \NDBall{t})        \oplus
(\NDBall{t} \setminus \NBall[\prime]{t}).
\end{align*}
The regions in each partition are indeed disjoint because
    $\NDBall{s} \cap \NDBall{t} = \NBall[\prime]{s} \cap \NBall[\prime]{t} = \emptyset$.
We define $\phi$ separately for each of the four subsets the domain.
First, any node in
    $(\NBall[\prime]{s} \cap \NDBall{s})$ or $(\NBall[\prime]{t} \cap \NDBall{t})$
is mapped to itself. Second, $\phi$ is an arbitrary bijection
    $(\NBall[\prime]{s} \setminus \NDBall{s}) \to (\NDBall{s} \setminus \NBall[\prime]{s})$
and
    $(\NBall[\prime]{t} \setminus \NDBall{t}) \to (\NDBall{t} \setminus \NBall[\prime]{t})$.
This completes the definition. For the second step, note that the respective domains and co-domains have the same size;
this is because
    $|\NDBall{s}| = |\NBall[\prime]{s}| = \kappa $ and $|\NDBall{t}| = |\NBall[\prime]{t}| = \kappa$.

Nodes $v \in (\NBall[\prime]{s} \setminus \NDBall{s}) \cup (\NBall[\prime]{t} \setminus \NDBall{t})$
called \emph{perturbed nodes}. By Lemma \ref{lem:ball-containment}, \NBall[\prime]{s} and \NBall[\prime]{t}
contain at most
  $\CPD \cdot (2 \AddDist{\ShrinkR{x}} + \beta)
    \cdot \ShrinkR{x}^{\DIM-1}$
perturbed nodes each.

By the perfectly uniform density assumption, at least
$\CPD\,r^{\DIM} - O(r^{\DIM-1})$ nodes have distance at most
$r$ from $s$. In particular, setting $r=\ShrinkR{x} + \beta$
gives us that at least $\kappa$ nodes satisfy the distance bound,
implying that every node $u \in \NDBall{s}$ satisfies
$\Dnorm(s,u) \leq  \ShrinkR{x} + \beta$,
Furthermore, by the second inclusion of Lemma \ref{lem:ball-containment},
every node $v \in \NBall[\prime]{s}$ satisfies
$\Dnorm(s,v) \leq  \ShrinkR{x} + 2 \AddDist{\ShrinkR{x}} + \beta$.
Similar bounds apply for $t$.
We thus get that $\Dnorm(v,\phi(v)) \leq 2\ShrinkR{x} + 2
\AddDist{\ShrinkR{x}} + O(1) < 3\ShrinkR{x}$
for all $v$, and of course $\Dnorm(v,\phi(v)) = 0$ for unperturbed nodes $v$.

Now consider a pair $u \in \NBall[\prime]{s}$ and $v \in \NBall[\prime]{t}$ such that
at least one of $u,v$ is perturbed.
(We call such a pair a \emph{perturbed pair}.)
By triangle inequality,
  $|\Dnorm(v,u) - \Dnorm(\phi(v), \phi(u))| \leq 6\ShrinkR{x}$,
and the number of perturbed pairs is at most
  $(4 \AddDist{\ShrinkR{x}} + 2\beta) \cdot \ShrinkR{x}^{2\DIM-1}$,
by the bound on the number of perturbed nodes.

Next, we bound how much a single perturbed pair
$u \in \NBall[\prime]{s}, v \in \NBall[\prime]{t}$
affects the expected number of edges between the balls.
Because
$x + 6 \ShrinkR{x} \geq \Dnorm(\phi(u), \phi(v))
  \geq x - 6\ShrinkR{x}$, we get that
\[\frac{\Dnorm(u,v)}{\Dnorm(\phi(u),\phi(v))} \in 1 \pm O(\ShrinkR{x}/x).\]
We can now express the difference between the probabilities
of the edges $(\phi(u),\phi(v))$ and $(u, v)$ as
\begin{eqnarray*}
\left| (x \pm O(\ShrinkR{x}))^{-\DIM} - (x \pm 2\ShrinkR{x})^{-\DIM} \right|
& = &
x^{-\DIM} \cdot \left| \left(1 \pm \frac{O(\ShrinkR{x})}{x}\right)^{-\DIM}
- (1 \pm \frac{2\ShrinkR{x}}{x})^{-\DIM} \right| \\
& = &
O\left(x^{-\DIM} \cdot \left(\frac{1}{1 \pm O(\ShrinkR{x}/x)} - \frac{1}{1 \pm 2\ShrinkR{x}/x}\right)\right)\\
& = &
O\left(x^{-\DIM} \cdot \ShrinkR{x}/x\right).
\end{eqnarray*}
In the second step, we truncated the Binomial expansion
(because $\ShrinkR{x}/x = o(1/\DIM)$), and the final step again used that
$\ShrinkR{x}/x$ is small.
Summing over all perturbed pairs, the total expected difference in the
number of edges can be bounded by above as follows:
\[ \left|\E\left[\,
    \NumEdges(\NBall[\prime]{s},\NBall[\prime]{t}) - \NumEdges(\NDBall{s},\NDBall{t})\,
\right]\right|
\leq O\left(\frac{\ShrinkR{x}}{x} \cdot x^{-\DIM} \cdot \AddDist{\ShrinkR{x}} \cdot \ShrinkR{x}^{2\DIM-1}\right)
= O(x \,\AddDist{\ShrinkR{x}}), \]
where the last step was obtained by substituting the definition of $\ShrinkR{x}$.
The concentration now follows from Chernoff Bounds.
\end{proof}

\begin{lemma} \label{lem:adddist}
$\AddDist{x} = O(\AddDist{\ShrinkR{x}})$.
\end{lemma}

\begin{proof}
Consider two nodes $s$ and $t$ at normalized distance $x$.
Using an analysis very similar to the one in the proof of
Lemma~\ref{lem:edge-count}, the expected number of edges between
\NDBall{s} and \NDBall{t} is
$(\DIMCONST\, \ShrinkR{x}^2/x)^{\DIM} \pm O(x)
= \DIMCONST^{\DIM} x^2 \pm O(x)$
(where \DIMCONST is the constant from the remark after
Theorem~\ref{thm:rec2ball-singleCat}).
\ShrinkR{x} is chosen so that Chernoff Bounds ensure that w.h.p., the
actual number of edges between \NDBall{s} and \NDBall{t} does not deviate from
its expectation by more than $O(x \cdot \AddDist{\ShrinkR{x}})$.
Combining this number of edges with the bound from Lemma \ref{lem:edge-count},
the expected number of edges between \NBall[\prime]{s} and \NBall[\prime]{t} is
$\NEdges{s}{t} = \DIMCONST^{\DIM} x^2
                   \pm O(x \cdot \AddDist{\ShrinkR{x}})$
with high probability.
(The big-$O$ term combines both the misestimates bounded by the
Chernoff Bound and the ones from Lemma \ref{lem:edge-count}.)

Because the algorithm estimates the distance as
$\DIMCONST\, \ShrinkR{x}^2 \NEdges{s}{t}^{-1/\DIM}$,
the additive distortion is at most
\begin{eqnarray*}
\left|x - \DIMCONST\, \ShrinkR{x}^2\, \NEdges{s}{t}^{-1/\DIM}\right|
& = &
x \cdot \left|1 - \frac{\DIMCONST\, x^{2/\DIM}}{%
(\DIMCONST^{\DIM}\, x^2 \pm O(x \AddDist{\ShrinkR{x}}))^{1/\DIM}}  \right|\\
& = &
x \cdot \left|1 - \left(\frac{\DIMCONST^{\DIM}\,x}{\DIMCONST^{\DIM}\, x \pm O(\AddDist{\ShrinkR{x}})}\right)^{1/\DIM}\right|\\
& = &
x \cdot \left|1 - \left(1 \pm \frac{O(\AddDist{\ShrinkR{x}})}{\DIMCONST^{\DIM}\, x \pm O(\AddDist{\ShrinkR{x}})}\right)^{1/\DIM}\right|\\
& \leq &
x \cdot \frac{O(\AddDist{\ShrinkR{x}})}{\DIMCONST^{\DIM}\, x \pm O(\AddDist{\ShrinkR{x}})}\\
& \leq &
O(\AddDist{\ShrinkR{x}}).
\end{eqnarray*}
In the penultimate inequality, we used that
$|1-(1\pm \delta)^{1/\DIM}| \leq \delta$ for any $\delta$, and the
final inequality used that $\AddDist{\ShrinkR{x}} = o(x)$ to simplify
the denominator.
\end{proof}

\asedit{To complete the proof of Theorem~\ref{thm:rec2ball-singleCat}, we consider the recursion depth $\rho(x) = O(\log\log n)$ as described in the first paragraph of this subsection, and observe that
    $\AddDist{x} = 2^{O(\rho(x))} \asedit{= \polylog(n)}$.}

% Alex: removed this (7/3/2014) because this repeats the text from the first para of this subsection.
%The distance estimates for a given node pair implicitly rely on
%recursion from distance scale $x$ to distance scale $\ShrinkR{x}$. Let
%$\rho(x)$ be the depth of this recursion: the number of steps until the
%distance scale goes below $\polylog(n)$.
%It is easy to see that $\AddDist{x} = 2^{O(\rho(x))}$ and that
%$\rho(x) = O(\log \log n)$.
%This completes the proof of Theorem~\ref{thm:rec2ball-singleCat}.

%%%%%%%%%%%%%%%%%%%%%%%%%%%%% Deleted Parts below  %%%%%%%%%%%%%%%%%%%%

%Next, we prove Lemma \ref{lm:RecTest}.
\OMIT{ %%%%%%%
\begin{lemma}\label{lm:RecTest}
Consider the SWG. Assume that the social distance is $\D=\ell_2^{\DIM}$, $\DIM\geq 2$.
Given a constant-distortion estimates $\DnormS(s,t)$, with probability at least $1-1/n$, for each
    $x \geq \polylog(n)$,
the \recTwoBall reconstructs each normalized distance $\D(s,t)=x$:
\begin{enumerate}
\item ($\DIM=2$) with additive distortion $\distortion_{x} = O(2^{O(\sqrt{\log{x}})})$.
\item ($\DIM>2$) with additive distortion  $O(\log^{O(1)}{n})$.
\end{enumerate}
\end{lemma}

\begin{proof}
In order to bound the additive distortion, we need to bound $i(x)$.
We consider the cases $\DIM=2$ and $\DIM > 2$ separately.

For $\DIM=2$, we show that $i(x) = \sqrt{\log x}$,  this will give us the desired stretch.
Let $r_i = \ShrinkR{x}$, $r_{i-1} = r(\ShrinkR{x})$ and so on.
And let $\epsilon_i = r_i / r_{i-1}$
Note that
\begin{align*}
x <  1/\epsilon_1 \cdot 1/\epsilon_2 \cdot... \cdot 1/\epsilon_i  \cdot x_1 <  \prod_{1 \leq  i \leq j-1}{(\alpha^{1/2})^{i}} = O({(\alpha^{1/2})^{i(x)^2}}),
\end{align*}
and we conclude that $j = O(\sqrt{x})$. This gives the desired additive distortion.

We now turn to analyze the \recTwoBall when the dimension is greater than 2.
For $\DIM> 2$, we get
$1 / \epsilon_j > x^{(1/2-1/\DIM)^{i(x)- j + 1}}$, we get that $i(x) = O(\log\log(x))$.
The lemma follows.
\end{proof}

For perfectly uniform density the algorithm is the same.
The only difference is in the analysis of Lemma \ref{lem:edge-count}.
Note that the difference in the expectation on the number of edges between two balls whose centers are at distance $x$ in the canonical SWG vs.
the perfectly uniform density is at most $O(1/x)$.
Therefore the difference in the expectation on the number of edges between \NDBall{s} and \NDBall{t} vs. \NBall[\prime]{s} and \NBall[\prime]{t}
in the perfectly uniform density stays $O(\AddDist{\ShrinkR{x}}/x)$ and the rest of the analysis is the same.

}

%\OMIT{ %%%%%%%
%It remains to define $\ShrinkR{x}$. We define increasing distance scales
%    $\{x_i \}_{i\in \N}$,
%for each distance scale $x_i$ we define the corresponding radius $r_{i,\,x}$ and we set $\ShrinkR{x}=r_{i,\,x}$ for $i$ such that
%    $x\in (x_i,x_{i+1}]$.
%Specifically, for the appropriately chosen constant $\alpha$ we define $\ShrinkR{x}$ as follows:

%\begin{align}
%\begin{equation} \begin{array}{lcl}
%\begin{cases}
%    r_{i,\,x}  &=& \frac{x^{1/2+1/\DIM}}{\alpha^{i/\DIM}}. \\
%    x_1     &=& \polylog(n). \\
%    x_{i+1} &=& \min\{x:\; r_{i,\,x} = x_i\},\; \text{for } i\geq 1. \\
%    \ShrinkR{x}    &=& r_{i,\,x}, \text{ where } x\in (x_i,x_{i+1}].
%\end{cases}
%\end{array} \label{eq:Rec2B-scales} \end{equation}
%\end{align}
%}
%%%%%%%%%%%%%%%%%%%%%%%%%%%%%%%%%%%%%%%

\OMIT{ %%%%%%
The high level idea of the proof is as follows.
In the first step, we show that the additive distortion for nodes at distance $x$ is some constant $\alpha$ times the additive distortion on $\ShrinkR{x}$, namely $\AddDist{x} = \alpha \AddDist{\ShrinkR{x}}$.
Note that the additive distortion $\AddDist{x}$ comes from two sources.
The first is from the fact that we don't know \NDBall{s} and \NDBall{t} but rather we have some estimations \NBall[\prime]{s} and \NBall[\prime]{t} on them.
The second source comes from the fact that even if we would have the balls \NDBall{s} and \NDBall{t}, the number of edges observed between these balls do not necessarily match the expectation.
To handle the first error, we show that \NBall[\prime]{s} and \NBall[\prime]{t} are ``close'' enough to \NDBall{s} and \NDBall{t}.
More specifically, we show the additive error from the difference of
the expectation on the number of edges between balls \NBall[\prime]{s}
and \NBall[\prime]{t} vs.~the balls \NDBall{s} and \NDBall{t} is $O(\AddDist{\ShrinkR{x}})$.
To handle the second error, we show by applying Chernoff bound that w.h.p.~the number of edges observed between the balls is ``close'' to the expectation.
Namely, w.h.p.~the additive error coming from the second source is also $O(\AddDist{\ShrinkR{x}})$.
From the first step we conclude that $\AddDist{x} = \alpha^{i(x)} \AddDist{x_1}$,
where $i(x)$ is the number of times we need to invoke $r$ on $x$ recursively until getting a number smaller than $x_1$.
In the second step of the analysis, we bound the number of iterations $i(x)$ and conclude the desired additive distortion.
} 
\section{Improving the distortion for multiple categories}
\label{sec:additive-mult}

\asedit{We next improve distortion from a multiplicative constant to
  $1+o(1)$ for multiple categories as well.}
We employ
the two algorithms from Section~\ref{sec:additive-singleCat}.
The main difference with the single-category case is that when we
count the number of edges between the balls in the original \multiSG
for some category $i$, some of these edges may come from other
categories, which might affect the estimation.
We would like to claim that the number of edges from other categories
between the two balls is small compared to the number of
edges from category $i$.
Unfortunately, such a claim does not follow from the
\LCD condition, which prompts the following stronger condition.

The stronger condition, called \globalCD[\GCDscale], states that at
all scales up to \GCDscale, categories look essentially ``random''
with respect to one another. More specifically, given a pair of balls
$B$, $B'$ in some category $i$, we count the number of node pairs
$(u,u')$, $u\in B$, $u'\in B'$ such that $u$ and $u'$ are close in
some other category $j$:
\begin{align}
\pairs_j(B,B',r)
    \triangleq |\Set{(u,u')}{u\in B,\, u'\in B',\, \D[j](u,u') < r}|.
\end{align}
If the node identifiers within each category are permuted randomly,
then the expected number of such node pairs is
    $\Theta(r^{\DIM}/n) \cdot |B|\, |B'|$,
and with high probability, the deviations are bounded by:
\begin{align}
\pairs_j(B,B',r)
\leq  O(r^{\DIM}/n) \cdot |B|\, |B'|\, + O(\log^2 n).
\label{eqn:globalCD}
\end{align}
\globalCD[\GCDscale] asserts that~\eqref{eqn:globalCD} holds
``locally:'' at all distance scales up to \GCDscale.

\begin{definition}\label{def:globalCD}
The \emph{\globalCD[\GCDscale]{}} condition states
that~\eqref{eqn:globalCD} holds for any two categories $i \neq j$, any
two disjoint category-$i$ balls $B$, $B'$ with
$|B| \cdot |B'| \leq \GCDscale^{\DIM}$,
and any $r \in (0,\GCDscale]$.
\end{definition}

\begin{note}{Remark.}
Equation~\eqref{eqn:globalCD} for randomly permuted categories is derived in Section~\ref{sec:permutations}. The expectation is relatively easy to derive, whereas the high-probability guarantee requires a more careful analysis.
We obtain (a slightly weaker version of) \LCD as a special case if $\GCDscale = \polylog(n)$ and $B$ is restricted to be a single node; \asedit{specifically, we obtain Equation~\eqref{eq:cat-disj} with right-hand side $O(\log^2 n)$.}
\end{note}

\OMIT{ %%%%%%%%%%%%%%%
\begin{definition}\label{def:globalCD}
The \emph{\globalCD[\GCDscale]} condition states that for any two
categories $i \neq j$, any two disjoint
category-$i$ balls $B$, $B'$ with
$|B| \cdot |B'| \leq \GCDscale^{\DIM+1-1/(\DIM+1)}$,
and any $r \in (0,\GCDscale^{1+1/(\DIM+1)}]$:
\begin{align}
\pairs_j(B,B',r)
 & \triangleq |\Set{(u,u')}{u\in B,\, u'\in B',\, \D[j](u,u') < r}|
\leq  O(r^{\DIM}/n) \cdot |B|\, |B'|\, + O(\log^2 n).
\label{eqn:globalCD}
\end{align}
When the parameter \GCDscale is omitted, we require that
\eqref{eqn:globalCD} hold for all $r$ and all balls $B, B'$,
and term the condition \emph{\globalCD.}
%\item The \emph{\LCDgen} condition (with radius $\LCDrad=\polylog(n)$)
%states that the bound \eqref{eqn:globalCD} holds whenever $B$ and $B'$
%have radius at most \LCDrad and $r \leq \LCDrad$.
%In that case, \eqref{eqn:globalCD} simplifies to
%\begin{align*}
%\pairs_j(B,B',\LCDrad) & = O(\log^2 n).
%\end{align*}
%\end{enumerate}
\end{definition}

\begin{note}{Remark.}
\globalCD (and thus also \globalCD[\GCDscale] for any \GCDscale) is ensured
w.h.p.~when the categories are randomly permuted, in the sense of
Lemma~\ref{lm:BasicLCD};
see Lemma~\ref{lm:globalCD} for a precise statement and proof.
Indeed, the expected number of pairs of nodes from $B$ and $B'$ at
distance $r$ under a random permutation is
$\Theta(r^{\DIM}/n) \cdot |B|\, |B'|$.
However, deriving the high-probability guarantee requires a more
careful analysis.
Also note that (a slightly weaker version of) \LCD is obtained
as a special case of \globalCD[\polylog(n)] when $B$ is restricted
to be a single node.
\end{note}
} %%%%%%%%%%%%%%

\OMIT{%%%%%%
Recall that the constant distortion in Theorem~\ref{thm:main} contains
a non-trivial absolute constant and moreover depends on the
parameters of our model. Our results are phrased to reduce this
constant to a specific parameter which describes how well the social
distance is approximated by its own unit-disk graph.

\begin{definition}
Given a metric space $(V,\D)$, let \Dsp be the shortest-paths metric of
the unit-disk graph. The \emph{shortest-path distortion} is defined as
    $\max_{u,v\in V} \Dsp(u,v)/\D(u,v)$.
\end{definition}
} %%%%%%%

We will improve over the constant distortion under the condition
above. We present two results: an extension of the \TwoBallTest
(Section~\ref{sec:extended-two-ball}) and an analysis of the
\recTwoBall for multiple categories
(Section~\ref{sec:recTwoBall-multicat}).

Like in the single-category case, we focus on normalized distances.
For each category $i$, let \CSW[i] and \kSW[i]
be the normalization constant and the target degree, respectively.
The \emph{normalized} category-$i$ distance between nodes $u,v\in V$ is
$\Dnorm[i] (u,v) \triangleq \D[i](u,v)/ (\CSW[i]\,\kSW[i])^{1/\DIM}$.
%Our goal is to estimate \Dnorm[i] with distortion that is close to $1$.
% the shortest-path distortion for \Dnorm[i].

\subsection{The \extTwoBall}
\label{sec:extended-two-ball}

The \globalCD[\GCDscale] condition does not apply to distance scales
beyond \GCDscale, and even for $\GCDscale=\infty$, the guarantee of Equation
\eqref{eqn:globalCD} is quite weak at very large scales.
Accordingly, we find that the \TwoBallTest becomes problematic at
large distance scales.
To deal with these issues, we apply the \TwoBallTest only to
distance scales small enough to provide strong guarantees.
The improved distance estimates define edge lengths,
% for the corresponding node pairs,
and a post-processing step computes shortest paths with respect to
these edge lengths.
The resulting algorithm, called \emph{\extTwoBall,} satisfies the
following theorem.

\begin{theorem}\label{thm:extTwoBall}
Assume the setting of Theorem~\ref{thm:main} with
    \globalCD[\GCDscale^{1+1/(\DIM+1)}],
$\GCDscale\geq \polylog(n)$ for a sufficiently large $\polylog(n)$.
Then, the \extTwoBall runs in polynomial time,
and with high probability produces distance estimates \DnormP[i] with
the following guarantee:
\begin{quote}
For any pair $(s,t)$ at normalized distance $x = \Dnorm[i](s,t)$, the estimate
$\DnormP[i](s,t)$ has multiplicative distortion
   $1 \pm \left[
        (\min(x,\GCDscale,\hat{\GCDscale}))^{-\DIM/(2\DIM+2)}
        \cdot O(\log^2 n)
   \right]$,
where $\hat{\GCDscale} = \left(\frac{n}{\log n}\right)^{(2\DIM+2)/(2\DIM^2+3\DIM)}$.
\end{quote}
\end{theorem}

\begin{note}{Remark.}
The distortion in Theorem~\ref{thm:extTwoBall} can be interpreted as
    $1\pm O\left( \, \ell^{-\DIM/(2\DIM+2)} \cdot \log^2 n \right)$,
where $\ell = \min(x,\GCDscale,\hat{\GCDscale})$ is, in some sense, the effective distance scale.
\end{note}

%\begin{align}\label{eq:scaleR}
%$\GCDscale \in [\polylog(n), (n/\log n)^{(2\DIM+2)/(2\DIM^2+3\DIM)}]$
%\end{align}

We begin by defining the \extTwoBall precisely.
The input consists of the \multiSG and the distance estimates
$\DnormS = \DnormS[i]$ for a given category $i$,
as guaranteed by Theorem~\ref{thm:main}.
Recall that these are non-contracting estimates with constant
expansion \EXPAN and $\polylog(n)$ additive error;
we assume that (an upper bound on) \EXPAN
is known to the algorithm. Apart from \EXPAN, the algorithm is
parameterized by the distance scale \GCDscale from
Theorem~\ref{thm:extTwoBall}.

The algorithm proceeds as follows.
(See Algorithm~\ref{alg:ext2B} for the pseudocode).
It focuses on the edge set
    $H  = \Set{(u,v)}{\DnormS(u,v) \leq \GCDscale}$.
For each edge $(u,v) \in H$, it applies the \TwoBallTest
with respect to distances \DnormS to obtain improved distance
estimates $\Dnorm[H](u,v)$.
These improved estimates are treated as edge lengths for $H$.
For each node pair $(s,t)$, we distinguish two cases.
If the edge $(s,t)$ is in $H$, we simply set
the final estimate $\DnormP[i](s,t) = \Dnorm[H](s,t)$.
Otherwise, the final distance estimate $\DnormP[i](s,t)$ is the length of the
shortest $s$-$t$ path using the edge set
\begin{align} \label{eq:ext-two-ball-Ht}
H_t = \Set{(u,v)\in H}{\DnormS(u,v) \geq \tfrac{\GCDscale}{2\EXPAN} \text{ or $v=t$}}.
\end{align}
In other words, the distance is estimated by the length of the
  shortest path using only ``sufficiently long'' edges, except for
  possibly the last edge, which may be short.

%%%%%%%%%%%%%%%%%%%%%%%%%%%%%%%%%%%%%%%
\begin{algorithm}[htb]
{\bf Inputs.} Original edge set \ESW and initial estimates $\DnormS = \DnormS[i]$ from Theorem~\ref{thm:main}.

{\bf Parameters.} Distance scale \GCDscale and expansion \EXPAN of \DnormS.

{\bf Output.} Improved distance estimates \DnormP[i].\vspace{2mm}

$H = \Set{(u,v)}{\DnormS(u,v) \leq \GCDscale}$. \vspace{2mm}

{\bf The \TwoBallTest}. For each node pair $(s,t)\in H $,
\begin{OneLiners}
\item[1.] $\NBall{s} = \NRBall[\DnormS]{s}{\kappa}$
and
    $\NBall{t} = \NRBall[\DnormS]{t}{\kappa}$,\\
where $\kappa = x^{\DIM(\DIM+2)/(2\DIM+2)}$ and $x = \DnormS(s,t)$.

\item[2.] \NEdges{s}{t} is the number of edges in \ESW
  between \NBall{s} and \NBall{t}.
\item[3.] $\Dnorm[H](s,t) = (\kappa^2 /\NEdges{s}{t})^{1/\DIM}$.
\end{OneLiners}

\vspace{2mm}
{\bf Post-processing}. For each node pair $(s,t)$,

If $(s,t)\in H$, then $\DnormP[i](s,t)=\Dnorm[H](s,t)$; otherwise

\begin{OneLiners}
\item[1.] $H_t = \Set{(u,v)\in H}{\DnormS(u,v) \geq \tfrac{\GCDscale}{2\EXPAN} \text{ or $v=t$}} $.

\item[2.] $\DnormP[i](s,t)$ is the length of the shortest $s$-$t$ path
  in $H_t$ with respect to edge lengths $\Dnorm[H]$.
\end{OneLiners}

\vspace{2mm}
{\bf Notation.} \NRBall[\DnormS]{u}{\kappa} is the set of
the $\kappa$ closest nodes to $u$ according to \DnormS,
breaking ties arbitrarily.\\

\caption{The \extTwoBall (for a given category $i$).}
\label{alg:ext2B}
\end{algorithm}
%%%%%%%%%%%%%%%%%%%%%%%%%%%%%%%%%%%%%%%

\OMIT{ %%%%% David's version
$\mathcal{P}_{s,t}$ be the set of all $s$-$t$
paths in $H$ such that no edge $e=(u,v)$ in $P$ except possibly the
last edge has $\DnormS[i](u,v) < \ell$.
In other words, $\mathcal{P}_{s,t}$ consists of paths all of whose
edges are fairly long, except for the last one (according to \DnormS[i]).
Define
\begin{align}
\DnormP[i](s,t)
& = \min_{P \in \mathcal{P}_{s,t}} \sum_{(u,v) \in P} \Dnorm[H](u,v)
\label{eqn:multi-estimate}
\end{align}
to be the shortest-path distance in $H$ between $s$ and $t$,
where paths are restricted to be from $\mathcal{P}_{s,t}$.
(Notice that given $H$, $\DnormP[i](s,t)$ can be computed efficiently
by searching over all candidate last edges, then applying a
shortest-path search in a graph restricted to the long edges in $H$.)
} %%%%%%%%%

\subsubsection{Analysis: the \TwoBallTest for multiple categories}

We begin the analysis by showing that for sufficiently small distances,
\globalCD[\GCDscale] ensures that the basic \TwoBallTest gives
accurate estimates.
%, even with multiple categories.

\begin{lemma} \label{lem:2ball-mult}
Assume that the \globalCD[\GCDscale^{1+1/(d+1)}] condition holds, and let
$(s,t)$ be a node pair at normalized category-$i$ distance
$\Dnorm[i](s,t) = x \leq \GCDscale$.
Then, the \TwoBallTest obtains a distance estimate $\DnormP[i](s,t)$ of
$\Dnorm[i](s,t)$ with the following guarantee:
\[ \left| \DnormP[i](s,t)- \Dnorm[i](s,t) \right|
    \leq \left( x^{(\DIM+2)/(2\DIM+2)} + \tfrac{x^{\DIM+1}}{n} \right) \cdot O(\log^2 n).
\]
\end{lemma}

\OMIT{\[
\Dnorm[i](s,t)
- O (x^{(\DIM+2)/(2\DIM+2)}
+ x^{\DIM+1} \cdot \frac{\log n}{n}
+ x^{1/(\DIM+1)} \log^2 n)
\; \leq \; \DnormP[i](s,t) \; \leq \;
\Dnorm[i](s,t)
+ O (x^{(\DIM+2)/(2\DIM+2)}).
\]}

\begin{proof}
Recall from the proof of Theorem~\ref{thm:simple2ball-singleCat} that to estimate
$\Dnorm[i](s,t)$, the \TwoBallTest
considers two balls \NBall{s}, \NBall{t} around $s$ and $t$,
respectively, and counts edges between them.
The balls were chosen so that
$|\NBall{s}| = |\NBall{t}| = \kappa \triangleq \ShrinkT{x}^{\DIM}$, where
    $\ShrinkT{x} = x^{(\DIM+2)/(2\DIM+2)}$.
The improved distance estimate is
    $\DnormP(s,t) \triangleq (\kappa^2/\NEdges{s}{t})^{1/\DIM}$,
where \NEdges{s}{t} is the number of edges between \NBall{s} and \NBall{t}.

If only edges from \ESW[i] were counted,
Theorem~\ref{thm:simple2ball-singleCat} would apply verbatim.
However, edges between \NBall{s} and \NBall{t} from other categories
can be erroneously included in the count.
The presence of other categories never \emph{decreases}
\NEdges{s}{t}, so the high-probability lower bound on
\NEdges{s}{t}, and hence the high-probability upper bound on
$\DnormP(s,t)$, carries over from
Theorem~\ref{thm:simple2ball-singleCat}.

We need to prove a lower bound on $\DnormP(s,t)$. Let
$\NEdges[i]{s}{t}$ be the number of category-$i$ edges between
\NBall{s} and \NBall{t}. In the proof of Theorem
\ref{thm:simple2ball-singleCat}, we showed that with high probability,
$\NEdges[i]{s}{t}
 \leq \frac{\kappa^2}{(x - 8c\, \ShrinkT{x})^{\DIM}}$,
for some constant $c$. This implies
$ \NEdges[i]{s}{t}
     \leq \frac{\kappa^2}{x^{\DIM}}(1 + O(c\,r_x/x))^d
     \leq \frac{\kappa^2}{x^{\DIM}}(1 + O(c \DIM\, r_x/x))$.

We next count edges from other categories between
\NBall{s} and \NBall{t}.
Fix some category $j \neq i$, and consider node pairs
$(u \in \NBall{s},u' \in \NBall{t})$.
We distinguish between two distance scales for $\Dnorm[j](u,u')$.
\begin{enumerate}
\item We first consider the case that
$\Dnorm[j](u,u') > \GCDscale^{1+1/(\DIM+1)}$.
The probability for the edge $(u,u')$ to exist is then at most
$O(\GCDscale^{-(\DIM+1-1/(\DIM+1))})$.
The number of candidate pairs $(u,u')$ is at most
$\kappa^2 = x^{\DIM+1-1/(\DIM+1)} \leq \GCDscale^{\DIM+1-1/(\DIM+1)}$,
so the expected number of such long edges is $O(1)$.
Using Chernoff Bounds, with high probability, the number of long edges
is at most $O(\log^2 n)$.
\item The other case is $\Dnorm[j](u,u') \leq \GCDscale^{1+1/(\DIM+1)}$.
We divide the range of possible distances into exponentially
increasing buckets of the form $(y,2y]$. Suppose that
$y \leq \Dnorm[j](u,u') \leq 2y$ (for some $y \leq \GCDscale/2$).
Then, the pair $(u,u')$ has an edge with probability at most $O(y^{-\DIM})$,
and by the \globalCD[\GCDscale^{1+1/(\DIM+1)}] condition,
there are at most
$O(y^{\DIM}/n) \cdot |\NBall{s}|\, |\NBall{t}|\, + O(\log^2 n)$ pairs
$(u,u')$ at this distance scale.
Using linearity of expectations, and summing over all $O(\log n)$
distance scales $y$, we obtain that
the expected number of short category-$j$ edges between
\NBall{s} and \NBall{t} is at most
$O(\frac{|\NBall{s}| \, |\NBall{t}| \, \log n}{n} + \log^2 n)$,
and Chernoff Bounds establish concentration.
\end{enumerate}

Combining both cases, and substituting that
$|\NBall{s}| = |\NBall{t}| = \kappa$ gives us that with high
probability, the number of category-$j$ edges
between \NBall{s} and \NBall{t} is
at most $O(\frac{\log n}{n} \cdot \kappa^2 + \log^2 n)$.
Combining these edges across all categories $j\neq i$
and plugging in the upper bound for \NEdges[i]{s}{t}, we obtain:
\begin{align*}
\NEdges{s}{t}
&\leq \frac{\kappa^2}{x^{\DIM}}
  \left( 1 + O\left( c \DIM\, \frac{\ShrinkT{x}}{x} \right)\right)
           + O(\NUMCAT) \left( \tfrac{\log n}{n} \cdot \kappa^2 + \log^2 n
  \right).
\end{align*}
Adding some $\log n$ factors for simplification, and hiding the
constants inside $O(\cdot)$, we can re-write this bound as follows:
\begin{align*}
\NEdges{s}{t}
&\leq \frac{\kappa^2}{x^{\DIM}}
    \left(1 + O(\log^2 n) \left(
        x^{-\DIM/(2\DIM+2)} + \tfrac{x^{\DIM}}{n}
    \right)\right).
\end{align*}
Substituting the definition
    $\DnormP(s,t) \triangleq (\kappa^2/\NEdges{s}{t})^{1/\DIM}$,
it follows that
\begin{align*}
\DnormP[i](s,t)
&\geq x\,\left(1 - O(\log^2 n) \left(
        x^{-\DIM/(2\DIM+2)} + \tfrac{x^{\DIM}}{n}
    \right)\right) \\
&\geq x - O(\log^2 n) \left(
        x^{(\DIM+2)/(2\DIM+2)} + \tfrac{x^{\DIM+1}}{n}
    \right).\qedhere
\end{align*}
\end{proof}

\subsubsection{Analysis: the post-processing step}

Theorem~\ref{thm:extTwoBall} easily follows from
Lemma~\ref{lem:2ball-mult} and the following Lemma
\ref{lm:ext-two-ball-SP}, which analyzes the post-processing step.
The lemma is not specific to the actual estimates produced by the
\TwoBallTest.
Instead, it states that if each individual edge's length is estimated
with small additive distortion (compared to the true edge length),
then the multiplicative distortion of the overall estimates is small.
For readability, we continue to omit the subscript $i$ from all metrics.

\begin{lemma}\label{lm:ext-two-ball-SP}
Assume the setting of Theorem~\ref{thm:main},
and let \EXPAN be the expansion in \DnormS.
Consider running the post-processing step of the \extTwoBall
(parameterized by some \GCDscale) on
distance estimates \Dnorm[H] satisfying the following for some
$\Delta < \tfrac{\GCDscale}{4 \EXPAN^2}$:
\begin{align}\label{eq:ext-two-ball-SP}
    \left| \Dnorm[H](u,v)-\Dnorm(u,v) \right| \leq \Delta
    \quad \text{ for all } (u,v) \in H.
\end{align}
Then, the final estimates $\DnormP(s,t)$ have multiplicative distortion
    $1+O(\EXPAN^2 \Delta/\GCDscale)$
for all node pairs $(s,t)$ not in $H$.
\end{lemma}

\begin{proofof}[Proof of Theorem \ref{thm:extTwoBall}]
Without loss of generality, assume that
    $\GCDscale \leq \hat{\GCDscale} $,
where
    $\hat{\GCDscale}$
is from the theorem statement.
(If $\GCDscale > \hat{\GCDscale}$, then we could parameterize the
algorithm by $\hat{\GCDscale}$ instead.)
Then the upper bound in Lemma~\ref{lem:2ball-mult} becomes
    $\Delta_x \triangleq x^{(\DIM+2)/(2\DIM+2)}\cdot O(\log^2 n)$.

To complete the proof of Theorem \ref{thm:extTwoBall}, notice that
all edges $(u,v) \in H$, by definition, satisfy
$\DnormS(u,v) \leq \GCDscale$. As \DnormS is non-contracting, this also
implies that $\Dnorm(u,v) \leq \GCDscale$, so the bound \eqref{eq:ext-two-ball-SP}
holds with
    $\Delta = \Delta_R$,
according to Lemma~\ref{lem:2ball-mult}.
If $(s,t) \in H$ (which happens when $\DnormS(s,t) \leq \GCDscale$),
then we can apply Lemma \ref{lem:2ball-mult} directly to the edge
$(s,t)$, obtaining the bound in terms of $x$.
\end{proofof}

\begin{proofof}[Proof of Lemma \ref{lm:ext-two-ball-SP}]
Fix a node pair $(s,t)\notin H$, and let $x = \Dnorm(s,t)$.
Because $(s,t) \notin H$, and the estimate \DnormS has expansion at
most \EXPAN, we get that
$\Dnorm(s,t) \geq \tfrac{1}{\EXPAN}\, \DnormS(s,t)
                > \tfrac{\GCDscale}{\EXPAN}$.
Let $H_t \subseteq H$ be the edge set defined in~\eqref{eq:ext-two-ball-Ht},
and for any path $P$, let $\Dnorm(P)$ the length of the path $P$
according to the distance function \Dnorm.

We claim that the edge set $H_t$ contains an $s$-$t$ path $P$ with
    $k = \cel{x/(\tfrac{\GCDscale}{\EXPAN}-1)}$
hops and length
    $\Dnorm(P)\leq \Dnorm(s,t) + k$.
Consider the straight line between $s$ and $t$ in $\R^{\DIM}$.
For each $i$, let
$p_i$ be the point at \Dnorm-distance
$i \cdot (\tfrac{\GCDscale}{\EXPAN}-1)$
from $s$ on the straight line between $s$ and $t$.
The point $p_i$ itself may not be the location of any node in the
social network. However, by near-uniform density (which
guarantees that every unit cube contains at least one node of the
network), each point $p_i$ has a node $u_i$ at distance
at most $\D(p_i,u_i) \leq \DIM$.
Thus, $\Dnorm (p_i,u_i) \leq \DIM/(\CSW\,\kSW)^{1/\DIM} \leq \half$
for large enough $n$, as $\CSW \kSW = \Omega(\log n)$.

Let $P$ be the path $(s=u_0,u_1,u_2, \ldots, u_{k-1}, t=u_k)$.
By triangle inequality, all edges $(u_i,u_{i+1}) \in P$
have \Dnorm-length within $\pm 1$ of the distance
$\D(p_{i},p_{i+1})$ between the corresponding points $p_i$.
Therefore, $\Dnorm(P) \leq \Dnorm(s,t)+ k$.
Moreover, because each edge $(u,v)\in P$ satisfies
    $\Dnorm(u,v) \leq \tfrac{\GCDscale}{\EXPAN}$,
the fact that \DnormS has expansion at most \EXPAN implies that
    $\DnormS(u,v) \leq \GCDscale$.
In particular, each edge of $P$ is in $H$.
Furthermore, all edges $(u_i,u_{i+1})\in P$ except possibly the last
one satisfy
    $\DnormS(u_i,u_{i+1}) \geq \Dnorm(u_i,u_{i+1})
                          \geq \tfrac{\GCDscale}{\EXPAN}-2$.
By definition of $H_t$, it follows that the path $P$ is in $H_t$,
completing the proof of the claim.

Next, we upper-bound the estimated distance $\DnormP(s,t)$.
Simply using the path $P$ we just exhibited, we see that
\[
\DnormP(s,t)
\; \leq \; \Dnorm[H](P)
\; \overset{\eqref{eq:ext-two-ball-SP}}{\leq} \;
\Dnorm(P) + k\Delta
\; \leq \;
\Dnorm(s,t) + k(\Delta+1),
\]
where the last inequality used
the property that $\Dnorm(P)\leq \Dnorm(s,t) + k$.
An upper bound of $1+O(k\EXPAN/\GCDscale)$ on the expansion now
follows by substituting $k=O(x\,\tfrac{\EXPAN}{\GCDscale})$.

It remains to bound the contraction, by proving that
each $s$-$t$ path $P$ in $H_t$ has
  $\Dnorm[H](P)\geq \Dnorm(s,t) - O(x\,\EXPAN^2 \Delta/\GCDscale)$.
By the same argument as in the preceding paragraph, this holds
whenever $P$ has at most $4 x\,\EXPAN^2/\GCDscale$ hops.
We therefore focus on the case when $P$ has at least
$4x\,\EXPAN^2/\GCDscale$ hops.
Each of these hops $(u,v)$, except possibly the last one, has
  $\DnormS(u,v) \geq \tfrac{\GCDscale}{2\EXPAN}$ by definition of $H$.
In turn, by the maximum expansion of \DnormS,
the actual length of each hop is at least
    $\Dnorm(u,v) \geq \tfrac{\GCDscale}{2\EXPAN^2}$, so that the
estimates \Dnorm[H] satisfy
    $\Dnorm[H](u,v) \geq \Dnorm(u,v)-\Delta
                    \geq \tfrac{\GCDscale}{4\EXPAN^2}$,
because we assumed that $\Delta \leq \tfrac{\GCDscale}{4\EXPAN^2}$.
Summing over all (at least) $4x\,\EXPAN^2/\GCDscale$ hops $(u,v)$,
we obtain that
    $\Dnorm[H](P) \geq x = \Dnorm(s,t)$,
so in this case, the estimate has no contraction at all.
This completes the proof of the lower bound.
\end{proofof}

%%%%%%%%%%%%%%%%%%%

\subsection{The \recTwoBall for multiple categories}
\label{sec:recTwoBall-multicat}

We show that the \recTwoBall from
Section~\ref{sec:rectwoball-single} can be applied in the
case of multiple categories with \globalCD, yielding
poly-logarithmic additive error.
\asedit{To construct improved distance estimates for a given category $i$,
the algorithm is parameterized with the original edge set \ESW and
    initial estimates $\DnormS = \DnormS[i]$ from Theorem~\ref{thm:main},
and run verbatim from then on.}
The analysis only needs to be modified slightly
to deal with edges from other categories.
However, our guarantees only apply to node pairs at distances
    $x \leq n^{1/(\DIM+1)} = D^{\DIM/(\DIM+1)}$,
where $D = n^{1/\DIM}$ is the diameter of the metric space.

\begin{theorem}\label{thm:rec2ball-multi}
Consider a \multiSG with $\CSW \kSW = \Omega(\log n)$,
with \globalCD and perfectly uniform density for each category.
Assume that the social distance in each category is defined by the
$\ell_2^{\DIM}$ norm, with $\DIM>2$.
Then, the \recTwoBall runs in polynomial time,
and produces distance estimates \DnormP[i] satisfying the following
guarantee with high probability:
\begin{quote}
For every pair $(s,t)$ of nodes at normalized distance
$\Dnorm[i](s,t) \leq n^{1/(\DIM+1)}$, we have that
$$| \DnormP[i](s,t) - \Dnorm[i](s,t) | \leq \polylog(n).$$
\end{quote}
\end{theorem}

For normalized distances larger than $n^{1/(\DIM+1)}$, even under
actual randomly permuted categories, the number of edges from other
categories grows prohibitively large; it seems
unlikely that this obstacle could be easily overcome.

However, we can \asedit{combine the improved estimates provided by}
Theorem~\ref{thm:rec2ball-multi} with the post-processing step from
the \extTwoBall (with $\GCDscale=n^{1/(\DIM+1)})$. The resulting
algorithm estimates normalized distances $x>\GCDscale$ with additive
error $(x/\GCDscale) \, \polylog(n)$. (This follows from the
shortest-path argument encapsulated in
Lemma~\ref{lm:ext-two-ball-SP}.)

\begin{proofof}[Proof of Theorem~\ref{thm:rec2ball-multi}]
The proof of Theorem \ref{thm:rec2ball-singleCat} applies almost
verbatim. Recall that the \recTwoBall counts edges between balls
\NBall[\prime]{s}, \NBall[\prime]{t} around $s$ and $t$, containing
$\kappa = \Theta(\ShrinkR{x}^{\DIM})$ nodes each, where
$\ShrinkR{x} = x^{1/2 + 1/\DIM}$.
These balls are calculated with respect to the distances estimated by
the algorithm in earlier stages.
The only added difficulty for the analysis in the case of multiple
categories is bounding the additional edges between \NBall[\prime]{s} and
\NBall[\prime]{t} arising from categories $j \neq i$.

Notice that there are $\kappa^2 = O(x^{\DIM+2}) \leq O(x \cdot n)$ pairs of
nodes that could have an edge between them.
Focus on one category $j \neq i$, and divide node pairs
$(u,v), u \in \NBall[\prime]{s}, v \in \NBall[\prime]{t}$ into buckets of the form
$(y,2y]$ depending on their distance in category $j$.
By \globalCD, the bucket $(y,2y]$ contains at most
$O(\frac{y^{\DIM}}{n} \cdot |\NBall[\prime]{s}| \, |\NBall[\prime]{t}| + \log^2 n)
= O(y^{\DIM} \cdot x + \log^2 n)$ node pairs.
Each of these node pairs gives rise to an edge with probability
at most $O(y^{-\DIM})$, and summing over all $O(\log n)$ buckets
$(y,2y]$ gives us that the expected number of category-$j$ edges
between \NBall[\prime]{s} and \NBall[\prime]{t} is at most
$O(x \log n + \log^2 n) = O(x \log^2 n)$.
Using Chernoff Bounds and a union bound over all categories, with high
probability, the total number of edges added by categories
$j \neq i$ is at most $O(\NUMCAT x \log^2 n)$.

Because $\log^2 n = O(\AddDist{\ShrinkR{x}})$ for sufficiently large
\asedit{(but poly-logarithmic)} $x$, the
$O(\NUMCAT x \log^2 n) = O(x \, \AddDist{\ShrinkR{x}})$ additional
edges are easily subsumed in the error bound of $O(x \, \AddDist{\ShrinkR{x}})$
already present in the proof of Lemma \ref{lem:adddist}.
For smaller distances $x$, the only change will be a slightly
different poly-logarithmic base case for \AddDist{\ShrinkR{x}}.
\end{proofof}

%We need to show that Lemma \ref{lem:edge-count} still holds for multiple categories, namely we need to show that
%the expected number of edges between the balls is at most $(1+ O(\AddDist{\ShrinkR{x}}/x')) E(\ShrinkR{x}, x')$, where
%$\AddDist{\ShrinkR{x}}$ is the additive distortion results by our algorithm on distances $\ShrinkR{x}$, $x'$ is the real normalized distance between the centers of the balls and
%$E(\ShrinkR{x},x)\triangleq (\DIMCONST\, r^2/x)^{\DIM}$,
%where $\DIMCONST$ is the constant from the remark after Theorem~\ref{thm:rec2ball-singleCat}.
%%$E(\ShrinkR{x}, x')$ is the expected number of edges
%% between two balls of radius $\ShrinkR{x}$ whose centers
%% at distance $x'$ in the canonical SWG.
%In order to show that Lemma \ref{lem:edge-count} still holds here, we only need to show that the expected number of edges added by other categories is bounded by $O(\AddDist{\ShrinkR{x}}/x')) E(\ShrinkR{x}, x')$, the rest of the analysis will stay the same.

%By the \globalCD~condition, the expected number of edges from other categories between the balls is $O(\log{n} \ShrinkR{x}^{2\DIM}/n)$.
%One can show by straight forward calculations that this number is at most $\log(n) E(\ShrinkR{x}, x') /x$, therefore the additive distortion is at most
%$\log n$ as required.

\section{Category Disjointness and Random Permutations}
\label{sec:permutations}

Recall that our motivation for the definition of the
\LCD and \globalCD[R] conditions
was that they intuitively capture the notion of categories looking
random with respect to one another ``locally.''
In this section, we confirm the intuition guiding the definition, by
showing that both conditions are satisfied with high
probability when the metric space for each category $i$ is randomly
permuted, in the sense that
    $\D[i](u,v) = \DP[i](\sigma_i(u),\,\sigma_i(v))$
for some ``base metric'' $\DP[i]$ and a random permutation $\sigma_i$
on the node set.
Accordingly, both conditions are indeed significantly
weaker (in particular, more local) than requiring that
metrics be randomly permuted.

\begin{lemma} \label{lm:CD}
Consider a \multiSG with near-uniform density.
For each category $i$, let \DP[i] be a ``base'' metric,
and $\sigma_i$ a uniformly random permutation of the node set.
(The permutations for different metrics are pairwise independent.)
For each node pair $(u,v)$, the category-$i$ distance is
$\D[i](u,v) = \DP[i](\sigma_i(u),\,\sigma_i(v))$.
Then, with high probability, the \LCD and \globalCD conditions are
satisfied.\footnote{Therefore \globalCD[R] is satisfied for any $R$.}
\end{lemma}

\begin{proof}
Our proof uses an extension of Chernoff Bounds to dependent random
variables in which the randomness comes from a random permutation
(Theorem~\ref{thm:chernoff-permutations}, stated and proved below).

We begin by proving that the \LCD condition is satisfied.
Fix two categories $i \neq i'$.
Consider balls $B$, $B'$ of radii $r, r' = \polylog(n)$ in categories
$i, i'$, respectively.
Note that $\Expect{|B \cap B'|} = O((r r')^{\DIM}/n) < 1$.

Define a mapping from category $i$ to category $j$ by
    $\sigma(u) \triangleq \sigma_j^{-1}(\sigma_i(u)): V \to V$.
$\sigma(u)$ captures at what point of the metric space \D[j] a node in
the metric space \D[i] ends up.
Because $\sigma_i, \sigma_j$ were independent uniform permutations on
$V$, so is $\sigma$.
We will consider nodes $u \in B$, which we capture by setting
$\alpha_u = \indicator{u\in B}$.
Such a node is also in $B'$ iff $\sigma(u) \in B'$.
Thus, defining $X_u = \indicator{ \sigma(u)\in B' }$, we get that
    $|B\cap B'| = \sum_{u\in V} \alpha_u\, X_u$,
and by Theorem~\ref{thm:chernoff-permutations}, this sum is at most
$O(\log n)$ with high probability.

Next, we prove that the \globalCD condition holds as well.
Fix a category $j$, distance scale $r > 0$, and
two disjoint sets $B, B'\subseteq V, |B'| \geq |B|$
(which will be balls in category $i$).
Define the random variable
$f(B,B') \triangleq \sum_{v \in B, v' \in B'} \, \indicator{\D[j](v,v') < r}$
to be the number of node pairs at category-$j$ distance at most $r$.

We will prove a high-probability bound on $f(B,B')$ conditioned
on the choice of all permutations $\sigma_i$ for $i \neq j$.
In other words, we consider the probability space induced by the
random choice of $\sigma = \sigma_j$.
We will prove that with high probability,
\begin{align} \label{eq:globalCD-pf}
f(B,B') & = O(r^{\DIM}/n) \cdot |B| \, |B'| + O(\log^2 n).
\end{align}
Then the \globalCD condition follows by taking
a Union Bound over all categories $i,j$, all pairs of balls $B,B'$ in
category $i$, and all distinct distances $r$ in category $j$.

We begin by calculating the expectation of $f(B,B')$ using linearity
of expectation. Notice that
$\Expect{\indicator{\D[j](v,v') < r}} = \Prob{\D[j](v,v') < r}$ is the
probability that $v'$ is mapped to a node in a ball around $v$ of
radius $r$. Since there are $\Theta(r^{\DIM})$ nodes in the ball around $v$
of radius $r$ (wherever $v$ itself is mapped), we get that
$\Expect{\indicator{\D[j](v,v') < r}} = \Theta(r^{\DIM}/n)$, and
$\Expect{f(B,B')} = \Theta(r^{\DIM}/n) \cdot |B| \, |B'|$.

It remains to prove that $f(B,B')$ is concentrated around its expectation.
Thereto, we will use Theorem~\ref{thm:chernoff-permutations} twice.
First, focus on an arbitrary node $v'$ and consider $f(B, \SET{v'})$.
We have that $\Expect{f(B, \SET{v'})} = \Theta(r^{\DIM}/n) \cdot |B|$.
We can reveal the randomness of $\sigma$ by first revealing
$\sigma(v')$, which defines a set
$U = \Set{u \in V}{\DP[j](u,\,\sigma(v')) < r}$.
Then,
\begin{align*}
f(B, \SET{v'})
  & = \textstyle \sum_{v \in V}\; \alpha_v\, \indicator{\sigma(v) \in U},
\end{align*}
where $\alpha_v = 1$ if $v \in B$, and $\alpha_v = 0$ otherwise.
Thus, Theorem~\ref{thm:chernoff-permutations}
implies concentration of $f(B, \SET{v'})$ for any $v'$,
and gives us that with high probability,
$f(B, \SET{v'}) = O(\max(\log n, \frac{r^{\DIM}}{n} \cdot |B|))$
for all $v'$.
Let $N := \Theta(\max(\log n, \frac{r^{\DIM}}{n} \cdot |B|))$
denote this high-probability bound.

Next, our goal is to sum over all $v' \in B'$.
First, reveal $\sigma(v)$ for all $v \in B$, and condition on this
choice, writing $T = \Set{\sigma(v)}{v \in B}$.
Then, $\sigma$ is defined by a uniformly random permutation
from $V \setminus B$ to $V \setminus T$,
or --- equivalently --- by a uniformly random permutation
$\sigma^{-1}$ from $V \setminus T$ to $V \setminus B$.
For each $u' \in V \setminus T$, let
$\beta_{u'} = \sum_{u \in T} \indicator{\DP[j](u,u') < r}$
be the number of nearby locations to which nodes in $B$ were mapped.
Then, we can write
\begin{align*}
 \textstyle
f(B,B')
    \; = \; \sum_{u' \in V \setminus T} \;
  \beta_{u'}\, \indicator{\sigma^{-1}(u') \in B'}
\; = \; N \cdot \sum_{u' \in V \setminus T} \;
  \frac{\beta_{u'}}{N} \cdot \indicator{\sigma^{-1}(u') \in B'}.
\end{align*}
Defining $\alpha_{u'} = \min(1,\frac{\beta_{u'}}{N}) \in [0,1]$,
we get that with high probability (in the high-probability event that
$f(B, \SET{v'}) \leq N$ for all $v'$),
\begin{align*}
\textstyle
f(B,B')
 \leq N \cdot \sum_{u' \in V \setminus T} \;
  \alpha_{u'} \, \indicator{\sigma^{-1}(u') \in B'},
\end{align*}
and $\sigma^{-1}$ is a uniformly random permutation.
By Theorem \ref{thm:chernoff-permutations}, with high probability,
\begin{align*}
f(B,B')
& \leq \textstyle \; N \cdot O(\sum_{u' \in V \setminus T} \;
  \alpha_{u'}\ \indicator{\sigma^{-1}(u') \in B'} + \log n)\\
&= \; O(\Expect{f(B,B')} + N \log n).
\end{align*}
If $N = \Theta(\log n)$, this bound is obviously
$O(\Expect{f(B,B')} + \log^2 n)$.
Otherwise, $N = \Theta(\frac{r^{\DIM}}{n} \cdot |B|)$,
and $\frac{r^{\DIM}}{n} \cdot |B| = \Omega(\log n)$, which implies
(because $r^{\DIM} \leq n$) that $|B| = \Omega(\log n)$.
And because we assumed that $|B'| \geq |B|$,
we get that
\[
\Expect{f(B,B')}
\; = \; \Theta(r^{\DIM}/n) \cdot |B| \, |B'|
\; \geq \; \Theta(r^{\DIM}/n) \cdot |B| \, \log n
\; \geq \; \Theta(N \log n)
\]
so that the $N \log n$ term is subsumed in the \Expect{f(B,B')} term.
This completes the proof of the lemma.
\end{proof}

\begin{theorem}[Chernoff Bounds for permutations]
\label{thm:chernoff-permutations}
Fix $n\in\N$ and a subset $I \subseteq \SET{1, \ldots, n}$.
Let $\sigma$ be a uniformly random permutation of $\SET{1, \ldots, n}$.
For each $i \in \SET{1, \ldots, n}$,
fix $\alpha_i \in [0,1]$ and let
    $X_i = \indicator{\sigma(i) \in I}$.
Let
    $X = \sum_{i=1}^{n} \alpha_i X_i$
and $\mu = \E[X]$. Then $X$ satisfies both conditions from Theorem~\ref{thm:Chernoff}:
\begin{align*}
\Prob{|X-\mu| > \delta\mu} & \leq \exp(-\mu\, \delta^2/3),
    \qquad \text{ for any } \delta>0 \\
\Prob{X > (1+\delta)\mu'} & \leq \exp(-\mu'\, \delta^2/3),
    \qquad \text{ for any } \delta\in(0,1).
\end{align*}
%both inequalities
% \eqref{eqn:chernoff-two-sided} and \eqref{eqn:chernoff-upper}.
\end{theorem}

While the result appears standard, we are not aware of a published
proof, so for completeness we provide a self-contained proof.
The proof uses Chernoff Bounds for \emph{negatively associated} random
variables (see, e.g., \cite{dubhashi:panconesi:concentration-book}).
We summarize the relevant result in the following theorem:

\begin{theorem}[{\cite[pages 34--35 and Problem 3.1]{dubhashi:panconesi:concentration-book}}]
\label{thm:neg-assoc}
Let $X_1, \ldots, X_n$ be random variables
jointly distributed on $[0,1]^n$ such that $\sum_i X_i$ is a constant.
For any subset $I \subseteq \{1, \ldots, n\}$, write
$S_I \triangleq  \sum_{i \in I} X_i$.
Assume that the following hold for any such subset:

\begin{itemize}
\item Any $X_i$ with $i \in I$ is conditionally independent of
the $X_j$ with $j \notin I$ given $S_I$.

\item For any coordinate-wise non-decreasing function
$f:\R^{|I|} \to \R$, the conditional expectation
$\ExpectC{f(X_i,\, i \in I)}{S_I=t}$ is non-decreasing as a
function of $t \in \R$.
\end{itemize}

Then, the random variables $X_1, \ldots, X_n$ are said to be
\emph{negatively associated}.
In particular, it follows that $X \triangleq \alpha_i\, X_i$
satisfies the bounds from Theorem~\ref{thm:chernoff-permutations},
for any fixed $\alpha_1, \ldots, \alpha_n \in [0,1]$.
\end{theorem}

\begin{proofof}[Proof of Theorem \ref{thm:chernoff-permutations}]
First note that by definition, $\sum_{i=1}^n X_i = |I|$ is a constant.
Thus, it suffices to verify that the random variables $X_i$ are
negatively associated.

Fix $I \subseteq \{1, \ldots, n\}$.
For each $t \in \N$,
let $\F_t$ be the set of all tuples $(x_i,\,i\in I)$
such that $x_i\in\{0,1\}$ and $\sum_{i\in I} x_i = t$.
Let $U_t$ be the uniform distribution over $\F_t$.

To establish the first property of negative association, simply note
that the conditional distribution of $(X_i,\,i\in I)$ given $S_I = t$
and any assignment for $(X_i,\,i\notin I)$ is $U_t$, so independence
is established.

For the second property, fix a coordinate-wise non-decreasing
function $f:\R^{|I|} \to \R$. Since the conditional distribution of
$(X_i,\,i\in I)$ given $\{S_I = t\}$ is $U_t$, we have that
\[ g(t) \triangleq \ExpectC{f(X_i,\, i\in I)}{S_I=t}
        = \Expect[\vec{x}\sim U_t]{f(\vec{x})}.
\]
We need to show that
    $g(t+1)\geq g(t)$.
We couple selections according to $U_t$ and $U_{t+1}$ as follows.
\begin{OneLiners}
\item Pick $\vec{x} \sim U_t$.
\item Pick $j$ uniformly at random from $\Set{i\in I}{x_i=0}$.
\item Set $y_j=1$, and $y_i=x_i$ for all $j\neq i$.
% Set $y_i = 1$ if $x_i = 1$ or $i=j$, and $y_i = 0$ otherwise.
\end{OneLiners}
Notice that $\vec{y} \sim U_{t+1}$.
By monotonicity of $f$, we have that
$f(\vec{y})\geq f(\vec{x})$. It follows that
\[ g(t+1)
    = \Expect[\vec{y}\sim U_{t+1}]{f(\vec{y})}
    \geq \Expect[\vec{x}\sim U_t]{f(\vec{x})}
    = g(t).
\]
The claim now follows from applying the result for negatively
associated random variables.
\end{proofof} 
\section{Constant target degree}
\label{sec:constDeg}

\newcommand{\AdaptiveEDP}{Adaptive EDP algorithm\xspace}

The analysis so far has relied heavily on the fact that the target degree
\kSW (essentially the expected average node degree) was at least
logarithmic.
Indeed, as discussed in Section \ref{Prelims}, the first
obvious problem with constant expected degree is that with
non-negligible probability, the social graph \ESW is disconnected.
To circumvent this problem, much of the past literature
(e.g., \cite{fraigniaud:small-worlds-survey,small-world,small-world-nips,nguyen:martel})
assumes that in addition to the random edges, the network also
contains a set \Eloc of \emph{local edges} deterministically.\footnote{%
Without loss of generality, \Eloc can also include all edges which
would be included by the basic small-world model with probability 1.}
In the literature, \Eloc is frequently the \DIM-dimensional grid.
We adopt a more general model in which \Eloc can be essentially any
set of short edges.
A constant target degree poses two additional challenges beyond mere
connectivity:
%Beyond mere connectivity, there are two more subtle ways in which
% constant target degree poses a challenge:

\begin{itemize}
\item There are insufficiently many long-range links to support
  pruning via counting common neighbors. Even for short distances,
  the number of common neighbors is only constant, and
  high-probability guarantees can therefore not be
  obtained.\footnote{See, e.g., the difficulties faced by
  \cite{fraigniaud:lebhar:lotker}. The authors of
    \cite{fraigniaud:lebhar:lotker} consider a small-world model
    with one random neighbor for each node.
    % in which each node only has one random neighbor.
    They can only make guarantees
    about pruning away all but a poly-logarithmic number of long-range
    edges. The main reason is that even distant nodes will choose the
    same random neighbor with probability $\Omega(1/n)$, and
    high-probability bounds therefore only guarantee at most
    poly-logarithmically many long random edges to remain.}
  Therefore, in order to identify short edges as such, we need to rely
  on the structure of \Eloc.

\item To avoid stochastic dependence between multiple stages (such as
  the \SimpleTest and \TwoBallTest), we had previously partitioned
  \ESW randomly into separate sets to be used in the stages.
  With constant node degrees, this may risk leaving the \SimpleTest
  with only half of the local edges \Eloc. Hence, partitioning the
  edges may not be viable any more.
  On the other hand, if the same edges are used in multiple stages,
  subtle stochastic dependencies between the stages are created; our
  analysis needs to carefully account for these dependencies.
\end{itemize}

In this section, we explore the changes (in modeling, algorithms and analysis)
necessary to deal with constant target degrees.
We focus on the single-category case for the remainder of the section.

\asedit{Formally, we posit $\Eloc \subset \ESW$: the observed social
  graph $\ESW$ consists of all edges in \Eloc and the random edges
  generated according to the model in Section~\ref{Distance-Based}. We
  assume that the maximal degree $\Kloc$ of $\Eloc$ is constant.}

Our results apply so long as the set of local edges is ``rich enough'' in local
connectivity.
% using only short edges.

\begin{definition}[``Richness'' of local edges] \label{def:local-structure}

\begin{enumerate}
\item An edge set $E$ is a
\emph{$(\CONTR, \EXPAN)$-spanner} if its shortest-path distance \Dsp satisfies
the following for all node pairs $(u,v)$:
\[
\CONTR \cdot \D(u,v) \; \leq \; \Dsp(u,v) \; \leq \; \EXPAN \cdot \D(u,v)
\]

\item A set $E$ of edges is \emph{$(\PN,\PL)$-connected} if for every
edge $(u,v) \in E$, $E$ contains \PN edge-disjoint $u$-$v$ paths of
at most \PL edges each.

\item \Eloc is \emph{$(\PN,\PL)$-rich} with distortion $(\CONTR,\EXPAN)$
if it is a $(\CONTR, \EXPAN)$-spanner and
contains a $(\PN,\PL)$-connected $(\CONTR, \EXPAN)$-spanner $E \subseteq \Eloc$
(called its \emph{connectivity witness}).
\end{enumerate}

\end{definition}

\begin{note}{Remark.}
As an example, the \DIM-dimensional toroidal grid is
$(2\DIM-1,3)$-rich and (for $\DIM \geq 2$) $(2\DIM,7)$-rich,
both with distortion $(1,O(1))$.%
%We first show that the \DIM-dimensional toroidal grid is $(2\DIM-1,3)$-rich.
\footnote{%
Fix an edge $(u,v)$.
As a base case, for $\DIM=2$, it is easy to construct three paths of
lengths $(1,3,3)$, or four paths of lengths $(1,3,5,7)$.
For each added dimension, there are two additional disjoint paths of
length 3, taking one edge along the new dimension, an edge parallel to
$(u,v)$, and another edge in the new dimension. These paths are
clearly disjoint.}
%Let $(u,v)$ be an edge in the \DIM-dimensional toroidal grid.
%Assume w.l.o.g. that $u$ and $v$ differ at the first coordinate, and
%that $v$'s first coordinate is bigger (by 1) than $u$'s first
%coordinate.
%Note that $u$ and $v$ are equal in all other coordinates.
%For a node $x$, let $x(i,k)$ be the node obtained by adding $k$
%(modulo the radius) to the $i$'th coordinate of $x$
%(similarly, let $x(i,k_1)(j,k_2)$ be the node obtained by adding $k_1$
%to the $i$'th coordinate of $x$ and $k_2$ to the $j$'th coordinate).
%Consider the following set of $2d-1$ paths $\{(u,u(i,1),v(i,1),v),
%(u,u(i,-1),v(i,-1),v) \mid 1 < i \leq \DIM\} \cup \{(u,v)\}$,
%each of length at most 3 hops, and it is not hard to verify that these
%paths are disjoint.
%Since this argument holds for arbitrary edge $(u,v)$, we have proved
%the the \DIM-dimensional toroidal grid is $(2\DIM-1,3)$-rich.
%To prove that the \DIM-dimensional toroidal grid is $(2\DIM,7)$-rich,
%consider the following set of $2d$ paths.

%\begin{eqnarray*}
%&&
%\{(u,u(i,1),v(i,1),v), (u,u(i,-1),v(i-1),v) \mid 2 < i \leq \DIM\} \cup \\
%&&
%\{(u,u(1,-1),u(1,-1)(2,1),u(2,1),v(2,1),v(1,1)(2,1), v(1,1),v)\} \cup \\
%&&
%\{(u,u(2,1),u(2,2),v(2,2),v(2,1),v)\} \cup \\
%&&
%\{(u,u(2,-1),v(2,-1),v)\} \cup \\
%&&
%\{(u,v)\}.
%\end{eqnarray*}

%Each path in that set is of length at most 7 hops, and it is not hard
%to verify that these paths are disjoint.
%Since this argument holds for arbitrary edge $(u,v)$, we have proved
%the the \DIM-dimensional toroidal grid is $(2\DIM,7)$-rich.}
\end{note}

Next we present a solution which relies on knowing parameters $(\PN,\PL)$ of the local structure's richness. In other words, the pruning algorithm needs to know
how rich a local structure to expect. In Section \ref{sec:adaptive-EDP}, we show how to
make the pruning algorithm adapt to the available richness under fairly mild assumptions.

\subsection{Basic Approach: Edge-Disjoint Paths}
\label{sec:EDP}

Our solution is based on a more careful design of the pruning stage,
where instead of counting common neighbors, the algorithm counts
edge-disjoint paths of bounded length.
The pruning stage is very simple: The algorithm starts with an edge set $E = \ESW$.
It prunes each edge $(u,v)\in E$ such that $E$ does not contain \PN
edge-disjoint $u$-$v$ paths of at most \PL hops each.
This is repeated until no more edges can be pruned.
We call this algorithm the \emph{\ItAlg{\PN}{\PL};}
here, \emph{EDP} stands for Edge-Disjoint Paths.
See Algorithm~\ref{alg:EDP-test} for pseudocode.

%%%%%%%%%%%%%%%%%%%%%%%%%%%%%%%%
\begin{algorithm}[htb]
{\bf Input.} Edge set $E$.

{\bf Repeat}
\begin{OneLiners}
\item[1.] Find any $(u,v) \in E$ s.t.~$E$ does not contain \PN
  edge-disjoint $u$-$v$ paths of at most \PL hops each.
\item[2.] Prune $(u,v)$ from $E$.
\end{OneLiners}
{\bf Until} no such edges $(u,v)$ remain.
\caption{The \ItAlg{\PN}{\PL}.}
\label{alg:EDP-test}
\end{algorithm}
%%%%%%%%%%%%%%%%%%%%%%%%%%%%%%%%

The idea is that this algorithm keeps a sufficiently rich subset of
local edges, and prunes all edges in \ESW whose length exceeds some
threshold \constDR (defined in Equation~\eqref{eq:EDPtest}).
(We call such edges \emph{long edges}.)
For edges of intermediate length, the algorithm makes no guarantees
about whether they are pruned.
Crucially, the pruned graph does not depend on the long edges, in the
following sense:
Let $\ESW, \ESWP$ be two edge sets generated according to the same
distribution, such that the random choices for non-long edges
are the same, and the random choices for long edges are independent.
Then, with high probability (over the random process
generating all edges of \ESW and \ESWP), the remaining set of edges
after pruning is the same for both \ESW and \ESWP.
The advantage of this guarantee is that we do not need to
worry about dependencies on the pruned graph, so long as the
post-processing stage only uses long edges.
Therefore, we can use the pruned graph to define the initial
estimates \DnormS for normalized distances and then use a suitably
modified and optimized version of the (Recursive) \TwoBallTest which only
considers node pairs $(s,t)$ for which $\DnormS(s,t)$ is sufficiently
large. We omit the (easy) modifications of the algorithm and analysis.

We start the analysis of the \ItAlg{\PN}{\PL} with several observations.
First, notice that
the pruned graph \PrunedG{E} is the maximal $(\PN,\PL)$-connected
subset of $E$, i.e., the union of all such subsets.
It follows that \PrunedG{E} does not depend on the order in which the edges
are pruned.
Second, because \PrunedG{E} is the maximal $(\PN,\PL)$-connected subset
of $E$, the pruned graph \PrunedG{E} does not depend
on the presence or absence of the pruned edges
$ e\in E\setminus \PrunedG{E}$.
Formally, $\PrunedG{E} = \PrunedG{E'}$
whenever $\PrunedG{E} \subseteq E' \subseteq E$.

To ensure correctness, we can use the \ItAlg{\PN}{\PL} only if
the local structure is $(\PN,\PL)$-rich.
The performance depends on the parameters $(\PN,\PL)$: we get better
estimates for larger \PN and smaller \PL.
We summarize our results as follows.
In a slight abuse of notation, here, the (Recursive) \TwoBallTest
refers to the suitably modified version that works with the
\ItAlg{\PN}{\PL}.

\begin{theorem} \label{thm:constant-main}
Consider a \singleSG of near-uniform density.
Suppose that the local edge set \Eloc is $(\PN,\PL)$-rich with
distortion $(\CONTR,\EXPAN)$,
\asedit{and has constant maximal degree $\Kloc$}.
Let $D = \Theta(n^{1/\DIM})$ be the diameter of the metric space,
\asedit{and assume that $\delta< D^{2/\PN}$.}
For any constant $\alpha>0$ (which need not be known to the
algorithm), let
\begin{align}\label{eq:EDPtest}
\constDR[\alpha]
& = \; D^{(2+\alpha)/\PN} \cdot \PL \cdot (O(\kSW \asedit{+\Kloc} + \log^{1+\alpha} n))^{2h/\DIM} \\
& = \; D^{(2+\alpha)/\PN} \cdot (O(\log n))^{O(\PL)}. \nonumber
%O(\polylog(n)).
\end{align}
Let \asedit{$T(E)$} be the edge set retained by the \ItAlg{\PN}{\PL}.
Then, with probability at least $1-O(n^{-\alpha})$, the following hold.
{\renewcommand{\labelenumi}{(\alph{enumi})}
\begin{enumerate}
\item \asedit{$T(E)$} contains the connectivity witness \ElocP of \Eloc
and no edges whose length exceeds \constDR[\alpha].
The algorithm makes no guarantees for other edges.

\item Let \Dsp be the shortest-path distance on \asedit{$T(E)$}.
Then, for all node pairs $(u,v)$, we have that
\[
\D(u,v) \; \leq \;
\beta\, \Dsp(u,v)
\; \leq \; \EXPAN \cdot \beta\, \D(u,v),
\text{ where }
\beta = \max(\tfrac{1}{\CONTR}, \constDR[\alpha]).
\]
In words, the shortest paths distance in \asedit{$T(E)$}, scaled up by
$\beta$, gives no contraction,
and expansion at most $\EXPAN\, \beta$.

\item The \TwoBallTest reconstructs all normalized distances
  $\Dnorm(u,v)$ with unit distortion and additive error
        $\constDR[\alpha] (\Dnorm^\gamma(u,v) + \constDR[\alpha])$,
where
        $\gamma = {\frac{\DIM+2}{2\DIM+2}}$.

\item Assume that the metric has perfectly uniform density,
and the social distance is the $\ell_2^{\DIM}$ norm for $\DIM \geq 3$
dimensions. Then the Recursive Two-Ball Algorithm reconstructs all
normalized distances with unit distortion and additive error
$\constDR[\alpha] \cdot \polylog(n)$.
\end{enumerate}}
\end{theorem}

\OMIT{ %%%%%%%% moved earlier
Notice that these guarantees crucially rely on the algorithm's
parameters $(\PN,\PL)$ matching those of the local structure's richness.
In other words, the pruning algorithm needs to know how rich a local
structure to expect. In Section \ref{sec:adaptive-EDP}, we show how to
make the pruning algorithm adapt to the available richness under
fairly mild assumptions.
However, first, we prove Theorem \ref{thm:constant-main}.
} %%%%%%%

\begin{proof}
Most of the proof will focus on the first part of the theorem, i.e.,
that  with high probability, all edges of length at least
\constDR[\alpha] are pruned.
The remaining parts then follow analogously to previous proofs.
The proof of the second part is virtually identical to the proof of
Lemma \ref{lem:shortest-paths}.
The analysis of the (Recursive) \TwoBallTest is also similar to the
high-degree case, as long we we
establish the independence between the pruned graph and the long
edges: the edges of length exceeding \constDR[\alpha].
The reason that this independence is sufficient is that the
(Recursive) \TwoBallTest only uses long edges, and its analysis can
then omit any conditioning on the pruned graph.

To prove independence formally, let \ESW be a random edge set, and $E$
the set of all its non-long edges (of length at most \constDR[\alpha]).
Let \ESWP be another random edge set drawn from the same distribution
whose non-long edges are also exactly $E$, while its long edges are
generated independently from those of \ESW.
With high probability, the \ItAlg{\PN}{\PL} will prune all long edges from
both \ESW and \ESWP.
By the observation preceding Theorem \ref{thm:constant-main}, this
implies that $\PrunedG{\ESW} = \PrunedG{E}$ and
$\PrunedG{\ESWP} = \PrunedG{E}$, so that the
\ItAlg{\PN}{\PL} will produce the same pruned edge set on both graphs.

The remainder of the proof focuses on the first part of the theorem,
i.e., the fact that with high probability, all long edges are pruned.
The proof involves an intricate Deferred Decisions argument
encapsulated in Lemma~\ref{lm:feasibleSets} below,
which may be of interest in its own right.

\newcommand{\KeepingProb}{\ensuremath{\pi_{s,t}}\xspace}

Fix parameters $(\PN,\PL)$ and a node pair $(s,t)$, and let
$r = \D(s,t) > \constDR[\alpha]$.
In applying Lemma \ref{lm:feasibleSets}, we consider the ``universal
set'' $U$ of all node pairs.
Recall that the edge set $E = \ESW$ includes each node pair $(u,v)$
independently with some probability $p_{(u,v)}$.
The ``feasible subsets'' of $U$
(``feasible paths'') are all simple $s$-$t$ paths of at most \PL
hops. Any such path must contain at least one hop of length at least
$\tfrac{r}{\PL}$.
\asedit{Since $\tfrac{r}{\PL} > D^{2/b} >\EXPAN$, this hop cannot
  belong to $\Eloc$; instead, it must be a random edge. The
  probability of this random edge being present in \ESW is at most}
% the corresponding edge is present with probability at most
 $q\triangleq \CSW\, \kSW\, (\PL/r)^{\DIM}$.
By Lemma~\ref{lm:feasibleSets}, we obtain that for each $c\in \N$,
\begin{align} \label{eq:DisjFeasPaths}
\KeepingProb
&\triangleq \Prob{\text{\ESW contains \PN disjoint feasible paths}} \\
& \leq \Prob{|E'|>c} + \tfrac{1}{1-cq}\, (c q)^{\PN}, \nonumber
\end{align}
where $E'$ is the set of all node pairs $(u,v)$ such that
$\ESW \cup \SET{(u,v)}$ contains a feasible path.

The edge $(s,t)$ is retained with probability at most
\KeepingProb.  Once we prove that $\KeepingProb = O(n^{-(2+\alpha)})$,
we can complete the proof by taking the Union Bound over all $n^2$
node pairs $(s,t)$. So it remains to upper-bound the right-hand side
of~\eqref{eq:DisjFeasPaths} by $O(n^{-(2+\alpha)})$.

We first bound $\Prob{|E'|>c}$ in~\eqref{eq:DisjFeasPaths}.
Let the random variable $\Delta$ denote the maximum degree of \ESW.
Any node pair \asedit{$(u,v)\in T(E)$} has the property that \ESW contains both an
$s$-$u$ path and a $v$-$t$ path of length at most \PL hops each.
Therefore, for fixed endpoints $(s,t)$, there are at most
$\Delta^{\PL}$ candidates for $u$ and at most $\Delta^{\PL}$
candidates for $v$, and thus at most $\Delta^{2h}$ candidates for
$(u,v)$. We have thus proved that
    $|E'| \leq \Delta^{2h}$.
Now, using Chernoff Bounds to upper-bound $\Delta$, we have:
\begin{align*}
    \Prob{\Delta \geq \asedit{\Kloc + } \Theta(\kSW+\log\tfrac{n}{\delta}) }
        \leq \delta/n^2, \quad \text{ for all } \delta>0.
\end{align*}
\asedit{The term $\Kloc$ accounts for the edges in the local structure
  $\Eloc$, and the other terms represent the contribution of the
  random edges. We conclude that}
    $\Prob{ |E'| \geq c} \leq \delta/n^2 $
for
    $c=(\asedit{\Kloc + } \Theta(\kSW+\log\tfrac{n}{\delta}))^{2\PL}$.

Substituting this choice of $c$ into~\eqref{eq:DisjFeasPaths} and
taking $\delta=n^{-\alpha}$, we obtain:
\begin{align*}
\KeepingProb  = O(n^{-(2+\alpha)} + (cq)^{\PN}).
\end{align*}
Finally, we show that $\KeepingProb = O(n^{-(2+\alpha)})$ by substituting
    $q= \CSW\, \kSW\, (\PL/r)^{\DIM}$
and
    $r \geq \constDR[\alpha]$.
\end{proof}

\begin{lemma} \label{lm:feasibleSets}
Consider a \emph{universe set} $U$ and a collection $\F$ of non-empty
subsets of $U$ called \emph{feasible sets}.
A random set $E \subseteq U$ is obtained by including each element
$e\in U$ independently with probability $p_e$. The goal is to bound
from above the number of disjoint feasible subsets of $E$.

Fix $q\in [0,1]$ such that each feasible set contains at least one element $e$ with $p_e\leq q$. Let $E'$ be the set of elements $e\in U$ such that \asedit{$F \subseteq E$ and $e \in F$} for some feasible set $F$. Then, for each $\PN \in \N$,
\begin{align} \label{eq:feasibleSets}
&\Prob{\text{$E$ contains \PN disjoint feasible sets}} \\
& \qquad \leq \min_{c\in\N} \,\left[ \Prob{|E'|>c} + \frac{1}{1-cq}\, (c q)^{\PN} \right]. \nonumber
\end{align}
\end{lemma}

\begin{proof}
An element $e\in U$ with $p_e\leq q$ is called a \emph{witness}.
Fix an arbitrary ordering $\rho$ of $U$ in which all non-witnesses
precede all witnesses. For each feasible set $F\in \F$, the latest
witness \asedit{$w$} in $F$ according to $\rho$ is called a \emph{canonical
  witness} for $F$. If furthermore $F\subseteq E$, then $w$ is called
\emph{$E$-important}. Since each feasible set $F\subseteq E$ contains
an $E$-important witness, from here on, we will focus on counting
distinct $E$-important witnesses (rather than disjoint feasible sets
$F\subseteq E$).

We reveal one by one whether elements of $U$ are included in $E$, in
the order of $\rho$.
For each witness $w$, let $E_w$ be the actual subset of $E$ that is
revealed \emph{before} $w$ is considered.
Let us say that $w$ is \emph{$\rho$-important} if it is a
canonical witness for some feasible set $F\subseteq E_w \cup
\SET{w}$. Then, $w$ is $E$-important if and only if
$w \in E$ and $w$ is $\rho$-important.
The latter two events, namely $\SET{\text{$w$ is $\rho$-important}}$
and $\SET{w \in E}$, are independent.

Let $w(t)$ be \Kth{t} $\rho$-important witness chosen in the above
revelation process, $X_t = \indicator{w(t)\in E}$, and let
%be the indicator variable of the event that $w_t\in E$. Let
$N$ be the total number of $\rho$-important witnesses.
Then, $S_N \triangleq \sum_{t=1}^N\, X_t $
is the total number of $E$-important witnesses.
Our goal is to bound $S_N$ from above.

We accomplish this goal via Lemma~\ref{lm:probab} below.
The sequence $\SET{X_t}$ and the stopping time $N$ satisfy the
conditions in Lemma~\ref{lm:probab} (the upper bound).
Specifically, we have established that
$\ExpectC{X_t}{N\geq t} =p_{w(t)} \leq q$,
and the event $\{ X_t=1\}$ is independent of the past history given
that $N\geq t$.
By Lemma~\ref{lm:probab}, we obtain that for all $c$,
\begin{align}
\Prob{S_N \geq \PN}
    & \leq \BDF{c}{q}{\PN} + \Prob{N > c},
\end{align}
where $\Binomial{c}{q}$ is a random variable distributed according to
the Binomial distribution with $c$ samples and success probability
$q$. We have
    $\Prob{N > c} \leq \Prob{|E'| > c}$,
since each $\rho$-important witness is in $E'$.
We complete the proof by noting that
\[
\BDF{c}{q}{\PN}
  \; = \; \sum_{l=\PN}^c \binom{c}{l}\, q^l (1-q)^{c-l}
  \; \leq \; \sum_{l=\PN}^c (cq)^l
  \; \leq \; \tfrac{1}{1-cq}\, (c q)^{\PN}. \qedhere
\]
\end{proof}

Lemma \ref{lm:probab} below is a technical lemma for analyzing a
certain kind of ``revelation process,'' in which a sequence of
history-dependent 0-1 random variables $X_t$ is revealed, and the length $N$ of
this sequence is also a history-dependent random variable.
The lemma shows that whenever the \asedit{conditional} expectation of each individual 0-1
random variable can be bounded, we can also bound the sum \asedit{$S_N$ of these random variables}: we bound
the distribution of $S_N$ in terms of the corresponding Binomial distribution.
We will also use this lemma in the analysis of the adaptive algorithm
in Section~\ref{sec:adaptive-EDP}.

%In what follows, $\Binomial(t,p)$ denotes the Binomial distribution
%with $t$ samples and success probability $p$.

\begin{lemma} \label{lm:probab}
Consider a stochastic process $X_t \in \SET{0,1}, t\in \N$ and a
stopping time $N$ on a filtration $\SET{\F_t: t \in \N}$.%
\footnote{\asedit{In other words, $\SET{\F_t: t \in \N}$ is an increasing sequence of $\sigma$-algebras such that each random variable $X_t$, $t\in \N$ is $\F_t$-measurable, and $\{N\leq t\}\in \F_t$ for each $t\in \N$.}}
Define $S_N \triangleq \sum_{t=1}^N X_t$.
Assume that for some constants $p \leq q$,
\begin{align*}
    \ExpectC{X_t}{N\geq t, F} & \in [p, q] \quad \text{for all } t \in \N, F \in \F_{t-1}.
\end{align*}

Let $\Binomial{t}{p}$ be a random variable distributed
according to the Binomial distribution
with $t$ samples and success probability $p$.
Then, for all $x, t \in \N$,
\begin{align}\label{eq:probab-lemma}
\BDF{t}{p}{x} - \Prob{N < t}
&\leq \;\; \Prob{S_N  \geq x} \\
&\leq \;\; \BDF{t}{q}{x} + \Prob{N < t}. \nonumber
\end{align}
\end{lemma}

\begin{proof}
It suffices to prove the lower bound in~\eqref{eq:probab-lemma}; the
upper bound is then derived from the lower bound applied to the
stochastic process $\Set{1-X_t}{t\in \N}$.
Let $\Set{Y_t}{t\in \N}$ be a family of mutually independent 0-1
random variables with expectation $p$, and define
\begin{align*}
X^*_t = \begin{cases}
    X_t,    & N \geq t \\
    Y_t,    & \text{otherwise}.
\end{cases}
\end{align*}
For each $t$, let $S_t = \sum_{s=1}^t X_s,\; S^*_t = \sum_{s=1}^t X^*_s$, and
$\F^*_t = \sigma(X^*_1, \ldots, X^*_t)$.
For each event $F\in \F^*_{t-1}$, we have that
\begin{align*}
\ExpectC{X^*_t}{F, N\geq t}
    &= \ExpectC{X_t}{F, N\geq t} \geq p, \\
\ExpectC{X^*_t}{F, N< t}
    &= \ExpectC{Y_t}{F} = p,
\end{align*}
which implies that
\begin{align*}
\ExpectC{X^*_t}{F}
&= \; \ExpectC{X^*_t}{F, N \geq t} \cdot \ProbC{N\geq t}{F}
      + \ExpectC{X^*_t}{F, N< t} \cdot \ProbC{N< t}{F} \\
&\geq  p.
\end{align*}
By induction on $t$, it follows that
    $\Prob{S^*_t \geq x} \geq \BDF{t}{p}{x}$
for all $x,t\in \N$.
Noting that $S_N \geq S_t = S^*_t$ whenever $N \geq t$, we obtain that
\begin{align*}
\Prob{S_N \geq x}
    & \geq \ProbC{S_N \geq x}{N\geq t} \cdot \Prob{N\geq t} \\
    & \geq \ProbC{S^*_t \geq x}{N\geq t} \cdot \Prob{N\geq t} \\
    & = \Prob{S^*_t \geq x \mbox{ and } N\geq t} \\
    & \geq \Prob{S^*_t \geq x} - \Prob{N<t}  \\
    & \geq \BDF{t}{p}{x} - \Prob{N<t}. \qedhere
\end{align*}

\end{proof}

\subsubsection{Running times in Theorem~\ref{thm:constant-main}}
While the main thrust in this paper is information-theoretic, the
algorithms in Theorem~\ref{thm:constant-main} are actually
polynomial. Let us discuss how to improve the running times to near-linear,
an important feature for the sizes of networks we are envisioning.

The na\"{i}ve implementation of the \ItAlg{\PN}{\PL} checks every
remaining edge at each iteration, which gives a running time of
$\tilde{O}(n^2)$. We show how to reduce it to to $\tilde{O}(n)$.

\begin{lemma}\label{lm:EDP-runtime}
The \ItAlg{\PN}{\PL} can be implemented in
$\tilde{O}(n)$ time for constant \PN and \PL.
\end{lemma}

\begin{proof}
We maintain a queue of edges to be checked,
initially containing all edges of \ESW.
In each step, one edge $e=(u,v)$ is removed from the queue and checked
for pruning with respect to the current pruned graph \Ecur.
If \Ecur does not contain the requisite \PN-tuple of edge-disjoint
paths of length at most \PL, then $e$ is pruned permanently.
Otherwise, the \PN-tuple of paths provides a ``certificate'' for $e$.
Later iterations may remove edges from this certificate; therefore,
for each edge $e'$ in the certificate, the algorithm stores a pointer
that $e'$ is part of the certificate for $e$.
If $e'$ is pruned at any point, then, following the pointers, the
algorithm can determine all edges $e$ whose certificates $e'$
participates in. Upon pruning $e'$, all such edges $e$ are then
re-enqueued and will need to be checked again for alternative certificates.
Once the queue becomes empty, the algorithm terminates.

Without loss of generality, the target degree \kSW
is $O(\log^2 n)$ (otherwise, the much more efficient \SimpleTest
from Section~\ref{sec:amoeba} would be used).
By Chernoff Bounds, all node degrees are $O(\log^2 n)$ with high probability.
Finding a certificate for a given edge using brute force then takes
only $\polylog(n)$ time.
Moreover, for each edge $e$, there can be at most $\polylog(n)$ edges
whose certificates $e$ participates in.
No new edges are added to the queue if the current edge is not pruned,
and at most $\polylog(n)$ edges are added otherwise.
Therefore, the running time is $\tilde{O}(n)$.
\end{proof}

We also comment on the running time of the \TwoBallTest.
Applying this algorithm to a given node pair $(u,v)$ can be
computationally expensive when $\D(u,v)$ is large (and consequently, the
algorithm needs to consider large balls around $u$ and $v$).
Thus, the \TwoBallTest for a given node pair can be viewed as a
precise but costly distance measurement.
Instead of applying it to \emph{every} node pair, we could
instead use the beacon-based triangulation technique
from~\cite{kleinberg:slivkins:wexler:triangulation}: here, one selects
$O((\tfrac{1}{\eps})\, (\tfrac{1}{\delta})^{\DIM})$ ``beacon nodes''
uniformly at random, and measures the distance from each node only to
each beacon. This technique achieves distortion $(1+\delta)C$ for all but an
$\eps$-fraction of node pairs, where $C$ is the distortion of the
Two-Ball test.

%%%%%%%%%%%%%%%%%%%%%%%%%%%%%%%%%%%%%%%%%%%%%

\subsection{Adapting to the ``optimal'' richness}
\label{sec:adaptive-EDP}

Theorem~\ref{thm:constant-main} assumes that the $(\PN,\PL)$-richness of
the local edge set \Eloc is known to the algorithm.
In reality, it is desirable to adapt to the ``optimal'' richness
without knowing it in advance.
Here, the ``optimal'' richness means the $(\PN,\PL)$ pair that minimizes
\constDR[\alpha] in Equation~\eqref{eq:EDPtest}, subject to the constraint
that \Eloc is $(\PN,\PL)$-rich with small distortion.
We show that such an automatic adaptation can be achieved if
\Eloc is ``robust,'' in the sense defined below.

Our algorithm, called \emph{\AdaptiveEDP,} proceeds as follows: for a
given set \PLSET of candidate hop counts, we try all $(\PN,\PL)$ pairs,
$\PL\in \PLSET$, in order of increasing \constDR[\alpha] until the pruned
graph is connected, and focus on the last pair. Without loss of
generality, we can start with $\PN$ equal to the smallest node degree
in $\ESW$. We can use binary search over the $(\PN,\PL)$ pairs (in the
same order) to reduce the number of pairs that we need to
consider.

\asedit{While the above algorithm is very simple, the challenge is to prove
that it works, in the sense that the chosen $(\PN,\PL)$ pair is optimal.
In particular, we need to identify suitable assumptions on \Eloc and
the set $H$ that parameterizes the algorithm.
While in practice, the choice of $H$ is heuristic, we prove that the
algorithm works as long as the assumptions are satisfied.}

%We identify a suitable ``robustness
%property'' of \Eloc and argue that under this property, the chosen
%$(\PN,\PL)$ pair is optimal.

Let \PrunedG[\PN,\PL]{E} denote the
pruned graph if \ItAlg{\PN}{\PL} is applied to the edge set $E$. We
rely on the following crucial observation:

\begin{lemma}\label{lm:adaptive-EDP-isolatedNodes}
Consider a \singleSG with near-uniform density. Suppose
that the local structure \Eloc is a $(\cdot,\EXPAN)$-spanner, and
moreover, \PrunedG[\PN,\PL]{\Eloc} contains at least $\eps n$
isolated nodes, for some parameters $\PN,\PL,\eps,\EXPAN$ such that
\begin{align} \label{eq:adaptiveEDPtest}
(2\EXPAN \PL)^{\DIM}\; \densityK^2\, \CSW\, \kSW \leq \tfrac{1}{6}.
\end{align}
Then $\PrunedG[\PN,\PL]{\ESW}$ is disconnected with high probability.
\end{lemma}
% e^{\DIM}\; \densityK\, \CSW\, \kSW \cdot \Kloc^{\PL} & \leq \tfrac{1}{3}.

\begin{note}{Remark.}
Since $\CSW = \Theta(1/\log n)$ and $\densityK = \Theta(1)$,
condition~\eqref{eq:adaptiveEDPtest} holds, for large enough $n$,
whenever $\kSW$, $\EXPAN$ and \PL are constants.
\OMIT{The condition plugs into Inequality~\eqref{eq:adaptiveEDPtest-conditionUsed} in the
proof. We conjecture that it can be somewhat relaxed using a more
careful analysis.}
\end{note}

Lemma~\ref{lm:adaptive-EDP-isolatedNodes} is proved below.
It naturally motivates the following definition of ``robustness.''

\begin{definition}\label{def:EDP-robust}
A connected graph $G=(V,E)$ is called \emph{$(\eps,\PL)$-robust} with
distortion $(\CONTR,\EXPAN)$, for some $\eps\in(0,1]$, if the
following holds for every \PN:
either $G$ is $(\PN,\PL)$-rich with distortion $(\CONTR,\EXPAN)$,
or \PrunedG[\PN,\PL]{E} contains at least $\eps n$ isolated nodes.\footnote{Note that any graph $G$ in Definition~\ref{def:EDP-robust} is a $(\CONTR,\EXPAN)$-spanner. This is because for $\PN=1$ no edges are pruned, and so $G$ must be $(1,\PL)$-rich with
distortion $(\CONTR,\EXPAN)$, which in turn implies that it is a $(\CONTR,\EXPAN)$-spanner.
}
\end{definition}

In the first case of this definition, we can use the \ItAlg{\PN}{\PL}
safely, while in the second case, we will show that
\PrunedG[\PN,\PL]{\ESW} is disconnected with high probability.

Notice that the toroidal grid is $(1,\PL)$-robust for any \PL.
We give more examples of robust graphs in
Section~\ref{sec:adaptive-EDP-examples}.

\begin{theorem}\label{thm:adaptiveEDPtest}
Consider a \singleSG with near-uniform density and local structure \Eloc
\asedit{of constant maximal degree $\Kloc$}.
%and maximum degree $\Kloc = \deg(\Eloc)$.
Suppose that for all $\PL \in \PLSET$, \Eloc is
$(\eps,\PL)$-robust with distortion $(\CONTR,\EXPAN)$ and~\eqref{eq:adaptiveEDPtest} holds.
\asedit{Assume that $\delta< D^{2/\PN}$.}
Then, when the \AdaptiveEDP is run with the
candidate set \PLSET, it will obtain the guarantees of
Theorem~\ref{thm:constant-main} for the optimum pair
$(\PN,\PL)$ among all $\PL \in \PLSET$.
\end{theorem}

%\begin{align}\label{eq:adaptiveEDPtest}
%\kSW\, \min(\Kloc^{\PL-1},\; \tfrac{\PL}{\DIM} \log\log(n))
%    & \leq C_{\DIM} \log(\eps n).
%\end{align}

\begin{proof}
The \AdaptiveEDP picks the pair $(\PN,\PL)$ with optimal
\constDR[\alpha] among all pairs $(\PN,\PL), \PL \in \PLSET$ such that the
pruned graph \PrunedG[\PN,\PL]{\ESW} is connected.
By Lemma~\ref{lm:adaptive-EDP-isolatedNodes}, with high probability,
this is the set of all pairs $(\PN,\PL), \PL\in \PLSET$ such that the
local structure
\Eloc is $(\PN,\PL)$-rich with distortion $(\CONTR,\EXPAN)$.
\end{proof}

\OMIT{ %%%%%%%
First, observe that whenever
$\Eloc \subseteq \PrunedG[\PN,\PL]{\ESW}$,
i.e., the local structure is preserved for some pair $(\PN,\PL)$,
the \AdaptiveEDP obtains the guarantees of
Theorem~\ref{thm:constant-main}.
The concern would thus be that for some pair $(\PN,\PL)$, the local
structure is not preserved, and the EDP-pruning algorithm arrives at
incorrect distance estimates.
The \AdaptiveEDP only uses the best pair
$(\PN,\PL), \PL \in \PLSET$ for which \PrunedG[\PN,\PL]{\ESW} is connected.
Thus, it suffices to prove that whenever
$\Eloc \not \subseteq \PrunedG[\PN,\PL]{\ESW}$,
with high probability, \PrunedG[\PN,\PL]{\ESW} is disconnected.
From now on, focus on one such pair $(\PN,\PL)$, and let
$\PRUNEDG = \PRUNEDG[\PN,\PL]$.
} %%%%%%%%

\begin{proofof}[Proof of Lemma~\ref{lm:adaptive-EDP-isolatedNodes}]
Fix $(\PN,\PL)$ and let $\PRUNEDG = \PRUNEDG[\PN,\PL]$.
Let $I$ be the set of $\eps n$ isolated nodes in $\PrunedG{\Eloc}$.

The high-level idea of the proof is as follows.
For each node $u \in I$ and any edge $(u,v) \in \Eloc$, the local
structure \Eloc alone does not contain \PN edge-disjoint paths of
length at most \PL.
Thus, for $u$ not to be isolated in \PrunedG{\ESW},
a small neighborhood of $u$ would have to be incident on at least one
random edge.
Because there are at least $\epsilon n$ such isolated nodes $u$,
we will be able to show that with high probability, at least one of
them will end up isolated in \PrunedG{\ESW}.
This is not trivial as there is significant dependence between the
isolation events for different nodes; we solve this issue by
considering a sufficiently spread-out subset \NET of $I$ (which
limits the dependence), and then applying Lemma~\ref{lm:probab}
to a carefully designed revelation process.
We now fill in the remaining technical details.

For any set $S \subseteq V$, let $\EE(S)$ denote the event that \ESW
contains no random edges incident on $S$.
We begin by lower-bounding $\Prob{\EE(u)}$ for individual nodes $u$.
Fix $u$ and a distance scale $r$, and let $U_r$ be the set of nodes
$v$ with $\D(u,v) \in (r, 2 r]$.
There are at most $\densityK \cdot (2r)^{\DIM}$ nodes in $U_r$, and for
each node $v\in U_r$, an edge $(u,v)$ is created independently with probability at most
    $q \triangleq \CSW\, \kSW\, r^{-\DIM}$.
Thus, the probability that $u$ has no edges to any nodes in $U_r$ is
at least
\begin{align*}
(1-q)^{|U_r|}
    = \left[ (1-q)^{1/q} \right]^{2^{\DIM}\, \densityK\cdot \CSW\,\kSW}
    \geq 4^{-2^{\DIM}\, \densityK\cdot \CSW\,\kSW}.
\end{align*}
Here, we used the fact that the function $f(q)=(1-q)^{1/q}$ is
decreasing in $q$, so in particular
    $f(q) \geq f(\tfrac12) = \tfrac14$
for any $q\leq \tfrac12$.

The event that $u$ has no random edges is now the intersection of
the events that $u$ has no random edges at scale $r$, with $r$ ranging
over powers of $2$. Thus, $\EE(u)$ is the intersection of $\log(n)$
independent events, each with probability at least
    $4^{-2^{\DIM}\, \densityK\cdot \CSW\,\kSW}$.
Thus, for each node $u$,
\begin{align*}
\Prob{\EE(u)}
\; \geq
\; 4^{-2^{\DIM}\, \densityK\cdot \CSW\,\kSW\, \log n}
\; = \;  n^{-2\cdot 2^{\DIM}\, \densityK\cdot \CSW\,\kSW}.
\end{align*}

For any node $u \in I$, let $V_u$ be the $(\PL-1)$-hop
neighborhood of $u$ in \Eloc. Note that
    $V_u \subseteq \Ball{u}{\EXPAN \PL}$,
so it contains at most $\densityK\,(\EXPAN \PL)^d$ nodes.
We consider events $\EE(V_u)$ that no node in $V_u$ is incident on any
random edges.
The absence of any random edges incident on a subset of nodes
$V'$ can only increase the probability that no random edge is incident
on a given node $u$, as there are fewer remaining candidate edges.
In this sense, the events $\Set{\EE(v)}{v}$ are positively correlated,
and we can bound
\begin{align} \label{eq:adaptiveEDPtest-conditionUsed}
\Prob{\EE(V_u)}
    & = \; \textstyle \Prob{\bigcap_{v \in V_u} \EE(v)}
    \; \geq \textstyle \; \prod_{v \in V_u} \Prob{\EE(v)}  \\
    & \geq \; n^{-2\cdot (2\EXPAN \PL)^{\DIM}\, \densityK^2\cdot \CSW\,\kSW}. \nonumber
\end{align}
By the assumption~\eqref{eq:adaptiveEDPtest}, the above expression is at most
    $p \triangleq n^{-1/3}$.

We claim that whenever $\EE(V_u)$ happens, the node
$u\in I$ is isolated in \PrunedG{\ESW}.
First, note that under the event $\EE(V_u)$, $u$ itself has no
incident random edges. Let $(u,v) \in \Eloc$ be arbitrary.
We show that $(u,v)$ must be pruned.
Because no random edges are incident on $V_u$, no path in
\PrunedG{\ESW} of length at most \PL starting from $u$ can
contain any random edge.
Thus, all $u$-$v$ paths of length at most \PL in \PrunedG{\ESW}
must be entirely in \Eloc.
However, $(u,v) \notin \PrunedG{\Eloc}$, so \Eloc does not
contain \PN edge-disjoint $u$-$v$ paths of length at most \PL.
Hence, $(u,v) \notin \PrunedG{\ESW}$.

It remains to show that with high probability, at least one of
the events $\EE(V_u), u \in I$ happens.
To limit the dependence between the events under consideration, we
focus on a subset $\NET \subseteq I$.
Let $\NET \subseteq I$ be a $2C \PL$-net for $(I, \D)$.%
\footnote{Recall that an
  $r$-net for a metric space $(V,\D)$ is a set of points
  $\NET \subseteq V$ such that (i) any two points in \NET
  are at distance at least $r$ from one another, and (ii) any point in
  $V$ is within distance at most $r$ from some point in \NET. It is
  a well-known fact that such sets exist and can be constructed
  greedily by adding one point at a time.}
Because there are at most $O((C \PL)^{\DIM})$ nodes within distance
$2C \PL$ of any node $u$,
we obtain that $|\NET|\geq \eps n/O((C \PL)^{\DIM})$.
Furthermore, because \Eloc has distortion at most $C$, we get that
$V_u \subseteq \Ball{u}{C(\PL-1)}$, implying that
the neighborhoods $V_u, u \in \NET$ are pairwise disjoint.

The events $\EE(V_u), u \in \NET$ are still not independent, but their
dependence is now more limited, making them amenable to the
technique of Lemma~\ref{lm:probab}.
We define an ordering for revealing the presence (or absence) of edges
$(u,v)$, along with a revelation of the events $\EE(V_u), u \in \NET$.
Fix some ordering $\varphi$ on \NET, and start with $R = \NET$.
$R$ throughout will be a set of candidate nodes $u$ such that the
event $\EE(V_u)$ has not been ruled out.
In step $t = 1, 2, \ldots$, if $R \neq \emptyset$, let $u_t \in R$
be the first remaining node in $R$ according to $\varphi$.
Reveal the presence or absence of all random edges incident on
$V_{u_t}$ which have not been revealed yet.
Whenever a random edge $(v,v')$ is revealed to be present such that
$v \in V_{u_t}, v' \in V_{u'}$ for some $u' \in R$, remove $u'$ from $R$.
(In this case, $\EE(V_{u'})$ clearly cannot happen any more.)
Once $R$ is empty, reveal the presence or absence of all remaining
random edges.
Clearly, this is an equivalent way of revealing the random edge set
\ESW.

Consider a particular step $t$, during which a node $u_t \in R$ is
processed.
If no edges incident on $V_{u_t}$ are revealed, the event
$\EE(V_{u_t})$ has happened, and \PrunedG{\ESW} will be
disconnected.
Conditioned on processing node $u_t$, the event $\EE(V_{u_t})$
happens with probability at least $p = n^{-4/9}$,
as the absence of
some edges incident on $V_{u_t}$ may already have been revealed
earlier, whereas no edges can have been revealed as present.
(Otherwise, $u_t$ would have been removed from $R$.)

Let $N$ be the number of steps $t$ of the revelation process,
and let $X_t$ be the indicator variable of the event $\EE(V_{u_t})$.
Thus, whenever each $V_u, u \in \NET$ has an incident random edge,
we have that $\sum_{t=1}^N X_t = 0$.
It thus suffices to upper-bound the probability that
$\sum_{t=1}^N X_t = 0$, which can be accomplished using the lower
bound of Lemma~\ref{lm:probab} with $x=0$:
\begin{align*}
\textstyle  \Prob{\sum_{t=1}^N X_t \geq 1}
& \geq  (1- (1-p)^t) - \Prob{N < t},
    \quad \mbox{ for all } t \in \N,
\end{align*}
or equivalently,
\begin{align} \label{lm:adaptiveEDPtest-event}
\textstyle \Prob{\sum_{t=1}^N X_t = 0}
& \leq (1-p)^t + \Prob{N < t},
    \quad \mbox{ for all } t \in \N.
\end{align}

We choose $t = \eps \sqrt{n}$. Then,
\begin{align*}
(1-p)^t
\; \leq \; (1-n^{-4/9})^{\eps \sqrt{n}}
\; \leq \; e^{-\eps n^{1/18}},
\end{align*}
so $(1-p)^t$ is exponentially small.
Finally, we bound the probability that $N < \eps \sqrt{n}$.
Consider a step $t$ of the revelation process.
With high probability, each node in $V_{u_t}$ has
at most $O(\log n)$ incident random edges, so that the total number of
random edges incident on $V_u$ is at most $O(\Kloc^{\PL} \log n)$,
\asedit{where $\Kloc$ is the maximal degree in $\Eloc$}.
Thus, with high probability, at most $O(\Kloc^{\PL} \log n)$ other
nodes $u$ can be removed from $R$ in any one step, implying that
the process will take at least
\begin{align*}
\frac{|\NET|}{O(\Kloc^{\PL} \log n)}
\; \geq \; \Omega\left( \frac{\eps n}{(C \PL)^{\DIM} \Kloc^{\PL} \log n} \right)
\; \geq \; \eps \sqrt{n}
\end{align*}
steps, for sufficiently large $n$.
In particular, $N \geq t$ with high probability, completing the proof.
\end{proofof}

\subsubsection{Examples of robust graphs}
\label{sec:adaptive-EDP-examples}

Recall that the toroidal grid is $(1,\PL)$-robust for any \PL.
The grid example extends to graphs that are edge-transitive on a small
scale, in the following sense.

\begin{definition}
%[transitive/near-transitive graphs]
Fix a graph $G$ and a path length \PL.
For any edge $e$, let $H_e$ be the induced subgraph of the
\PL-hop neighborhood of $e$ in $G$.
Two edges $e,e'$ are \emph{locally \PL-equivalent}
if there exists an isomorphism $\phi_{e,e'}: H_e \to H_{e'}$ with
$\phi(e) = e'$. $G$ is called \emph{edge-transitive} at scale \PL if
any two edges are locally \PL-equivalent.
\end{definition}

Notice that the traditional definition of edge-transitive graphs is
obtained when \PL equals the graph's diameter.

\OMIT{ %%%%%%
\begin{enumerate}
\item Two edges $e,e'$ are \emph{locally equivalent}
if there exists an isomorphism $\phi_{e,e'}: H_e \to H_{e'}$ with
$\phi(e) = e'$. A graph $G$ is \emph{\PL-transitive} if any two edges are locally equivalent.

\item An edge $e'$ is \emph{at least as dense as} $e$ if there exists a
subgraph $\hat{H} \subseteq H_{e'}$ and
an isomorphism $\phi_{e,e'}: H_e \to \hat{H}$ with $\phi(e) = e'$.
A graph $G$ is \emph{$(\eps,\PL)$-transitive} if there exists a
set $S$ of at least $\eps n$ nodes with the following property:
for each pair of edges $e,e'$ such that $e$ has at least one endpoint
in $S$, $e'$ is at least as dense as $e$.~\footnote{\ASedit{It follows that any two edges with at least one endpoint in $S$ are locally equivalent.}}
\end{enumerate}
}

Let $G$ be an edge-transitive graph at scale \PL that is a
$(\CONTR,\EXPAN)$-spanner for \D.
It is easy to see that $G$ is $(1,\PL)$-robust with distortion
$(\CONTR,\EXPAN)$.
Indeed, the \PL-hop neighborhood of a given edge $(u,v)$ determines
whether this edge is $(\PN,\PL)$-connected (i.e., whether there exist \PN
edge-disjoint $u$-$v$ paths of at most \PL hops each).
So for every given \PN, either every edge in $G$ is
$(\PN,\PL)$-connected, or every edge in $G$ is not
$(\PN,\PL)$-connected and therefore pruned by the \ItAlg{\PN}{\PL}.

We further generalize this example to graphs $G$ with some short edges
added. Specifically, pick an arbitrary node set $S\subseteq V$ such that
its $(\PL+1)$-neighborhood in $G$ contains at most $1-\eps n$ nodes,
for some $\eps\in (0,1)$. Add arbitrary edges $(u,v)$ such that
$u,v \in S$ and $\D(u,v)\leq \EXPAN$. Note that the resulting graph
$G'$ is also a $(\CONTR,\EXPAN)$-spanner for \D.

We claim that $G'$ is $(\eps,\PL)$-robust with distortion
$(\CONTR,\EXPAN)$. Indeed, if $G$ is $(\PN,\PL)$-connected for some
\PN then $G'$ is $(\PN,\PL)$-rich with distortion $(\CONTR,\EXPAN)$
and connectivity witness $G$.
Otherwise, no edge in $G$ is $(\PN,\PL)$-connected in $G$ alone.
Consider the complement $S'$ of the $(\PL+1)$-neighborhood of $S$.
Any edge $e$ in $G'$ with at least one endpoint in $S'$ is also
present in $G$, and moreover has the same \PL-neighborhood in both
graphs.
It follows that $e$ is not $(\PN,\PL)$-connected in $G'$;
consequently, it is pruned by the \ItAlg{\PN}{\PL}.
Therefore every node in $S'$ is isolated in \PrunedG[\PN,\PL]{G'}.

\section{Conclusions}
\label{sec:conclusions}

We have shown that, under standard assumptions about generative models
for social networks, it is possible to reconstruct social
spaces with small distortion from a multiplex social network;
indeed, it is possible to do so in near-linear time.
The edges do not need to be labeled with their ``origin,''
so long as the categories are ``locally sufficiently uncorrelated.''
Under increasingly stronger assumptions, the distortion can be
improved from constant, to $1+o(1)$, to poly-logarithmic additive error.
While these results rely on having poly-logarithmic node degree, we
also show that small polynomial distortion can be obtained in the
constant-degree regime, so long as the social network contains a
sufficiently rich local structure.
This is possible even if the algorithm only possesses very rudimentary
knowledge about the local structure.

While our results can be interpreted as a proof of concept --- it is
possible in principle to efficiently separate the different
dimensions of social interactions and identify similarities between
individuals --- they set the stage for a number of possible extensions.
\begin{enumerate}
\item There are several specific technical open questions within our
model, the most immediate of which is extending the multi-category
results to the constant-degree regime.
\item We assumed that the algorithm had knowledge of various
input parameters (the number of categories, the number of dimensions,
etc.), whereas ideally, the algorithm should be able to learn these
parameters from input data as well.

\OMIT{
\item We would like to be able to
infer metric spaces beyond near-uniform density point sets.
Dealing with highly non-uniform metrics may be more challenging.}

% Since even routing in small worlds requires a switch to
% rank-based friendship models for such metrics
% \cite{barbella:kachergis:liben-nowell:sallstrom:sowell,kumar:liben-nowell:tomkins},
% it seems likely that the reconstruction task would also have to be
% based on a rank-based model, implying that perhaps not the metric, but
% only the ranks, could be reconstructed.
\item For our multi-category algorithms to work, we required a
``category disjointness''
condition, essentially stating that locally,
metrics look uncorrelated with respect to each other. It seems
unlikely that one could reconstruct metrics if categories were
extremely similar, but it is an interesting open question how much our
current condition could be weakened while still allowing for provable
reconstruction.
In particular, we conjecture that future work will be able to deal
with a few localized violations of the category disjointness
condition, so that they lead to incorrect distance
estimates only for the affected nodes, without propagating to other
parts of the metric space.
\item Our model so far also assumes that the node degrees are
essentially uniform across nodes, which will usually not hold in
practice. A corresponding extension for the single-category case might
not be too difficult, but inferring the individual node degrees for
multiple categories appears more difficult.

\item Finally, and perhaps most importantly, one may want to consider
  ``host spaces'' other than Euclidean space with near uniform density,
  such as ultrametrics, more general ``group structures'' (e.g.,
  \cite{small-world-nips}), or point sets with significantly non-uniform density.
  It would be particularly interesting if an
  algorithm did not need to know the structure of the host space in
  advance, and instead could infer it from the data.
\end{enumerate}

In practice, there will usually be additional information available
beyond the edges. This may include information about nodes' locations,
interests, or demographics (as collected by social networking sites);
partial interaction statistics along the edges;
or perhaps a social network that has been previously embedded
in a social distance space, but is now being extended by a few new nodes.
In either case, it is an interesting question how to formalize the
benefits that can be obtained with such side information.
In particular, time stamps on edges introduce a temporal dimension into
the problem: now, instead of fixed node locations in the social space,
one could ask about nodes' trajectories over time.

% Finally, returning to the motivation of our problem, it would be desirable to combine
% the reconstruction results with algorithms for clustering in metric
% spaces or for ad selection
% %\cite{kleinberg:slivkins:upfal:metric}  A.S. not a good citation for this
% in order to obtain provable guarantees for the concrete problems faced by
% practitioners.

%This question is particularly of interest if the
%algorithm provides further feedback about the metric space, as in the
%case of multi-armed bandit approaches for ad selection.
% A.S.: no published MAB paper does that, I think.

\subsubsection*{Acknowledgments}
We would like to thank Christian Borgs, Jennifer Chayes, Moises
Goldszmidt, Bobby Kleinberg, Jon Kleinberg, Peter Monge, Satish Rao
and Ken Wilbur for useful discussions and pointers, and an
  anonymous referee for detailed and helpful feedback.

\bibliographystyle{siam}

%bibliography{bibliography/bibliography}
%\begin{small}
\bibliography{bibliography/names,bibliography/conferences,bibliography/publications,bibliography/bibliography,bibliography/bib-nodeLabeling}
%\bibliography{bibliography/bib-abbrv,bibliography/bib-random,bibliography/bib-nodeLabeling,bibliography/names,bibliography/conferences,bibliography/publications,bibliography/bibliography,bibliography/bib-Alex}
%\end{small}

%\newpage
%\appendix

%\input{app-random.tex}

%\input{app-extensions.tex}
%\input{app-relatedWork.tex}
%\input{app-amoeba.tex}
%\input{app-additive-proofs}
%\input{app-constDeg.tex}
%\input{app-adaptive.tex}
%\input{appendix.tex}

%\vspace{5mm}
\newpage
\appendix
\asedit{\section{Important notation}}
\label{app:notation}

\asedit{[Note to review team: This appendix is new.]}

\begin{tabular}[h]{l|l}
\multicolumn{2}{l}{{\bf Single category}} \\
$V$, $n$& $V$ is the ground set of $n$ nodes \\
$\ESW$  & the realized social graph \\
$\D(u,v)$    & social distance between nodes $u$ and $v$\\
$\Ball{u}{r}$
    &   $\Set{v}{\D(u,v) \leq r}$ closed ball w.r.t.~$\D$, with center $u$ and radius $r$ \\
$\DIM$     & Euclidean dimension of the metric space\\
$\densityK$ & the constant in the definition of \emph{near-uniform density} \\
$\kSW$ & target degree of nodes\\
$\CSW$ & normalization constant for the edge distribution\\
$f(\cdot)$ & $f(r) = \min(1, \CSW \kSW\, r^{-\DIM})$
              probability of edge of social distance $r$; \\

\multicolumn{2}{l}{{\bf Multiple categories}} \\
$\NUMCAT$     & number of categories \\
$\D[i](u,v)$ & social distance in category $i$ \\
$\Ball[i]{u}{r}$
    &   $\Set{v}{\D(u,v) \leq r}$ closed ball w.r.t.~$\D[i]$, with center $u$ and radius $r$ \\
$\ESW$, $\ESW[i]$
    & $\ESW[i]$ is the realized social graph for category $i$;
     $\ESW = \bigcup_{i=1}^\NUMCAT \ESW[i]$ \\
$\kSW$, $\kSW[i]$
    & $\kSW[i]$ is the target degree for category $i$;
    $\kSW = \frac{1}{\NUMCAT} \cdot \sum_i \kSW[i]$ \vspace{2mm} \\

\multicolumn{2}{l}{{\bf The Amoeba Algorithm}} \\
$\localR$ & $\localR = \Theta((\CSW \kSW)^{1/\DIM})$ the \emph{local radius} \\
$\prunedR$ & $\prunedR = \Theta(\localR \NUMCAT^{2/\DIM})$ the \emph{pruning radius} \\
$\prunedE$ & the \emph{pruned set} (of edges) \\
$\amoebaE[i]$ & the edge set constructed by the \Amoeba algorithm for category $i$ \\
$\Mtwo(u,v)$  & the number of two-hop $u$-$v$ paths in \ESW \vspace{2mm}  \\

\multicolumn{2}{l}{{\bf (Recursive/Extended) Two-Ball Algorithm}} \\
$\Dnorm(u,v)$ & normalized social distance
    $\Dnorm(u,v) = \D(u,v)/ (\CSW\,\kSW)^{1/\DIM}$ \\
$\Dnorm[i](u,v)$ & normalized social distance in category $i$\\
$\NRBall[\Dnorm]{u}{\kappa}$
    & the set of the $\kappa$ closest nodes to $u$ according to metric $\Dnorm$ \\
$\CPD$      & the constant in the definition of \emph{perfectly uniform density} \\
$\DIMCONST$ & the constant in expected \#edges between any two balls
    (see the remark after Theorem~\ref{thm:rec2ball-singleCat}) \\
\NEdges{s}{t} & the number of edges in $\ESW$ between the two balls constructed by the algorithm  \vspace{2mm}  \\

\multicolumn{2}{l}{{\bf Constant target degree}} \\
$\Eloc$ & the local structure: edges deterministically present in $\ESW$ \\
$\Kloc$ & the maximal degree of $\Eloc$ \\
$\constDR[\cdot]$ & the threshold radius in the \ItAlg{\PN}{\PL} (see Equation~\eqref{eq:EDPtest})
\end{tabular}

\end{document}